\newif\ifarxiv
\arxivtrue 

\ifarxiv

\documentclass[11pt,letterpaper]{article}
\usepackage{amsmath,amsthm,nicefrac}
\usepackage{amsfonts, amstext}
\usepackage{bm}
\usepackage{libertine}
\usepackage[libertine]{newtxmath}
\usepackage{booktabs}

\usepackage{subfigure}
\usepackage{hhline}
\usepackage[table,dvipsnames]{xcolor}

\usepackage{typearea}
\typearea{14}
\usepackage[font={small,it}]{caption}

\usepackage{setspace}
\usepackage[compact]{titlesec}

\usepackage[shortlabels]{enumitem}
\usepackage[breaklinks]{hyperref}
\hypersetup{colorlinks=true,%
            citebordercolor={.6 .6 .6},linkbordercolor={.6 .6 .6},%
citecolor=blue,urlcolor=black,linkcolor=blue,pagecolor=black}

\usepackage[nameinlink]{cleveref}
\Crefname{algocf}{Algorithm}{Algorithms}
\crefname{algocfline}{line}{lines}
\Crefname{invariant}{Invariant}{Invariants}
\Crefname{claim}{Claim}{Claims}
\Crefname{subclaim}{Subclaim}{Subclaims}

\usepackage{epsfig}
\usepackage{amsthm}

\usepackage{mathrsfs}
\usepackage{xspace}
\usepackage{soul}
\usepackage{latexsym}
\usepackage{bbm}

\usepackage{accents}

\usepackage{framed}

\definecolor{DarkGray}{rgb}{0.66, 0.66, 0.66}
\definecolor{DarkPowderBlue}{rgb}{0.0, 0.2, 0.6}
\definecolor{fluorescentyellow}{rgb}{0.8, 1.0, 0.0}
\definecolor{cerulean}{rgb}{0.0, 0.48, 0.65}
\definecolor{bleudefrance}{rgb}{0.19, 0.55, 0.91}

\usepackage[ruled,vlined,linesnumbered,algonl]{algorithm2e}
\SetEndCharOfAlgoLine{}
\SetKwComment{Comment}{\footnotesize$\triangleright$\ }{}

\SetCommentSty{mycommfont}

\usepackage{thmtools,thm-restate}

\allowdisplaybreaks

\newtheorem{theorem}{Theorem}[section]
\newtheorem{lemma}[theorem]{Lemma}
\newtheorem{observation}[theorem]{Observation}
\newtheorem{claim}[theorem]{Claim}
\newtheorem{fact}[theorem]{Fact}
\newtheorem{corollary}[theorem]{Corollary}
\newtheorem{assumption}[theorem]{Assumption}

\theoremstyle{definition}
\newtheorem{define}{Definition}

\theoremstyle{remark}

\newcommand{\bF}{{\bf F}}
\newcommand{\bB}{{\bf B}}
\newcommand{\bR}{{\bf R}}
\newcommand{\tup}{\mathcal{T}^{\rm up}}

\newcommand{\bH}{{\bf H}}

\newcommand{\tdown}{\mathcal{T}^{\rm down}}

\newcommand{\calR}{\mathcal{R}}
\newcommand{\Sstar}{\mathcal{S}^{\star}}
\newcommand{\Sin}{\mathcal{S}^{\text{in}}}
\newcommand{\Ustar}{U^{\star}}
\newcommand{\Cstar}{C^{\star}}
\newcommand{\lmax}{{\ell_{\max}}}
\newcommand{\wstar}{w^{\star}}

\newcommand{\tO}{\widetilde{O}}
\newcommand{\tmB}{\widetilde{\mathbf{B}}}
\newcommand{\tbH}{\widetilde{\mathbf{H}}}
\newcommand{\csus}{c_{\rm sus}}
\newcommand{\tausus}{\tau^{\rm sus}}
\newcommand{\deltasus}{\Delta_{\rm sus}}
\newcommand{\cspd}{c_{\rm spd}}
\newcommand{\OPT}{\textsc{OPT}}

\newcommand{\tsus}{t^{\rm sus}}
\newcommand{\tend}{t^{\rm end}}
\newcommand{\tprep}{t^{\rm prep}}
\newcommand{\tcopy}{t^{\rm copy}}
\newcommand{\ttail}{t^{\rm tail}}
\newcommand{\tcomp}{t^{\rm comp}}
\newcommand{\tstart}{t^{\rm start}}

\newcommand{\cov}{{\rm cov}\xspace}
\newcommand{\lev}{{\rm lev}\xspace}
\newcommand{\plev}{{\rm plev}\xspace}
\newcommand{\Tset}{\mathbf{T}_{\rm set}}
\newcommand{\Telem}{\mathbf{T}_{\rm elem}}
\newcommand{\Tunc}{\mathbf{T}_{\rm unc}}
\newcommand{\Tact}{\mathbf{T}_{\rm A}}
\newcommand{\Tpas}{\mathbf{T}_{\rm P}}
\newcommand{\Tdor}{\mathbf{T}_{\rm D}}
\newcommand{\Texp}{\mathbf{T}_{\rm E}}
\newcommand{\Tin}{\mathbf{T}_{\rm in}}

\newcommand{\sayan}[1]{{\color{blue}[\textbf{Sayan}: #1]}}

\newcommand{\alert}[1]{{\color{red}#1}}

\newcounter{mynote}[section]

\newcommand{\eps}{\varepsilon}

\newcommand{\poly}{\operatorname{poly}}

\renewcommand{\emptyset}{\varnothing}

\newcommand{\calS}{\mathcal{S}}

\newcommand{\junk}[1]{}
\newcommand{\eat}[1]{}

\newif\ifhideproofs

\ifhideproofs
\usepackage{environ}
\NewEnviron{hide}{}

\fi

\author{
{Sayan Bhattacharya\thanks{Department of Computer Science, University of Warwick, Coventry, UK. Email: {\tt S.Bhattacharya@warwick.ac.uk}.}}
\and{Ruoxu Cen\thanks{Department of Computer Science, Duke University, Durham, NC. Email: {\tt ruoxu.cen@duke.edu}.}}
\and{Debmalya Panigrahi\thanks{Department of Computer Science, Duke University, Durham, NC. Email: {\tt debmalya@cs.duke.edu}.}}
}

\date{November 2025}

\title{Fully Dynamic Set Cover: Worst-Case Recourse and Update Time}

\fi

\begin{document}

\maketitle

\begin{abstract}
    In (fully) dynamic set cover, the goal is to maintain an approximately optimal solution to a dynamically evolving instance of set cover, where in each step either an element is added to or removed from the instance. The two main desiderata of a dynamic set cover algorithm are to minimize at each time-step,
    \begin{itemize}
        \item[-] the {\em recourse}, which is the number of sets removed from or added to the solution, and 
        \item[-] the {\em update time} to compute the updated solution.
    \end{itemize}    
    This problem has been extensively studied over the last decade leading to many results that achieve ever-improving bounds on the recourse and update time, while maintaining a solution whose cost is comparable to that of offline approximation algorithms.
    
    In this paper, we give the {\bf first} algorithms to simultaneously achieve non-trivial {\bf worst-case bounds} for recourse and update time. Specifically, we give fully-dynamic set cover algorithms that simultaneously achieve $O(\log n)$ recourse and $f\cdot \poly\log(n)$ update time in the worst-case, for both approximation regimes: $O(\log n)$ and $O(f)$ approximation. (Here, $n, f$ respectively denote the maximum number of elements and maximum frequency of an element across all instances.) 
    Prior to our work, all results for this problem either settled for amortized bounds on recourse and update time, or obtained $f\cdot \poly\log(n)$ update time in the worst-case but at the cost of $\Omega(m)$ worst-case recourse. (Here, $m$ denotes the number of sets. Note that {\em any} algorithm has recourse at most $m$.)
\end{abstract}

\thispagestyle{empty}

\newpage
\tableofcontents
\newpage

\setcounter{page}{1}

\part{Extended Abstract}\label{part:extended}

\section{Introduction}\label{sec:intro}
Consider the fundamental problem of computing a minimum {\em set cover}. As input, we receive a universe $U$ of $n$ elements and a collection $\calS \subseteq 2^U$ of $m$ sets defined over $U$, such that $\bigcup_{s \in \calS} s = U$. A {\em set cover} is a collection of sets  $\calS^\star \subseteq \calS$ that covers every element, i.e., $\bigcup_{s \in \calS^\star} s = U$. The goal is to return a set cover of minimum size. Let $f$ denote an upper bound on the maximum {\em frequency} of any element in $U$, where the frequency of a given element is the number of sets in $\calS$ it belongs to. There are two textbook algorithms for this problem: one greedy and the other primal-dual. In $O(fn)$ time, they respectively achieve an approximation ratio of $\ln (n)$ and $f$ (see e.g.,~\cite{WS11}). Both the runtime and approximation guarantees of these two algorithms are tight, under standard complexity theoretic assumptions.\footnote{The size of the input can be $\Omega(fn)$, and hence this $O(fn)$ runtime is asymptotically optimal. Further, it is highly unlikely that we can beat either the $f$ or the $\ln(n)$ approximation guarantee,  even if we allow for any arbitrary polynomial runtime~\cite{DinurS14,KhotR03}.}

We focus on the minimum set cover problem in a {\em dynamic} setting, where the input keeps changing via a sequence of {\em updates}. Each update corresponds to the insertion/deletion of an element. Throughout these updates, we have to maintain an approximately minimum set cover of the current input. It is common to quantify the performance of such a {\em dynamic algorithm} according to three parameters.
\begin{itemize}
\item {\em Approximation ratio}: It indicates the quality of the solution $\calS^\star \subseteq \calS$ maintained by the algorithm.
\item {\em Recourse}: It equals the number of changes (i.e., insertions/deletions of sets) in the maintained solution $\calS^\star$, per update. It measures the {\em stability} of the solution maintained by the algorithm.
\item {\em Update time}: It is the time taken to process an update. It measures the {\em efficiency} of the  algorithm. 
\end{itemize}

The ultimate goal in this setting is  to design a dynamic algorithm that has {\bf worst-case} guarantees, and  matches the performance of the best known static algorithm.\footnote{In other words, we want to recover a near-optimal static algorithm by feeding the universe $U$ as a sequence of $n$ element-insertions, to be processed by the concerned dynamic algorithm. Thus, if the dynamic algorithm has update time $O(\tau)$, then the resulting static algorithm would have a runtime of $O(\tau \cdot n)$.} This  leads us to the following two natural questions.

\begin{itemize}
\item {\bf Q1:} Is there a dynamic set cover algorithm with $O(f)$ approximation ratio, $\tilde{O}(1)$ {\bf worst-case} recourse  and $\tilde{O}(f)$ {\bf worst-case} update time? (Throughout the paper, the $\tilde{O}(.)$ notation hides $\text{polylog}(n)$ factors.)
\item {\bf Q2:} Is there a dynamic set cover algorithm with $O(\log n)$ approximation ratio, $\tilde{O}(1)$ {\bf worst-case} recourse and $\tilde{O}(f)$ {\bf worst-case} update time? 
\end{itemize}

The study of a  fundamental, textbook problem like set cover  occupies a central place in the dynamic algorithms literature. Indeed, as we outline in \Cref{sec:prior:work}, over the past decade an extensive and influential line of work has been devoted towards designing dynamic set cover algorithms, exploring various trade-offs between the three performance measures (approximation ratio, recourse and update time). Surprisingly, however, {\bf all previous dynamic set cover algorithms with non-trivial update times had $\Omega(m)$ worst-case recourse}. Note that this worst-case recourse bound is trivial since {\em any} algorithm has recourse at most $m$. So, while there was substantial progress in dynamic set cover, particularly for amortized bounds, in prior work, it did not offer an answer to {\bf Q1} or {\bf Q2}.

We answer {\em both} {\bf Q1} and {\bf Q2} in the affirmative, by obtaining the first algorithms to  {\em simultaneously} achieve near-optimal worst-case recourse and worst-case update time bounds for dynamic set cover. 

Our results are summarized in the following two theorems.

\begin{theorem}\label{thm:main:f}
    There is a deterministic $O(f)$-approximation dynamic set cover algorithm with  $O(\log n)$ {\bf worst-case} recourse and $O(f \log^{3}(n))$ {\bf worst-case} update time. 
\end{theorem}

\begin{theorem}\label{thm:main:lnn}
    There is a deterministic $O(\log n)$-approximation dynamic set cover algorithm with $O(\log n)$ {\bf worst-case} recourse and $O(f \log^{3}(n))$ {\bf worst-case} update time. 
\end{theorem}

\subsection{Prior Work}
\label{sec:prior:work}

\paragraph{Low Frequency Regime ($O(f)$-approximation).} 
First, note that if $f = 2$, then the dynamic set cover problem is equivalent to  {\em dynamic vertex cover}. Here, we have an input graph $G = (V, E)$, with $n$ nodes and $m$ edges, that is undergoing a sequence of edge insertions/deletions, and we have to maintain an approximately minimum vertex cover at all times. Accordingly, {\bf Q1} can be recast as follows, {\em in the special case when $f = 2$}.

\begin{itemize}
\item {\bf Q1*:} Can we maintain a $O(1)$-approximate minimum vertex cover in a dynamic graph with $\tilde{O}(1)$ {\bf worst-case} recourse and $\tilde{O}(1)$ {\bf worst-case} update time?
\end{itemize}

\begin{table}[h!]
  \centering
  \begin{tabular}{ccccccc}
    \toprule
    Approx & Amortized & {\bf Worst Case} & Amortized & {\bf Worst Case} & Paper \\
    Ratio & Update Time & Update Time & Recourse & Recourse & & \\
    \midrule
    \color{red} $O\left(f^2\right)$ \color{black} &  $O(f \log (m+n))$ & -  & $O(\log (m+n))$ & \color{red}$\Omega(m)$\color{black} & \cite{BhattacharyaHI15} \\
     \midrule 
     \color{red} $O\left(f^3\right)$ \color{black} & $O\left(f^2\right)$ & - & $O(f)$ & \color{red}$\Omega(m)$\color{black} & \cite{BhattacharyaCH17} \\
     \midrule 
     \color{red} $O\left(f^3\right)$ \color{black} & $O\left(f^2\right)$ & - & $O(f)$ & \color{red}$\Omega(m)$\color{black} & \cite{GuptaKKP17} \\
      \midrule
    \color{blue} $(1+\epsilon)f$ \color{black} & $O\left(f^2 \log (n)/\epsilon^5\right)$ & - & $O\left(f \log (n)/\epsilon^5\right)$ & \color{red}$\Omega(m)$\color{black} & \cite{AbboudAGPS19} ** \\
     \midrule
     \color{blue} $(1+\epsilon)f$ \color{black} & $O\left(f \log(n)/\epsilon^2\right)$ & - & $O\left(\log(n)/\epsilon^2\right)$ & \color{red}$\Omega(m)$\color{black} & \cite{BhattacharyaHN19} \\
     \midrule
  \color{blue} $f$ \color{black} &  $O\left(f^2 \right)$ & - & $O(f)$ & \color{red}$\Omega(m)$\color{black} & \cite{AssadiS21}** \\
     \midrule 
     \color{blue} $(1+\epsilon)f$ \color{black} & $O\left(f^2/\epsilon^2\right)$ & - & $O\left(f/\epsilon^2\right)$ & \color{red}$\Omega(m)$\color{black} & \cite{BhattacharyaHNW21} \\
     \midrule
     \color{blue}$(1+\epsilon)f$ \color{black} & $O\left(f \log^2 (n)/\epsilon^3 \right)$ & \color{blue}$O\left(f \log^2 (n)/\epsilon^3 \right)$\color{black} & $O\left(\log^2 (n)/\epsilon^3 \right)$ & \color{red}$\Omega(m)$\color{black} & \cite{BhattacharyaHNW21} \\
     \midrule
     \color{blue} $(1+\epsilon)f$ \color{black} & $O\left( \frac{f\log f}{\epsilon} + \frac{f}{\epsilon^3} \right)$ & - & $O\left( \frac{\log f}{\epsilon} + \frac{1}{\epsilon^3} \right)$ & \color{red}$\Omega(m)$\color{black} & \cite{BukovSZ25} \\
     \midrule
     \color{blue} $(1+\epsilon)f$ \color{black} & $O\left( \frac{f\log^\star f}{\epsilon} + \frac{f}{\epsilon^3} \right)$ & - & $O\left( \frac{\log^\star f}{\epsilon} + \frac{1}{\epsilon^3} \right)$ & \color{red}$\Omega(m)$\color{black} & \cite{BukovSZ25}* \\
     \midrule
     \color{blue} $(1+\epsilon)f$ \color{blue} & $O\left(f \log (n)/\epsilon^2\right)$ & \color{blue} $O\left(f \log (n)/\epsilon^2\right)$ \color{blue} & $O\left(\log (n)/\epsilon^2\right)$ & \color{red}$\Omega(m)$\color{black} & \cite{SolomonUZ24} \\
     \midrule 
     \color{blue}$O(f)$\color{black} & $O\left(f \log^3 (n)\right)$ & \color{blue}$O\left(f \log^3 (n)\right)$\color{black} & $O\left(\log n \right)$ & \color{blue}$O\left(\log n\right)$\color{black} & {\bf our result} \\
    \bottomrule
  \end{tabular}
    \caption{Prior results for dynamic set cover in the low frequency regime. The entries marked with ** (resp.~*) correspond to randomized algorithms that work only against an {\em oblivious adversary} (resp.~{\em adaptive adversary}). All the other entries correspond to deterministic algorithms. Whenever the guarantee specified by {\bf Q1} is attained for a certain parameter (i.e., approximation ratio, worst-case update time or worst-case recourse), the concerned entry is marked in {\textcolor{blue} {blue}}. Otherwise, the concerned entry is marked in {\textcolor{red} {red}}. If an entry for worst-case update time is left blank, then it implies a trivial worst-case update time of $\tilde{O}(fn)$.}
  \label{tab:results:f}
\end{table}

The first nontrivial progress towards answering this question was made by \cite{OnakR10} in 2010. They designed a randomized dynamic algorithm that works against an oblivious adversary, with an approximation ratio of $O(1)$,  an {\bf amortized} recourse of $\tilde{O}(1)$ and an {\bf amortized} update time of $\tilde{O}(1)$.\footnote{We say that a dynamic algorithm has amortized recourse (resp.~amortized update time) of $O(\tau)$ iff it incurs a total recourse (resp.~time) of $O(t \cdot \tau)$  to process any sequence of $t$ updates starting from an empty input instance. Note that a worst-case recourse (resp.~update time) of $O(\tau)$ implies an amortized recourse (resp.~update time) of $O(\tau)$, but not vice versa.} 
This was followed by a long sequence of results on  dynamic vertex cover and its dual problem of dynamic matching~\cite{BaswanaGS11,NeimanS13,BhattacharyaHI15a,BernsteinS16,BhattacharyaHN17,ArarCCSW18,BhattacharyaK19,Wajc20,BernsteinDL21,BhattacharyaKSW23,Behnezhad23,BehnezhadG24,AssadiKK25}. Already in 2017,  \cite{BhattacharyaHN17} answered {\bf Q1*} in the affirmative, by presenting  a deterministic dynamic vertex cover algorithm with $(2+\epsilon)$ approximation ratio, $O(\log n)$ worst-case recourse and $O(\log^3 (n)/\text{poly}(\epsilon))$ worst-case update time.

Beyond the special case of $f = 2$,  however, progress towards resolving {\bf Q1} has been rather limited. \Cref{tab:results:f} summarizes the prior results on dynamic set cover in the low frequency regime.\footnote{Several of the results in \Cref{tab:results:f,tab:results:logn} do not explicitly state an amortized recourse bound. The bounds stated in the tables are the best ones we could infer from the algorithms, but even if better amortized bounds were achievable, it would not affect the contributions of the current article.} This line of work was initiated by \cite{BhattacharyaHI15} in 2015. By 2019, there was a deterministic dynamic set cover algorithm with $(1+\epsilon)f$ approximation ratio, $\tilde{O}(1)$  amortized recourse and $\tilde{O}(1)$  amortized update time \cite{BhattacharyaHN19}. Following up on this, \cite{BhattacharyaHNW21} and \cite{SolomonUZ24} showed how to deamortize the result of \cite{BhattacharyaHN19}, which led to new algorithms with $(1+\epsilon)f$ approximation ratio and $\tilde{O}(f)$ worst-case update time. Unfortunately, both these  latter algorithms \cite{BhattacharyaHNW21,SolomonUZ24} have a worst-case recourse of $\Omega(m)$. This is because each of these algorithms maintains $O(\log n)$ different solutions (i.e., collection of sets) {\em in the background}, while spending $\tilde{O}(f)$ worst-case time per update on each of these background solutions. It is possible, however, that none of these background solutions constitute a valid set cover. Instead, the actual output of the algorithm is a {\em foreground} solution, which is guaranteed to be a valid $(1+\epsilon)f$ approximate set cover of the current input. This foreground solution is obtained via carefully combining all the background solutions, by means of manipulating some pointers. Thus, although each individual background solution incurs $\tilde{O}(1)$ worst-case recourse, because of the pointer switches the foreground solution might incur a recourse as large as $\Omega(m)$ during a single update. In sharp contrast, as summarized in \Cref{thm:main:f},  we {\em simultaneously} achieve $O(f)$ approximation ratio, $\tilde{O}(1)$ worst-case recourse and $\tilde{O}(f)$ worst-case update time, thereby resolving {\bf Q1}.

\begin{table}[h!]
  \centering
  \begin{tabular}{ccccccc}
    \toprule
    Approximation & Amortized & {\bf Worst Case} & Amortized & {\bf Worst Case} & Paper \\
    Ratio & Update Time & Update Time & Recourse & Recourse & & \\
    \midrule
    \color{blue} $O(\log n)$ \color{black} & exponential & \color{red} exponential \color{black} & $O(1)$ & \color{blue} $O(1)$ \color{black} & \cite{GuptaKKP17} \\
    \midrule
    \color{blue} $O(\log n)$ \color{black} &  $O(f \log n)$ & -  & $O(1)$ & \color{red} $\Omega(m)$ \color{red} & \cite{GuptaKKP17} \\
     \midrule 
    \color{blue} $(1+\epsilon) \ln n$ \color{black} & $O\left(f \log (n)/\epsilon^5\right)$ & - & $O\left(1/\epsilon^4\right)$ & \color{red} $\Omega(m)$ \color{black} & \cite{SolomonU23} \\
     \midrule
     \color{blue} $(1+\epsilon)\ln n$ \color{black} & $O\left( f \log (n)/\epsilon^2\right)$ & \color{blue} $O\left( f \log (n)/\epsilon^2\right)$ \color{black} &  $O\left(\log (n)/\epsilon^2\right)$  & \color{red} $\Omega(m)$ \color{black} & \cite{SolomonUZ24} \\
     \midrule 
     \color{blue} $O(\log n)$ \color{black} & $O(f \log^3 (n))$  & \color{blue} $O(f \log^3 (n))$ \color{black} & $O(\log n)$ & \color{blue} $O(\log n)$ \color{black} & {\bf our result} \\
    \bottomrule
  \end{tabular}
    \caption{Prior results for dynamic set cover in the high frequency regime.  All the algorithms in this table are  deterministic. Whenever the guarantee specified by {\bf Q2} is attained for a certain parameter (i.e., approximation ratio, worst-case update time or worst-case recourse), the concerned entry is marked in \color{blue}blue\color{black}. Otherwise, the concerned entry is marked in \color{red}red\color{black}. If an entry for worst-case update time is left blank, then it implies a trivial worst-case update time of $\tilde{O}(fn)$.}
  \label{tab:results:logn}
\end{table}

\paragraph{High Frequency Regime ($O(\log n)$-approximation).} \Cref{tab:results:logn} summarizes the prior results on  $O(\log n)$ approximation algorithms for dynamic set cover. This line of work was initiated by \cite{GuptaKKP17}, who designed a deterministic dynamic algorithm with $O(\log n)$ approximation ratio, $O(1)$ {\bf amortized} recourse and $O(f \log n)$ {\bf amortized} update time. They also developed a dynamic algorithm with $O(\log n)$ approximation ratio and $O(1)$  worst-case recourse, albeit with an exponential update time. In 2024, \cite{SolomonUZ24} presented the first dynamic set cover algorithm that achieves $\tilde{O}(f)$ worst-case update time with $O(\log n)$, or even a $(1+\epsilon)\ln (n)$, approximation ratio. Just as in the low frequency regime (and for exactly the same reason), however, the worst-case recourse of the \cite{SolomonUZ24} algorithm is $\Omega(m)$. In sharp contrast,  as stated in \Cref{thm:main:lnn}, we  {\em simultaneously} achieve good worst-case recourse and update time, thereby resolving {\bf Q2}.

\subsection{Organization of the Rest of the Paper} 

Both our algorithms for \Cref{thm:main:f} and \Cref{thm:main:lnn} utilize a common set of high-level ideas. However, because they maintain different approximation ratios, the technical details are quite different. To highlight the key ideas without going deep into the technical details, we  organize the paper in three parts: \Cref{part:extended} is an extended abstract of our paper that describes the high-level ideas, using our $O(\log n)$-approximation algorithm to instantiate these ideas, \Cref{part:lnn} gives the proof of \Cref{thm:main:lnn} with all technical details, and \Cref{part:f} gives the proof of \Cref{thm:main:f} with all technical details.


Within the extended abstract, we define the problem formally and introduce some preliminary definitions in \Cref{sec:prelim}. Next, in \Cref{sec:overview}, we  highlight the key ideas behind our algorithm, and how they overcome the main obstacles to obtaining  worst-case recourse. We conclude with some future directions in \Cref{sec:closing}.

\eat{

\alert{Should we separately give the (stronger) result for the decremental setting?}

\alert{Sayan's summary of prior results:}

\noindent {\bf Notations:} $n = $ number  of elements, $m = $ number of sets, $f = $ maximum frequency of an element,  $C =$ the ratio between maximum and minimum cost of a set.

\medskip
\noindent {\bf Deterministic $O(\log n)$-approximation.} 
\begin{enumerate}
\item \href{https://arxiv.org/pdf/1611.05646}{[GKKP'17]}~\cite{GuptaKKP17}: $O(\log n)$-approximation in $O(f \log n)$ amortized update time and $O(1)$ amortized recourse. Further, the paper also gives a $O(\log n)$-approximation algorithm with $O(1)$ {\bf worst-case} recourse, but this algorithm has exponential update time.
\item \href{https://arxiv.org/pdf/2312.17625}{[SU'23]}~\cite{SolomonU23}: $(1+\epsilon)\ln n$-approximation in $O(\epsilon^{-5} \cdot f \log n)$ amortized update time, and  $O(\epsilon^{-4} \cdot \min(\log n, \log C))$ amortized recourse.  The authors claim (but do not prove explicitly) that the update time dependence on $\epsilon$ can be reduced to $\epsilon^{-4}$ with some further effort.
\item \href{https://arxiv.org/pdf/2407.06431}{[SUZ'24]}~\cite{SolomonUZ24}: $(1+\epsilon)\ln n$-approximation in $O(f \log n/\epsilon^2)$ {\bf worst-case} update time, {\bf but the worst-case recourse here should be $\Omega(m)$.} {\bf The paper does {\em not} explicitly state its recourse bound.}
\end{enumerate}

\medskip
\noindent {\bf Deterministic $O(f)$-approximation.}
\begin{enumerate}
\item \href{https://arxiv.org/pdf/1604.05337}{[BHI'15]}~\cite{BhattacharyaHI15}: $O(f^2)$-approximation in $O(f \log (m+n))$ amortized update time and recourse. {\bf The paper does {\em not} explicitly state its recourse bound. But the actual recourse should be a factor $f$ smaller than the update time.}
\item \href{https://arxiv.org/pdf/1611.00198}{[BCH'17]}~\cite{BhattacharyaCH17} and \href{https://arxiv.org/pdf/1611.05646}{[GKKP'17]}~\cite{GuptaKKP17}: $O(f^3)$-approximation in $O(f^2)$-amortized update time and recourse. {\bf The papers do {\em not} explicitly state its recourse bound. But the actual recourse should be a factor $f$ smaller than the update time.}
\item \href{https://arxiv.org/pdf/1909.11600}{[BHN'19]}~\cite{BhattacharyaHN19}: $(1+\epsilon)f$-approximation in $O(f \log (Cn)/\epsilon^2)$ amortized update time and recourse. {\bf The papers do {\em not} explicitly state its recourse bound. But the actual recourse should be a factor $f$ smaller than the update time.}
\item \href{https://arxiv.org/pdf/2002.11171}{[BHNW'21]}~\cite{BhattacharyaHNW21}: $(1+\epsilon)f$-approximation in $O((f^2/\epsilon^3) + (f/\epsilon^2) \log C)$ amortized update time and recourse. {\bf The paper does {\em not} explicitly state this recourse bound. But the actual recourse should be a factor $f$ smaller than the update time.} A second algorithm in the paper achieves $(1+\epsilon)f$-approximation in $O(f \log^2 (Cn)/\epsilon^3)$ {\bf worst-case} update time.  {\bf The paper does {\em not} explicitly state the corresponding recourse bound. But the worst-case recourse here should be $\Omega(m)$.}
\item \href{https://arxiv.org/pdf/2308.00793}{[BSZ'25]}~\cite{BukovSZ25}: $(1+\epsilon)f$-approximation in $O\left(\frac{f \log f}{\epsilon} + \frac{f}{\epsilon^3} + \frac{f \log C}{\epsilon^2}\right)$ amortized update time and recourse. A second algorithm in the paper is randomized (but works against an adaptive adversary), and achieves $(1+\epsilon)f$-approximation in $O\left(\frac{f \log^\star f}{\epsilon^2} + \frac{f}{\epsilon^3} + \frac{f \log C}{\epsilon^2}\right)$ expected amortized update time and recourse. {\bf The paper does {\em not} explicitly state its recourse bounds. But the actual recourse bounds should be a factor $f$ smaller than the update time bounds.}
\item \href{https://arxiv.org/pdf/2407.06431}{[SUZ'24]}~\cite{SolomonUZ24}: $(1+\epsilon)f$-approximation in $O(f \log n/\epsilon^2)$ {\bf worst-case} update time. {\bf The paper does {\em not} explicitly state its recourse bound. But the worst-case recourse here should be $\Omega(m)$.}
\end{enumerate}

\medskip
\noindent {\bf Randomized $O(f)$-approximation against an oblivious adversary.}
\begin{enumerate}
\item \href{https://arxiv.org/pdf/1611.05646}{[GKKP'17]}~\cite{GuptaKKP17}: $O(f)$-approximation in $O(1)$ amortized recourse and polynomial update time.
\item \href{https://arxiv.org/pdf/1804.03197}{[AAGPS'19]}~\cite{AbboudAGPS19}: $(1+\epsilon)f$-approximation in $O(f^2 \log n/\epsilon^5)$ expected amortized update time and recourse. It works only in the unweighted setting. {\bf The paper does {\em not} explicitly state its recourse bound. But the actual recourse  should be a factor $f$ smaller than the update time.}
\item \href{https://arxiv.org/pdf/2105.06889}{[AS'21]}~\cite{AssadiS21}: $f$-approximation in $O(f^2)$ expected (and whp) amortized update time and recourse. It works only in the unweighted setting. {\bf The paper does {\em not} explicitly state its recourse bound. But the actual recourse  should be a factor $f$ smaller than the update time.}
\end{enumerate}

}

\section{Preliminaries}\label{sec:prelim}\paragraph{The Fully Dynamic Set Cover Problem.}
In this problem, we are given a (fixed) universe $U$ of elements and a collection of sets $\calS$ that cover the elements in the universe, i.e., $\cup_{s\in \calS} \ s = U$. At any {\em time-step} $t \ge 0$, there is a set of {\em live} elements $L(t)$ that have to be covered by the algorithm. Elements that are not live are called {\em dormant} and denoted by $D(t) := U\setminus L(t)$. Initially, all elements are dormant, i.e. $L(0) = \emptyset$. The set $L(t)$ evolves over time via the following operations:
\begin{itemize}
    \item[-] if $L(t) = L(t-1) \cup \{e\}$ for some $e\notin L(t-1)$,  we say that the element $e$ has been {\em inserted} at time-step $t$, and
    \item[-] if $L(t) = L(t-1) \setminus \{e\}$ for some $e \in L(t-1)$, we say that the element $e$ has been {\em deleted} at time-step $t$.
\end{itemize}
We denote the number of elements to be covered at time-step $t$ by $n(t) := |L(t)|$, and its maximum value over time by $n := \max_t n(t)$.  The {\em frequency} of an element $e\in U$ in $\calS$ is the number of sets containing it $f(e) := |\{s\in \calS: e\in s\}|$, and the maximum frequency over all elements is denoted $f := \max_{e\in U} f(e)$.

Any solution is a collection of sets in $\calS$. At time-step $t$, a solution $S(t) \subseteq \calS$ is said to be {\em feasible} if it covers all the live elements at that time-step, i.e. if $\cup_{s\in S(t)} s \supseteq L(t)$. The {\em cost} of a solution is the number of sets in the solution, $c(S(t)) = |S(t)|$. The goal of the algorithm is to maintain a solution of approximately optimal cost. Let $S^*(t)$ denote an optimal solution at time-step $t$. Then, the algorithm is said to be an $\alpha$-approximation if for all time-steps $t$, we have $c(S(t)) \le \alpha\cdot c(S^*(t))$.

The two important attributes of dynamic set cover algorithms that we study in this paper are {\em recourse} and {\em update time}. Recourse refers to the change in the solution from one time-step to the next. Formally, the recourse at time-step $t$ is defined as $\eta_t := |S(t) \setminus S(t-1)| + |S(t-1) \setminus S(t)|$. We say that an algorithm has $\beta$-recourse if the recourse at all time-steps is at most $\beta$. The update time at time-step $t$ is the running time of the algorithm in this time-step. Similar to recourse, we say that an algorithm has update time $T$ if for all time-steps $t$, the update time at time-step $t$ is at most $T$.

{\em Throughout the rest of the paper, we use the symbol $X(t)$ to denote the status of any object $X$ at time-step $t$.}

\paragraph{Hierarchical Solution.}
Throughout the paper, we consider certain structured set cover solutions that we call {\em hierarchical} solutions. 
\begin{define}\label{def:hierarchical-sol}
    A \emph{hierarchical solution} $S \subseteq \calS$ is a collection of sets such that each set $s\in S$ is associated with a non-empty \emph{coverage} set $\cov(s)\subseteq s$. We denote $\cov(S) := \cup_{s\in S} \cov(s)$ to be the coverage set of the entire solution $S$. The coverage sets in the solution  are required to be disjoint, i.e. $\cov(s) \cap \cov(s') = \emptyset$ for any distinct $s, s' \in S$. In particular, the coverage set of a feasible solution $S(t)$ at time-step $t$ contains all live elements at that time-step, i.e., $\cov(S(t)) = \cup_{s\in S(t)} \cov(s) \supseteq L(t)$. In addition to the live elements, the coverage sets can contain dormant elements as well.
    
    Each set $s$ in a hierarchical solution $S$ is assigned a non-negative integer \emph{level} denoted $\lev(s)$. The level of a set is inherited by all elements in its coverage set, i.e., $\lev(e) := \lev(\cov^{-1}(e))$. (Note that the level of an element is unambiguous because the coverage sets are disjoint for different sets.) If $e$ is not in the coverage set of a solution $S$, then its level is undefined.

    We require the levels of elements in a hierarchical solution satisfy the following invariant. Intuitively, we would like $\lev(s)$ to be approximately $\log |\cov(s)|$, but for some technical reasons, we relax this to only hold in one direction.
    \begin{itemize}
        \item[-] (Level Invariant) $\forall s\in S$, we have $\lev(s) \le \log_{1+\eps}|\cov(s)|$.
    \end{itemize}
    
    A hierarchical solution $S$ also assigns a \emph{passive level} $\plev(e)$ for elements. We require the passive levels to be defined for all elements in $\cov(S)$, but it can also be defined on other elements.
    The passive level $\plev(e)$ of an element $e$ is a dynamic parameter maintained by the algorithm that we will define later. We require the passive levels to satisfy the following invariant that it will always be at least $\lev(e)$ for elements $e\in \cov(S)$.
    \begin{itemize}
        \item[-] (Passive Level Invariant) $\forall e\in \cov(S)$, we have $\lev(e) \le \plev(e)$.
        In addition, if $e\in \cov(S)$ is dormant, then we require $\plev(e) = \lev(e)$.
    \end{itemize}

    \eat{We require that the levels and passive levels of elements in a hierarchical solution satisfy two basic invariants. Intuitively, we would like $\lev(s)$ to be approximately $\log |\cov(s)|$, but for some technical reasons, we relax this to only hold in one direction. This constitutes the first invariant. The passive level $\plev(e)$ of an element $e$ is a dynamic parameter maintained by the algorithm that we will define later, but it will always be at least $\lev(e)$. This constitutes the second invariant. We write the invariants below, where $\eps > 0$ is a small constant that we will fix later:
    \begin{enumerate}
        \item (Level Invariant) $\forall s\in S$, we have $\lev(s) \le \log_{1+\eps}|\cov(s)|$.
        \item (Passive Level Invariant) $\forall e\in \cov(S)$, we have $\lev(e) \le \plev(e)$.
    \end{enumerate}
    }
\end{define}
This completes the definition of a hierarchical solution. The algorithm ensures that the maximum level $\ell_{\max}$ is at most $O(\log n)$.


To disambiguate between different hierarchical solutions, we use the superscript $S$ to indicate that a parameter refers to a specific solution $S$. In particular, $\cov^S(s)$ denotes the set of elements in the coverage set of $s$ in solution $S$, while $\lev^S(e), \plev^S(e)$ respectively denote the level and passive level of an element $e\in \cov(S)$. 

We next introduce some new notation for a hierarchical solution $S$.
\begin{itemize}
    \item (Sub-universe) For a level $k$, the set of live elements at time-step $t$ whose level is at most $k$ is denoted by $L_k^S(t) := \{e\in L(t): \lev^S(e) \le k\}$. In particular, we define $L^S(t) := L^S_{\ell_{\max}}(t) = L(t)\cap \cov(S)$.
    \item (Coverage) The collection of elements in the coverage sets at levels at most $k$ is denoted by $C_k^S := \{e\in \cov(S) : \lev^S(e) \le k\}$.
    \item (Active coverage) Of the elements covered at levels at most $k$, the ones that have passive level larger than $k$ are denoted by $A^S_k := \{e \in C_k^S:  \plev^S(e) > k\}$.
    \item (Passive coverage) The remaining elements covered at levels at most $k$ are denoted $P^S_k := C^S_k \setminus A^S_k$.
\end{itemize} 
Elements in $\cov(S)$ with $\plev(e) = \lev(e)$ cannot appear in $A^S_k$ for any $k$. We say such elements are {\em universally passive} in $S$.

\paragraph{Approximation Bound.} To show that the algorithm has a bounded approximation factor, we define the notions of $\eps$-stable and $\eps$-tidy solutions. We define these next:
\begin{define}
\label{def:stable}
    A solution $S$ is $\eps$-stable at level $k$ iff for every set $s\in\calS$, it satisfies $|A^S_k\cap s| < (1+\eps)^{k+1}$. Furthermore, we say that the solution $S$ is $\eps$-stable iff it is $\eps$-stable at every level $k\in[0, \ell_{\max}]$.
    (Note that the inequality needs to hold for every set in $\calS$ and not just those in the solution $S$.) 
\end{define}
Intuitively, the above condition is a form of dual feasibility that we use in our analysis.
\begin{define}
    A solution $S$ is $\eps$-tidy at a level $k$ if $|P^S_{k}| \le \eps\cdot |A^S_k|$, and $\eps$-dirty at level $k$ otherwise. Furthermore, we say that the solution $S$ is $\eps$-tidy iff it is $\eps$-tidy at every level $k\in[0, \ell_{\max}]$. 
\end{define}
Intuitively, the above condition is a form of approximation of the dual objective that we use in our analysis.

Throughout the rest of the paper, we use the symbol $\OPT(X)$ to denote the size of minimum set cover for the universe of elements $X$.
The next lemma from \cite{SolomonUZ24} asserts that it suffices to show that a solution satisfies the above two properties in order to establish its approximation factor:

\begin{lemma}[Lemma 3.1 of \cite{SolomonUZ24}]\label{lem:approx-factor-full}
Suppose $S$ is a hierarchical solution that is $\eps$-tidy and $\eps$-stable at time-step $t$. Then, $|S| \le (1+O(\eps))\ln n\cdot \OPT(L(t))$. 
\end{lemma}

We remark that the constant $\eps$ is merely for analysis.
We want to call \Cref{lem:approx-factor-full} directly, which depends on the parameter $\eps$ for measuring tidy and stable properties, as well as defining the level invariant of hierarchical solutions.
We simply set $\eps=0.5$ to define the level invariant, and we will establish the properties for some constant $\le \eps$ to obtain $O(\log n)$ approximation. 

In our analysis of the approximation bound, we will establish the $\eps$-stable and $\eps$-tidy properties of our algorithm and then invoke this lemma.
%
%
Overall, our goal is an algorithm that achieves $O(\log n)$ approximation ratio, $O(\log n)$ worst case recourse and $O(f\log^3 n)$ worst case update time. 

\section{Technical Overview}\label{sec:overview}
To highlight the main conceptual ideas behind our approach, throughout this section we don't try to optimize the precise polylogarithmic factors in our recourse and update time bounds. Instead, we focus on achieving $O(\log n)$-approximation, $\tO(1)$ worst-case recourse and $\tO(f)$ worst-case update time. Furthermore, for ease of exposition, we ignore all the low-level details about the data structures maintained by our algorithm.

{\bf Roadmap.} In \Cref{sec:basic:greedy}, we present a basic algorithmic template, which we refer to as ``catch-up greedy'', that serves as a very useful subroutine for designing dynamic set cover algorithms with worst-case update time~\cite{SolomonUZ24}. Subsequently, in \Cref{sec:review}, we present (our interpretation of) the dynamic set cover algorithm of~\cite{SolomonUZ24}, and explain why it achieves $O(\log n)$-approximation ratio and $\tO(f)$ worst-case update time, {\bf but incurs large worst-case recourse}. We omit all the technical proofs (of lemmas, corollaries etc.) in  \Cref{sec:basic:greedy} and \Cref{sec:review}, because they follow, either explicitly or implicitly, from the work of~\cite{SolomonUZ24}.

In \Cref{sec:reduction}, we present our first original insight in this paper, which shows that for our purpose it suffices to design a dynamic set cover algorithm with $O(\log n)$-approximation ratio, $\tO(f)$ worst-case update time and $\tO(1)$ worst-case {\bf insertion recourse} (which measures how many sets get inserted into the maintained solution per time-step). Motivated by this insight, we develop our overall dynamic algorithms in three stages; they are presented in \Cref{sec:take1}, \Cref{sec:take2} and \Cref{sec:take3}, respectively, culminating in \Cref{th:overview:main}.

\subsection{A ``Catch-Up Greedy'' Algorithm}
\label{sec:basic:greedy}

Imagine a scenario where we wish to run a {\em static} algorithm for set cover on an input instance. During the execution of the algorithm, however, the input keeps getting updated (via insertions/deletions of elements) on the fly. We want to ensure the property that when the algorithm terminates, it still returns an $O(\log n)$-approximate set cover solution on the current input (which is different from the original input at the time the algorithm started). Our goal in this section is to explain that the standard greedy algorithm for set cover indeed satisfies this property, provided the {\em speed} at which it runs is sufficiently large compared to the speed at which the input keeps getting updated. We will refer to the greedy algorithm in this setting as {\bf catch-up greedy}. Before elaborating further, we need to set up the stage by introducing the following key notion.

\paragraph{Lazily Updated Hierarchical Solution.}   We say that a subset $X \subseteq L$ of live elements is {\bf slowly evolving} at time-step $t$ iff  $|X(t) \oplus X(t-1)| \leq 1$,\footnote{The symbol $\oplus$ denotes the symmetric difference between two sets.} and  a  hierarchical solution $\bH$ is {\bf contained within} $X$ iff $L^{\bH} \subseteq X$. 

Now, consider a hierarchical solution $\bH$ and a subset of live elements $X$ such that: (i) $X$ is slowly evolving at time-step $t$ and (ii) $\bH$ is contained within $X$ just before time-step $t$. Then, we say that $\bH$ is {\bf lazily updated} at time-step $t$ w.r.t.~$X$ iff  it undergoes the changes described in \Cref{alg:greedy:lazy:0} below.

\begin{algorithm}
\caption{Lazily updating hierarchical solution $\bH$ w.r.t.~$X$ at time-step $t$}
    \eIf{$X(t) = X(t-1)$}{
        $\bH$ remains the same, i.e., $\bH(t) = \bH(t-1)$ \;
    }
    {
       Let $e := X(t) \oplus X(t-1)$, where $\oplus$ denotes the symmetric difference between two sets \;
       \eIf{$e = X(t-1) \setminus X(t)$}{
                $\plev^{\bH}(e) \gets \lev^{\bH}(e)$ \; 
       }
       {        $e = X(t) \setminus X(t-1)$ \; 
                \eIf{there is at least one set in $\bH$ that contains $e$}{
                    Let $s \in \bH$ be the set containing $e$ with largest $\lev^{\bH}(s)$ \;
                    Assign $e$ to $s$, i.e., $e$ is added to $\cov^{\bH}(s)$ (but we don't change $\lev^{\bH}(s)$) \;
                    Set $\plev^{\bH}(e) := \lev^{\bH}(e) := \lev^{\bH}(s)$ \;
                }
                {
                    Let $s$ be any set that contains $e$ \label{new:1} \;
                    Add $s$ to $\bH$ at level $0$ \label{new:2} \;
                    Set $\cov^{\bH}(s) := \{e\}$ \label{new:3} \;
                    Set $\plev^{\bH}(e) := \lev^{\bH}(e) := \lev^{\bH}(s) =0$ \label{new:4} \;
                }
        }
    }
\label{alg:greedy:lazy:0}
\end{algorithm}

\paragraph{Catch Up Greedy.} Suppose that we are given a {\bf starting time-step} $t$, a level $k \in [0, \ell_{\max}]$, and a subset $X \subseteq L$ of live elements that is  slowly evolving from time-step $t$ onward. For our purpose, the set $X$ will always correspond to a subset of  live elements in some hierarchical solution. Accordingly,  we let $\plev^{X}(e)$  denote the passive level of an element $e \in X$ in the concerned hierarchical solution that generates $X$. In addition, we are given access to another hierarchical solution $\bH$ that is initially (i.e., in the beginning of time-step $t$) empty. We refer to $(t, k, X)$ as the {\bf source}, and $\bH$ as the {\bf target} of the following computational task: 

We have to ensure that at a not-too-distant future time-step $\tend$ of our own choosing, the hierarchical solution $\bH(\tend)$ is a feasible and $O(\log n)$-approximate  set cover on $X(\tend)$, and additionally that  $\lev^{\bH}(e) \leq k+1$ for all elements $e \in \cov(\bH)$. Moreover, we can only perform $\tO(f)$ units of computation per time-step, i.e., the worst-case update time of the procedure is $\tO(f)$.

\medskip
\noindent {\bf Remark.} For technical reasons, we need to consider the parameter $k$ as part of our target. In the ensuing discussions, the reader might find it more intuitive to think of the scenario where $k = \ell_{\max}$, in which case we trivially have $\lev^{\bH}(e) \leq k+1$ for all elements $e \in \cov(\bH)$.


\begin{algorithm}
\caption{A catch up greedy algorithm with source $(t, k, X)$ and target $\bH$}
    $p_k\gets k+1$\;
    \While{the uncovered sub-universe $X \setminus \cov(\bH)$ is nonempty}{
    Find the set $s$ that covers the maximum number of uncovered elements in the sub-universe, i.e. 
    $$s = \arg\max_{s\in \calS} |s\cap (X\setminus \cov(\bH))|.$$
    
    Add $s$ to $\bH$, with $\cov^{\bH}(s)\gets s\cap (X \setminus \cov(\bH))$. \;
    $\lev^{\bH}(s) \gets \min\left\{p_k, \lfloor \log_{1+\eps}|\cov(s)|\rfloor\right\}$.\;

        \For{For each $e\in \cov^{\bH}(s)$}
        {$\lev^{\bH}(e)\gets \lev^{\bH}(s)$.\; 
        \If{$\plev^{\bH}(e)$ is not already defined}{  $\plev^{\bH}(e) \leftarrow \max\left\{k+1, \plev^{X}(e)\right\}$ \;
        }
        }
        Update $p_k \leftarrow \min\left\{p_k, \lev^{\bH}(s) \right\}$.
}\label{alg:greedy:catchup}
\end{algorithm}

We wish to run the greedy algorithm for set cover on the input $X$, as specified in \Cref{alg:greedy:catchup}, but the issue is that the input $X$ itself is evolving, albeit slowly, over time. To cope with this challenge, we take a very natural approach: We ensure that $\bH$ is contained within $X$ from time-step $t$ onward (this is trivially true at time-step $t$, when $\bH$ is empty). Moreover, in the beginning of time-step $t$, we initialize a {\bf copy pointer} $p_k \leftarrow k+1$, which would henceforth keep monotonically decreasing over time.  Subsequently, during each time-step $t' \geq t$, we perform the following operations.

\begin{algorithm}
\caption{Updating $\bH$ at time-step $t'$, while running catch-up greedy with source $(t, k, X)$}
    \eIf{$X(t') = X(t'-1)$}{
        $\bH$ remains the same, i.e., $\bH(t') = \bH(t'-1)$ \;
    }
    {
       Let $e := X(t') \oplus X(t'-1)$, where $\oplus$ denotes the symmetric difference between two sets \;
       \eIf{$e = X(t'-1) \setminus X(t')$}{
                \If{$e \in \cov(\bH)$}{
                                   $\plev^{\bH}(e) \gets \lev^{\bH}(e)$ \;  
                }
       }
       {        $e = X(t') \setminus X(t'-1)$ \; 
                \eIf{there is at least one set in $\bH$ that contains $e$}{
                    Let $s \in \bH$ be the set containing $e$ with largest $\lev^{\bH}(s)$ \;
                    Assign $e$ to $s$, i.e., $e$ is added to $\cov^{\bH}(s)$ (but we don't change $\lev^{\bH}(s)$) \;
                    Set $\plev^{\bH}(e) := \lev^{\bH}(e) := \lev^{\bH}(s)$ \;
                }
                {
                   Set $\plev^{\bH}(e) \gets p_k$ \; 
                }
        }
    }
\label{alg:greedy:lazy}
\end{algorithm}

\begin{itemize}
\item First, we  update $\bH$ as per the procedure in \Cref{alg:greedy:lazy}. 
\item Then, we execute the next $\log n$ steps\footnote{All logarithms are base 2 in default.} of computation as specified by \Cref{alg:greedy:catchup}.
This steps only measures the operations to process elements in the inner for loop, which is the bottleneck of \Cref{alg:greedy:catchup}.
\end{itemize}
The algorithm terminates when $\bH$ becomes a feasible hierarchical solution w.r.t.~the input $X$, i.e., when $L^{\bH} = X$.


We now observe a few basic properties.
\begin{itemize}
\item The algorithm covers $\log n$ elements from $X$ per time-step. 
We refer to this as the {\em speed} of the catch up greedy algorithm.
(In contrast, the set $X$ itself changes by at most one element per time-step.)
\item It assigns sets in $\bH$ to levels that decrease monotonically over time. Specifically, suppose that the sets $s_1$ and $s_2$ get added to $\bH$ at time-steps $t_1$ and $t_2$, respectively. If $t_1 < t_2$, then $\lev^{\bH}(s_1) \geq \lev^{\bH}(s_2)$. 
\item It essentially implements the static greedy algorithm on $X$, which picks the sets in decreasing order of their marginal coverages, with the computation being spread out across a sequence of time-steps.
\item The algorithm can be implemented in $\tO(f)$ worst-case update time. The bottleneck is due to maintaining some data structure for number of uncovered elements in each set. In each time-step, the algorithm processes $\tO(1)$ elements, and each element is incident to at most $f$ sets.
\item $\lev^{\bH}(e) \leq k+1$ for all elements $e \in \cov(\bH)$.
\end{itemize}

\medskip
\noindent {\bf Extension of the Target Hierarchical Solution $\bH$.} Consider any time-step  $t'$ 
 during the lifetime of a catch-up greedy algorithm, with source $(t, k, X)$ and target $\bH$, as described above. We use the symbol $\tbH(t')$ to define another hierarchical solution, which we refer to as the {\bf extension of $\bH$} at time-step $t'$, as follows: $\tbH$ is the hierarchical solution obtained by continuing to run \Cref{alg:greedy:catchup} at the current time-step $t'$ (i.e., in a static setting, without any element insertion or deletion) until \Cref{alg:greedy:catchup} completes execution. It is easy to verify that $\bH(t') \subseteq \tbH(t')$ at every time-step $t'$, and that $\bH = \tbH$ when the catch-up greedy algorithm terminates. We will crucially use the notion of such an {\bf extended hierarchical solution $\tbH$} in our analysis.

The lemma below summarizes the key guarantees of the catch up greedy algorithm, and it follows as a direct consequence of the properties summarized above.
{\bf Throughout this technical overview section, we use $\epsilon$ as an arbitrarily small constant.} 

\begin{lemma}
\label{lm:catchup:end}
The catch up greedy algorithm described above terminates at time-step $\tend \le  t   + \frac{2}{\log n} \cdot |X(t)|$. Furthermore, at time-step $t \in \left[ t, \tend\right]$, the following conditions hold.
\begin{enumerate}
\item \label{lm:catchup:end:1} $\tbH$ is $\epsilon$-tidy at every level $k' \in [0, k]$.
\item \label{lm:catchup:end:2} $\tbH$ is $\epsilon$-stable.
\item \label{lm:catchup:end:3} $|\bH|\le |\tbH| = O(\log n) \cdot \text{{\sc OPT}}(X)$, where $\text{{\sc OPT}}(X)$ denotes the size of the minimum set cover on input $X$.
\end{enumerate}
\end{lemma}

\subsection{Achieving Good Worst-Case Update Time: An Overview of \cite{SolomonUZ24} Algorithm}
\label{sec:review}

We now summarize the main insight behind how \cite{SolomonUZ24} achieves $\tO(f)$ worst case update time for maintaining a $O(\log n)$-approximate minimum set cover, albeit with a trivial worst case recourse bound of $O(m)$ where $m$ is the number of sets.  In a bit more detail, the algorithm maintains two different sets of hierarchical  solutions, some of which correspond to different levels. 
The solutions are:
\begin{itemize}
    \item[-] a {\bf foreground} solution $\bF$, and 
    \item[-] a set of $\ell_{\max} +1$ {\bf background} solutions $\bB_k$ for levels $k\in [0,\ell_{\max}]$.
\end{itemize}  

The output of the dynamic algorithm comprises only of the foreground solution, which is feasible for the set of live elements $L(t)$ at each time-step $t$ (i.e., $L^{\bF}(t) = L(t)$). Moreover, the foreground solution $\bF$ gets lazily updated w.r.t.~$L$ in the beginning of every time-step. Accordingly, the subset of live elements $L^{\bF}_k$ is also slowly evolving, for all $k \in [0, \ell_{\max}]$. 

We next focus on the background solution $\bB_k$ for a given level $k \in [0, \ell_{\max}]$. This background solution evolves in {\bf phases}, and is maintained by a subroutine that we refer to as {\bf background thread} $T_k$. Each phase lasts for a sequence of consecutive time-steps. Let us consider a specific phase that starts at (say) time-step $\tstart$. In the beginning of time-step $\tstart$, the hierarchical solution $\bB_k$ is empty. Throughout the duration of the phase, we run a catch up greedy algorithm with source $(t, k, L_k^{\bF})$ and target $\bB_k$. The phase ends (say) at time-step $\tend \geq \tstart$, because of one of the following two reasons.
\begin{enumerate}
\item The catch up greedy algorithm for $\bB_k$ completes execution at time-step $\tend$. We refer to this event as a {\bf normal termination} of the background thread $T_k$. In this event, we first {\bf switch} the background solution $\bB_k$ to the foreground, by removing all sets in $\bF$ at levels $\leq k$ and replacing them with all the sets in the corresponding levels in $\bB_k$. Additionally, $\bB_k$ might have sets at level $k+1$, which are now merged into level $k+1$ in the foreground solution $\bF$. Next, for each level $k' \in [0, k-1]$, we reset $\bB_{k'}$ to be empty and terminate the background thread $T_{k'}$. We say that the thread $T_{k}$ {\bf naturally terminates}, whereas for each $k' \in [0, k-1]$ the thread $T_{k'}$ {\bf abnormally terminates}. From the next time-step $\tend+1$ onward, each of these threads $T_0, \cdots, T_k$ {\bf restarts} a new phase.
\item A higher background thread $T_{k'}$, with $k' \in [k+1, \ell_{\max}]$, normally terminates at time-step $\tend$, and this leads to an abnormal termination of the thread $T_k$.
\end{enumerate}
The above discussion makes it clear that we need a tie-breaking rule, for the scenario when we have a collection $\mathcal{T}$ of multiple threads that complete execution at the same time-step. In such a scenario, we select the thread $T_k \in \mathcal{T}$ with the largest index $k$ for normal termination, which in turn leads to abnormal terminations of all threads in $\{ T_{k'} : k' < k \} \supseteq \mathcal{T} \setminus \{T_k\}$.

Using appropriate pointers and data structures that support merging, switching a background thread to the foreground can be handled in $\tO(1)$ time. {\bf Such a switch, however, might potentially lead to $\Omega(m)$ worst-case recourse}, $m$ being the total number of sets in the input. Furthermore, since each background thread runs in $\tO(f)$ worst-case update time, and there are $\ell_{\max}+1 = O(\log n)$ background threads, we conclude that:

\begin{corollary}
\label{cor:new:udpate:time}
The algorithm described above has an overall worst-case update time of $\tO(f)$.
\end{corollary}

Finally, using \Cref{lm:catchup:end}, we can argue that if a background thread $T_k$ starts a new phase at (say) time-step $t$, then that phase ends (due to either a normal or an abnormal termination) by time-step $t^\star \leq t + \frac{2}{\log n} \cdot \left| L_k^{\bF}(t) \right|$. This observation, along with \Cref{lm:catchup:end}, in turn gives us the following corollary. 
\begin{corollary}
\label{cor:new:thread}
Consider any background thread $T_k$. 
\begin{enumerate}
\item \label{cor:new:1} If $T_k$ enters a new phase at time-step $t$, then that phase lasts for at most $\frac{2}{\log n} \cdot \left| L_k^{\bF}(t) \right|$ time-steps.
\item \label{cor:new:2} The extended hierarchical solution $\tmB_k$ is always $\epsilon$-tidy at every level $k' \in [0, k]$, and $\epsilon$-stable. 
\item \label{cor:new:3} We always have $|\bB_k| = O(\log n) \cdot \text{{\sc OPT}}(L_k^{\bF})$, where $\text{{\sc OPT}}(L_k^{\bF})$ is  the minimum set-cover size on input $L_k^{\bF}$.
\end{enumerate}
\end{corollary}

Now, \Cref{cor:new:2} of \Cref{cor:new:thread} guarantees that immediately after switching the background solution $\bB_k$ to the foreground, the foreground solution $\bF$ is $\epsilon$-tidy at every level $k' \in [0, k]$ and $\epsilon$-stable at every level $k' \in [0, k]$. In contrast, \Cref{cor:new:1} of \Cref{cor:new:thread} guarantees that within every $\frac{2}{\log n} \cdot \left| L_k^{\bF}(t) \right| \leq \epsilon \cdot \left| L_k^{\bF}(t) \right|$ time-steps (the inequality holds because $\epsilon$ is an absolute constant), some background solution $\bB_{k'}$ with $k' \geq k$ is switched to the foreground. Since the interval between any two consecutive switches is so short, the foreground solution $\bF$ only accumulates a small number of passive elements within that interval. In particular, we can show that $\bF$ always remain $2\epsilon$-tidy for every level $k \in [0, \ell_{\max}]$, as it keeps getting lazily updated w.r.t.~the set of live elements $L$. Furthermore, the foreground solutions $\bF$ can {\em not} cease to be $\epsilon$-stable at some level $k$ because of a lazy update. The preceding discussion leads to the corollary below.

\begin{corollary}
\label{cor:new:approx}
The foreground solution $\bF$ is $2\epsilon$-tidy and $\epsilon$-stable at every time-step.
\end{corollary}

The main result in \cite{SolomonUZ24} now follows from  \Cref{cor:new:udpate:time}, \Cref{lem:approx-factor-full} and \Cref{cor:new:approx}.

\begin{theorem}
\label{th:new:prior}
There is a dynamic $O(\log n)$-approximation for set cover with $\tO(f)$ worst case update time.
\end{theorem}

\subsection{A New Reduction from Worst-Case Insertion Recourse}
\label{sec:reduction}

Say that a dynamic set cover algorithm incurs one unit of {\bf insertion recourse} (resp.~{\bf deletion recourse}) whenever a set gets inserted into (resp.~deleted  from) its maintained solution. The overall recourse of the algorithm is the sum of its insertion recourse and deletion recourse. It is easy to verify that a worst-case bound on insertion recourse immediately implies the same bound on deletion recourse, but only in an amortized sense. Our first insight in this paper is the observation that in a black-box manner, we can transform any $O(\log n)$-approximate dynamic (unweighted) set cover algorithm with worst-case $O(\log n)$ insertion recourse into another dynamic algorithm with $O(\log n)$-approximation ratio and $O(\log n)$ worst-case overall recourse. 

This insight is summarized in \Cref{lem:new:recourse-reduction} below. Instead of arguing specifically about our algorithm, the lemma prove a more general property about de-amortizing deletion recourse. Consider any online optimization problem where a solution comprises a set of objects, and the goal is to minimize the number of objects subject to a feasibility constraint that changes online. Let us call such a problem an online unweighted minimization problem. We prove the following for any such problem (for our purpose, set $\alpha = O(\log n)$, $\beta = \tO(1)$ and $\delta = 1$ in the lemma statement). 

\begin{restatable}{lemma}{recourse}\label{lem:new:recourse-reduction}
    Consider an online unweighted minimization problem, where the size of the optimal solution changes by at most $\delta$ in each time-step.
    Suppose there exists an algorithm $\cal A$ that maintains an $\alpha$-approximate solution and has worst-case insertion recourse at most $\beta$, for some parameters $\alpha, \beta\ge 1$. (The worst-case deletion recourse of algorithm $\cal A$ can be arbitrarily large.)
    Then, there exists an alternative algorithm $\cal A'$ that also maintains an $\alpha$-approximate solution and has worst-case recourse (both insertion and deletion) at most $2\beta+\delta\alpha$.
\end{restatable}
\begin{proof}
    We define algorithm $\cal A'$ to simulate $\cal A$, but with the following modification: 
    Whenever $\cal A$ deletes objects from its solution, move those objects to a garbage set instead of deleting them immediately in $\cal A'$. The garbage set is added to the output of $\cal A$ to form the output of $\cal A'$.
    Moreover, $\cal A'$ performs an additional garbage removal operation at the end of the processing in each time-step: delete any $\beta+\delta\alpha$ objects (or all objects if fewer) from the garbage set.
    The recourse bound for $\cal A'$ is straightforward from this description.

    We bound the approximation factor of $\cal A'$ by induction over time. 
    For the base case, note that whenever the garbage set is empty (which is the case at initiation), the solutions of $\cal A$ and $\cal A'$ are identical. Hence, at these time-steps, $\cal A'$ is an $\alpha$-approximation.
    The remaining case is when the garbage set is non-empty at the end of a time-step $t$.
    In this case, $\cal A'$ must have deleted $\beta+\delta\alpha$ objects from the garbage set in time-step $t$. 
    Since $\cal A'$ has insertion recourse at most $\beta$, the solution size at the end of time-step $t$ decreases by at least $\delta\alpha$ compared to that at the end of time-step $t-1$. Combined with the induction hypothesis that the solution was an $\alpha$-approximation at the end of $t-1$, and the assumption that optimal solution size can only change by at most $\delta$ in a time-step, we conclude that the solution is still an $\alpha$-approximation at the end of time-step $t$.
\end{proof}

In the full version (see \Cref{lm:full:reduction}), we show that for our specific dynamic algorithm, the reduction guaranteed by \Cref{lem:new:recourse-reduction} can be implemented while incurring only a $\tO(1)$ factor overhead in the worst-case update time. Accordingly,  henceforth we focus on designing a dynamic set cover algorithm with $O(\log n)$ approximation ratio, $\tO(f)$ worst-case update time and $\tO(1)$  worst-case {\bf insertion recourse}.

\subsection{Our Dynamic Algorithm: Take I}
\label{sec:take1}

Armed with \Cref{lem:new:recourse-reduction}, we propose a very natural modification of the algorithm from \Cref{sec:review}, as described below. As we will soon see, the issue with this approach is that it leads to an approximation ratio of $O(\log^2 n)$.

 The algorithm maintains three different sets of  solutions, some of which correspond to different levels: 
\begin{itemize}
    \item[-] a {\bf foreground} solution $\bF$, and 
    \item[-] a set of $\ell_{\max} +1$ {\bf background} solutions $\bB_k$ for levels $k\in [0,\ell_{\max}]$.
  \item[-] a set of $\ell_{\max} +1$ {\bf buffer} solutions $\bR_k$ for levels $k\in [0,\ell_{\max}]$.
\end{itemize}  
The foreground and the background solutions evolve in exactly the same manner as in \Cref{sec:review}. As opposed to the foreground and background solutions (which are hierarchical solutions), each buffer solution $\bR_k$ is simply a collection of sets. The buffer solutions are used as intermediaries between the background and foreground solutions, and they help ensure good worst-case insertion recourse. 
The output maintained by the algorithm, against which we measure its approximation ratio and recourse, is the union of the foreground solution and all the buffer solutions. 

Specifically, a background thread $T_k$ ensures that $\bB_k = \bR_k$ at every time-step. In other words, within a given phase, whenever the thread $T_k$ adds a set $s$ to $\bB_k$, it copies the same set $s$ to $\bR_k$. Since the thread $T_k$ works at a speed of $\log n$, at most $\log n$ sets get added to $\bR_k$ during any given time-step. Since $\ell_{\max} = \tO(1)$, the worst-case insertion recourse of the maintained solution $\bF \bigcup_{k \in [0, \ell_{\max}]} \bR_k$ is also $\tO(1)$. Crucially, note that whenever a background solution $\bB_k$ is switched into the foreground or a background thread $T_k$ abnormally terminates, this leads to {\em zero} insertion recourse; although the deletion recourse due to such an event can be $\Omega(m)$, where $m$ is the total number of sets.   (Whenever a thread $T_k$ terminates, normally or abnormally, we reset {\em both} its background solution $\bB_k$ and its buffer solution $\bR_k$ to being empty.)

Clearly, the worst-case update time of this algorithm is $\tO(f)$. 

As we alluded to it before, the issue with this approach is that the approximation ratio of the resulting algorithm is $O(\log^2 n)$, instead of $O(\log n)$. This is because  $|\bR_k| = |\bB_k| = O(\log n) \cdot \text{{\sc OPT}}(L_k^{\bF}) = O(\log n) \cdot \text{{\sc OPT}}(L)$ for each buffer solution $\bR_k$ (see \Cref{cor:new:3} of \Cref{cor:new:thread}), and $\left| \bF \right| = O(\log n) \cdot \text{{\sc OPT}}(L)$ (see \Cref{th:new:prior}). Thus, we get $\left| \bF \right| + \sum_{k=0}^{\ell_{\max}} \left| \bR_k \right| = O(\log^2 n) \cdot \text{{\sc OPT}}(L)$. This is a fundamental challenge we need to overcome.

To address this challenge, our first attempt would be to modify the algorithm so as to ensure that at each time-step, only a subset of buffer solutions are nonempty, and the sizes (i.e., the number of sets) of the nonempty buffer solutions form a decreasing geometric series. This, in turn, necessitates splitting up the lifetime of a background thread into three different phases. In \Cref{sec:take2}, we explain our new framework in more detail.

\subsection{Our Dynamic Algorithm: Take II}
\label{sec:take2}

As before, we maintain three different sets of  solutions, some of which correspond to different levels: 
\begin{itemize}
    \item[-] a {\bf foreground} solution $\bF$, and 
    \item[-] a set of $\ell_{\max} +1$ {\bf background} solutions $\bB_k$ for levels $k\in [0,\ell_{\max}]$.
  \item[-] a set of $\ell_{\max} +1$ {\bf buffer} solutions $\bR_k$ for levels $k\in [0,\ell_{\max}]$.
\end{itemize}  
The foreground and background solutions are hierarchical solutions, whereas each buffer solution $\bR_k$ is simply a collection of sets. The output maintained by the algorithm is given by $\bF \bigcup_{k \in [0, \ell_{\max}]} \bR_k$. 

The foreground solution is always feasible for the live elements, i.e., $L^{\bF}(t) = L(t)$ at each time-step $t$. Furthermore, in the beginning of each time-step, we lazily update $\bF$ w.r.t.~$L$. Accordingly, the subset of live elements $L_k^{\bF}$ is also slowly evolving, for all $k \in \left[0, \ell_{\max}\right]$.

Next, consider any level $k \in \left[0, \ell_{\max}\right]$, and focus on the background solution $\bB_k$ maintained by the background thread $T_k$. The lifetime of this thread consists of three phases -- a {\bf computation phase}, a {\bf suspension phase} and a {\bf copy phase} -- one after another. Each phase lasts for a sequence of consecutive time-steps. 

 \paragraph{I: Computation Phase.} Suppose that the thread $T_k$ enters computation phase at (say) time-step $\tcomp$. At the start of time-step $\tcomp$, both the background solution $\bB_k$ and the buffer solution $\bR_k$ are empty. In fact, the buffer solution $\bR_k$ will remain empty throughout the duration of the computation phase. 

Now, the thread $T_k$ runs a catch-up greedy algorithm with source $(\tcomp, k, L_k^{\bF})$ and target $\bB_k$. Let $\tsus-1$ be the time-step at which the catch-up greedy algorithm completes execution. Then, the computation phase also ends at time-step $\tsus-1$, and the thread $T_k$ enters suspension phase from the next time-step $\tsus$. Before entering the suspension phase, the thread $T_k$ takes a snapshot of the solution size $\left|\bB_k\right|$ and denotes it by $\tausus_k$. 

\paragraph{II: Suspension Phase.}  Throughout the duration of the suspension phase, the thread $T_k$ lazily updates $\bB_k$ w.r.t.~$L_k^{\bF}$, and the buffer solution $\bR_k$ continues to remain empty.

 During each time-step, the algorithm runs a {\em scheduler} to determine which (if any) of the threads currently in suspension phase will enter copy phase. The scheduler for time-step $t$ implements  \Cref{alg:new:transition} below.

\begin{algorithm}
\caption{Transitioning background threads from suspension to copy phase}   
    $\tau \gets \frac 12  \cdot \min_{j} \tausus_j$, where the minimum is taken over all the threads currently in copy phase \;
    \For{$k = \ell_{\max}$ down to $0$}
    { \If{$\tausus_k \leq \tau$}{
        The thread $T_k$ will enter copy phase from the beginning of the next time-step $t+1$ \;
        Update $\tau \gets \frac 12 \cdot \tausus_k$\;
    }
    }
\label{alg:new:transition}
\end{algorithm}

 \paragraph{III: Copy Phase.}  Throughout the duration of the copy phase, the thread $T_k$ lazily updates $\bB_k$ w.r.t.~$L_k^{\bF}$. Furthermore, at each time-step during its copy phase, the thread $T_k$ copies $\log n$ sets from $\bB_k$ into the buffer solution $\bR_k$. When all the sets from $\bB_k$ have been copied into $\bR_k$, the thread $T_k$ completes execution. 

 The algorithm now performs the following procedure at the end of each time-step $t$. Let $Q(t) \subseteq \left[0, \ell_{\max}\right]$ denote the collection of indices of all those background threads that have just completed execution, and let $j^\star \gets \max\{Q(t)\}$. We {\bf switch} the background solution $\bB_{j^\star}$ to the foreground. Moreover, for each $k \in [0, j^\star]$, the background solution $\bB_k$ and the buffer solution $\bR_k$ are both reset to being empty at the end of time-step $t$, and the thread $T_k$ restarts its computation phase from the next time-step $t+1$. We say that the thread $T_{j^\star}$ {\bf normally terminates}, and  every thread $k \in [0, j^\star-1]$ {\bf abnormally terminates} at time-step $t$. 

{\bf A key conceptual contribution  in this paper to prove the lemma below}, which bounds for how long can a background thread remain waiting in the suspension phase.


 \begin{lemma}
 \label{lm:new:copy}
Suppose that a background thread $T_k$ enters suspension phase from the beginning of time-step $\tsus_k$. Then, the thread $T_k$ terminates (either naturally or abnormally) at or before time-step $\tsus_k+ o(1) \cdot \tausus_k$.
 \end{lemma}

 \subsubsection{Proof (Sketch) of \Cref{lm:new:copy}}

For the sake of analysis, we consider a modified version of the algorithm, where a background thread $T_k$ {\bf aborts} after waiting in suspension phase for $\frac{8}{\log n} \cdot \tausus_k$ time-steps. Specifically, suppose that the thread $T_k$ enters suspension phase at time-step $\tsus_k$. If $T_k$ is still at suspension phase at (say) time-step $\tend_k = \tsus_k + \frac{8}{\log n} \cdot \tausus_k$, then from the beginning of the next time-step $\tend_k+1$, the thread $T_k$ restarts a computation phase (after resetting  $\bB_k$ to being empty). We will show that a background thread actually never gets aborted in this manner (see \Cref{cl:abort:assume}), and so this modified version of the algorithm is equivalent to the original description. The only reason for sticking with this version is that it makes the proof of \Cref{lm:new:copy}  simpler and more intuitive.

We start by upper bounding the duration of a copy phase in the claim below.

\begin{claim}
\label{cl:new:copy:main}
A background thread $T_k$ remains in copy phase for less than $\frac{2}{\log n} \cdot \tausus_k$ time-steps.
\end{claim}

\begin{proof}(Sketch) Suppose that the thread enters suspension phase and copy phase at the start of time-steps $\tsus_k$ and $\tcopy_k > \tsus_k$, respectively. Since $T_k$ was {\em not} aborted while waiting in the suspension phase, we have 
\begin{equation}
\label{eq:abort}
\tsus_k < \tcopy_k < \tsus_k + \frac{8}{\log n} \cdot \tausus_k.
\end{equation}
Throughout the suspension phase, the background solution $\bB_k$ gets lazily updated w.r.t.~$L_k^{\bF}$, and so the solution size $\left| \bB_k \right|$ grows at most by an additive one per time-step. Thus, at the start of the copy phase, we have 
\begin{equation}
\label{eq:abort:2}
\left| \bB_k(\tcopy_k) \right| \leq \left| \bB_k(\tsus_k) \right| + (\tcopy_k - \tsus_k) < \left(1+\frac{8}{\log n}\right) \cdot \tausus_k.
\end{equation}
Next, throughout the copy phase, the background solution $\bB_k$ continues to get lazily updated w.r.t.~$L_k^{\bF}$, and so the solution size $\left| \bB_k \right|$ continues to grow at most by an additive one per time-step. Simultaneously, the thread $T_k$ keeps copying $\log n$ sets from $\bB_k$ into the buffer solution $\bR_k$ per time-step, and the copy phase ends when all the sets from $\bB_k$ have been copied into $\bR_k$ (or, if the thread $T_k$ gets abnormally terminated before that). Since there are $< \left(1+\frac{8}{\log n}\right) \cdot \tausus_k$ sets in $\bB_k$ at the start of the copy phase (see \Cref{eq:abort:2}),  the copy phase lasts for at most $\frac{\left(1+\frac{8}{\log n}\right) \cdot \tausus_k}{(\log n) - 1} < \frac{2}{\log n} \cdot \tausus_k$ time-steps. This concludes the proof.
\end{proof}

To highlight the main idea behind the proof of \Cref{lm:new:copy}, we now make a simplifying assumption.

\begin{assumption}
\label{assume:simple}
Let $\mathcal{Q}(t)$ be the collection of the indices of the threads that are in either suspension or copy phase at the start of time-step $t$. Then, at every time-step $t$, we have $\tausus_j \geq \tausus_k$ for all $j, k \in \mathcal{Q}(t)$ with $j > k$.  
\end{assumption}

In other words, \Cref{assume:simple} enforces a certain {\em monotonicity} on the $\tausus_k$ values of the relevant threads, guaranteeing that threads corresponding to lower levels have smaller $\tausus_k$ values.  

\paragraph{Remark.} A priori, \Cref{assume:simple} might intuitively seem to be true. After all, a thread $T_k$ is responsible for the elements in the foreground solution that are at levels $\leq k$, which is a subset of the elements under the purview of the next thread $T_{k+1}$. Accordingly, we expect that the sizes of the background solutions $\bB_k$ would be monotonically non-decreasing with the index $k$. Since $\tausus_k$ is  a snapshot of the size of $\bB_k$ at the moment the thread $T_k$ enters suspension phase, this provides an intuitive justification for \Cref{assume:simple}.

Unfortunately, however, \Cref{assume:simple} is not necessarily true. This is because:
\begin{enumerate}
\item For two distinct threads $T_j$ and $T_k$, the values $\tausus_j$ and $\tausus_k$ refer to two different snapshots in time (corresponding to the possibly different time-steps at which $T_j$ and $T_k$ enter suspension phase). 
\item The catch-up greedy algorithms run by $T_j$ and $T_k$, which eventually determine the values $\tausus_j$ and $\tausus_k$, are {\em not} synchronized with each other. 
\end{enumerate}
We provide a complete, formal proof of \Cref{lm:new:copy} in \Cref{sec:lifetime} (see \Cref{lm:wait:time}). In this technical overview section, however, we rely on \Cref{assume:simple} while presenting a proof-sketch of \Cref{lm:new:copy}; since our goal here is only to showcase the high-level intuitions behind our algorithm and analysis.

\begin{claim}
\label{cl:abort:assume} A background thread $T_k$ never gets aborted for waiting too long in the suspension phase.
\end{claim}

\begin{proof}(Sketch)
Suppose that the thread $T_k$ enters suspension phase at time-step $\tsus_k$. Consider any time-step $t \geq \tsus_k$ in the suspension phase, and suppose that the thread $T_k$ fails to enter copy phase during time-step $t$. This happens iff one of the following two (mutually exclusive) scenarios hold.
\begin{enumerate}
\item There is another background thread $T_j$, with $j > k$ and $\tausus_j < 2 \cdot \tausus_k$, that is either in copy phase at the start of time-step $t$, or  \Cref{alg:new:transition} decides that $T_j$ will enter copy phase from the next time-step $t+1$. We say that the thread $T_k$ is {\bf blocked-from-above} by the thread $T_j$ at time-step $t$.
\item The thread $T_k$ is {\em not} blocked-from-above at time-step $t$. However, there is another background thread $T_j$, with $j < k$ and $\tausus_j \leq \tausus_k$ (as per \Cref{assume:simple}), that is in copy phase at the start of time-step $t$. Furthermore, at the start of time-step $t$, no other thread $T_{j'}$ with $j' \in [j+1, k-1]$ is in copy phase. We say that the thread $T_k$ is {\bf blocked-from-below} by the thread $T_j$ at time-step $t$.
\end{enumerate}

Now, we consider two possible cases that might occur if  $T_k$ continues to remain in the suspension phase.

\medskip
\noindent {\em Case 1. The thread $T_k$ is blocked-from-above by some other thread $T_j$ at some time-step $t \in \left[\tsus_k, \frac{4}{\log n} \cdot \tsus_k\right]$.} This means that $j > k$, $\tausus_j < 2 \cdot \tausus_k$, and either $T_j$ is already in copy phase at the start of time-step $t$, or $T_j$ enters copy phase at the start of the next time-step $t+1$. Accordingly, by \Cref{cl:new:copy:main}, the thread $T_j$ terminates by time-step $t + \frac{2}{\log n} \cdot \tausus_j < t + \frac{4}{\log n} \cdot \tausus_k \leq  \tsus_k + \frac{8}{\log n} \cdot \tsus_k$. Since $j > k$, whenever $T_j$ terminates, this leads to an abnormal termination of the thread $T_k$. Thus, we conclude that in this case, the thread $T_k$ is terminated before waiting for $\frac{8}{\log n} \cdot \tausus_k$ time-steps in the suspension phase. In other words, the thread $T_k$ does {\em not} get aborted while waiting in the suspension phase.

\medskip
\noindent {\em Case 2. The thread $T_k$ is {\em not} blocked-from-above at any time-step $t \in \left[\tsus_k, \frac{4}{\log n} \cdot \tsus_k\right]$.} Here, the thread $T_k$ must be  blocked-from-below by some thread $T_j$, with $j < k$ and $\tausus_j \leq \tausus_k$ (see \Cref{assume:simple}), at time-step $\tsus_k$. Now, by \Cref{cl:new:copy:main}, the thread $T_j$ terminates at some time-step $t_j < \tsus_k + \frac{2}{\log n} \cdot \tausus_j \leq \tsus_k + \frac{2}{\log n} \cdot \tausus_k$. This, in turn, leads to all threads $T_{j'}$ with $j' \in [0, j-1]$ getting abnormally terminated at the same time-step $t_j$. 

Note that due to \Cref{assume:simple}, no thread $T_{j''}$ with $j'' > j$ enters the copy phase during the time-interval $[\tsus_k, t_j]$. Moreover, since $t_j+1 \in \left[\tsus_k, \frac{2}{\log n} \cdot \tsus_k\right]$, the thread $T_k$ is neither blocked-from-above nor blocked-from-below at time-step $t_j+1$. Accordingly,  the thread $T_k$ must enter copy phase at time-step $t_j + 1$.  In other words, the thread $T_k$ does {\em not} get aborted while waiting in the suspension phase.
\end{proof}

\paragraph{Proof of \Cref{lm:new:copy}.} Consider any background thread $T_k$ that enters suspension phase at time-step $\tsus_k$. By \Cref{cl:abort:assume}, the thread $T_k$ never gets aborted while waiting in the suspension phase. 
Thus, from the description of our modified version of the algorithm, it follows that $T_k$ either gets abnormally terminated or enters copy phase within the next $\frac{8}{\log n} \cdot \tausus_k$ time-steps. Further, if $T_k$ enters copy phase, then by \Cref{cl:new:copy:main} it remains in copy phase for less than $\frac{2}{\log n} \cdot \tausus_k$ time-steps. Overall, from the preceding discussion we infer that $T_k$ terminates (either naturally or abnormally) at or before time-step $\tsus_k + \frac{8}{\log n} \cdot \tausus_k + \frac{2}{\log n} \cdot \tausus_k = \tsus_k + o(1) \cdot \tausus_k$.

\subsection{Our Dynamic Algorithm: Take III}
\label{sec:take3}

In this section, we outline the final version of our algorithm, and sketch the proof of the following theorem.

\begin{theorem}
\label{th:overview:main}
There is a dynamic set cover algorithm with $O(\log n)$-approximation ratio, $\tO(f)$ worst-case update time and $\tO(1)$ worst-case recourse.
\end{theorem}

We first explain that  \Cref{lm:new:copy} is not sufficient to guarantee an $O(\log n)$-approximation ratio of the algorithm presented in \Cref{sec:take2}. To see why this is the case, define a {\bf critical level $\ell_k$} for the background thread $T_k$, as follows: It is the maximum level $\ell \in \left[0, k \right]$ such that $\left| C_{\ell}^{\bB_k}(\tsus_k) \right| \leq  \frac{1}{2} \cdot \tausus_k$, where $\tausus_k :=  \left| \bB_k(\tsus_k) \right|$. Next, consider some level $k' \ll \ell_k$, where 
$\left| C_{k'}^{\bB_k}(\tsus_k) \right| \ll \tausus_k$.
During each time-step when the thread $T_k$ is in suspension or copy phase, the background solution $\bB_k$ gets lazily updated w.r.t.~$L_k^{\bF}$. Such a lazy update might potentially involve the creation of a universally passive element (because of an external deletion) at level $\leq k'$ in $\bB_k$. If this keeps happening repeatedly, then after very few time-steps an overwhelming majority of the elements at level $\leq k'$ in $\bB_k$ would become universally passive. This would imply that $\bB_k$ is no longer $\epsilon$-tidy at level $k'$ when $T_k$ terminates, which, in turn, would lead to a violation of  \Cref{cor:new:approx}.

To address this issue, we create a new {\bf tail phase} for background threads, which occurs after the copy phase. As in \Cref{sec:take2}, we maintain a foreground solution $\bF$, a set of background solutions $\bB_k$ and buffer solutions $\bR_k$ for each level $k \in \left[0, \ell_{\max}\right]$. The foreground solution $\bF$ keeps getting lazily updated w.r.t.~$L$. 

A background thread $T_k$, which is responsible for maintaining $\bB_k$, now runs in four phases: (i) {\bf computation phase}, (ii) {\bf suspension phase}, (iii) {\bf copy phase} and (iv) {\bf tail phase}. The computation and suspension phases work in exactly the same manner as in \Cref{sec:take2}, whereas the copy and tail phases are different. Accordingly, below we only explain these last two phases for a given background thread $T_k$.

 \paragraph{III: Copy Phase.} At each time-step in the copy phase, the thread $T_k$ copies $\log n$ sets $s \in \bB_k$ with $\lev^{\bB_k}(s) > \ell_k$ into the buffer solution $\bR_k$. In addition, throughout the duration of the copy phase, the thread $T_k$ lazily updates $\bB_k$ w.r.t.~$L_k^{\bF}$. The phase ends when all the sets in $\bB_k$ at levels $> \ell_k$ have been copied into $\bR_k$. Suppose that this happens during time-step $\ttail_k-1$. Then, the tail phase begins from time-step $\ttail_k$. 

 \paragraph{IV: Tail Phase.} Throughout the tail phase, the thread $T_k$ lazily updates $\bB_k$ w.r.t.~$L_k^{\bF}$. Furthermore, in the beginning of time-step $\ttail_k$, it initializes an empty hierarchical solution $\bB^{\rm tail}_k$.  Next, the thread $T_k$ runs a catch-up greedy algorithm with source $\left(\ttail_k, \ell_k, L_{\ell_k}^{\bB_k}\right)$ and target $\bB^{\rm tail}_k$. In addition, whenever the catch-up greedy algorithm adds a new set $s$ into $\bB_k^{\rm tail}$, the thread $T_k$ immediately copies the set $s$ into $\bR_k$ on the fly. The tail phase ends at some time-step $\tend_k$ (say), when the catch-up greedy algorithm completes execution.

 We introduce the notation $\bB_k^\star$ to denote a new hierarchical solution, which is obtained by removing from $\bB_k$ all the sets that are at level $\leq \ell_k$, and adding back to it the  hierarchical solution $\tmB_k^{\rm tail}$ (the extension of $\bB_k^{\rm tail}$; see the discussion just before the statement of \Cref{lm:catchup:end}). 
 It is easy to verify that during the tail phase we always have $\bR_k \subseteq \bB^\star_k$. Furthermore, using appropriate pointer switches, {\em at the termination of the tail phase} we can construct $\bB_k^{\star}$ from  $\bB_k$ and $\bB_k^{\rm tail}$ in $\tO(1)$ time.\footnote{This is because when the tail phase terminates, we have $\tmB_k^{\rm tail} = \bB_k^{\rm tail}$.}

The algorithm performs the following procedure at the end of each time-step $t$. Let $Q(t) \subseteq \left[0, \ell_{\max}\right]$ denote the collection of indices of all those background threads that have just completed execution in their tail phases, and let $j^\star \gets \max\{Q(t)\}$ be the highest such thread. We {\bf switch} the  solution $\bB_{j^\star}^\star$ to the foreground. 
After that, for each $k \in [0, j^\star]$, the background solution $\bB_k$ and the buffer solution $\bR_k$ are both reset to being empty at the end of time-step $t$, the solution $\bB_k^{\rm tail}$ is thrown away, and $T_k$ restarts computation phase from the next time-step $t+1$. We say that the thread $T_{j^\star}$ {\bf normally terminates}, and  every thread $k \in [0, j^\star-1]$ {\bf abnormally terminates} at time-step $t$. 


We now summarize a few basic properties of the dynamic algorithm presented above.

 \begin{lemma}
\label{lm:new:tail:1}
Throughout the  tail phase, the hierarchical solution $\bB_k^\star$ always satisfies the following properties.
\begin{enumerate}
    \item \label{item:lm:new:tail:1:1} It is $2\epsilon$-tidy at every level $k' \in [0, k]$, and also $\epsilon$-stable.
    \item \label{item:lm:new:tail:1:2} $\left| \bB_k^\star \right| = O(\log n) \cdot \text{{\sc OPT}}$, where $\text{{\sc OPT}}$ is the minimum set-cover size on input $L$. 
\end{enumerate}
 \end{lemma}

 \begin{proof}(Sketch)
Suppose that the thread $T_k$ enters suspension phase in the beginning of time-step $\tsus_k$. 

Consider any level $k' \in \left[ \ell_k +1, k \right]$. By definition, we have $\left| C_{k'}^{\bB_k}(\tsus_k) \right| > \frac 12 \cdot \tausus_k$. Using the same argument as in the proof of \Cref{lm:new:copy}, we conclude that the thread $T_k$ remains in suspension, copy or tail phase for at most $o(1) \cdot \tausus_k$ time-steps. During each such time-step, at most one element in $\bB_k$ can become passive at level $
\leq k'$ (as $\bB_k$ gets lazily updated). Thus, at every time-step $t$ in the copy or tail phase, we have 
\begin{equation}
\label{eq:new:101}
\left| P_{k'}^{\bB_k}(t) \right| \leq \left| P_{k'}^{\bB_k}(\tsus_k) \right| + o(1) \cdot \tausus_k < \left| P_{k'}^{\bB_k}(\tsus_k) \right| + o(1) \cdot \left| C_{k'}^{\bB_k}(\tsus_k) \right|
\end{equation}
Next, by the same argument as in \Cref{cor:new:2} of \Cref{cor:new:thread}, the background solution $\bB_k$ is $\epsilon$-tidy at level $k'$ at time-step $\tsus_k$. Combining this observation with \Cref{eq:new:101}, we get that 
\begin{equation}
\label{eq:999}
\bB_k \text{ is } 2\epsilon\text{-tidy at  level } k'   \text{ throughout the copy and tail phases.}
\end{equation} 
We claim without proof that replacing the sets at level $\le \ell_k$ by  $\tmB^{\rm tail}_k$ does not increase dirty for level $k' > \ell_k$. Then, we also have $\bB^\star_k$ is  $2\epsilon$-tidy at level $k'$ throughout the tail phase. 

Next, consider any level $k' \in \left[0, \ell_k\right]$. Recall that during the tail phase the thread $T_k$ runs catch-up greedy with source $\left(\ttail_k, \ell_k, L_{\ell_k}^{\bB_k} \right)$ and target $\bB_k^{\rm tail}$. Thus, by \Cref{lm:catchup:end:1} of \Cref{lm:catchup:end}, we get that throughout the tail phase $\tmB^{\rm tail}_k$ is $\epsilon$-tidy at level $k'$ .  Since $k' \leq \ell_k$, this also implies that $\bB^\star_k$ is  $2\epsilon$-tidy at level $k'$ throughout the tail phase.

Next, using similar ideas as in \Cref{cor:new:2} of \Cref{cor:new:thread} (resp.~\Cref{lm:catchup:end:2} of \Cref{lm:catchup:end}), we infer that $\bB^\star_k$ is $\epsilon$-stable at every level $k' \in \left[ \ell_k+1, k\right]$ (resp.~$k' \in \left[0, \ell_k\right]$) throughout the tail phase. 

Finally, using similar ideas as in \Cref{cor:new:3} of \Cref{cor:new:thread}, we can show that throughout the tail phase, the number of sets at levels $> \ell_k$ in $\bB_k$ is at most $O(\log n) \cdot \text{{\sc OPT}}$. In contrast, using similar ideas as in \Cref{lm:catchup:end:3} of \Cref{lm:catchup:end}, we can show that throughout the tail phase we have $\left| \tmB^{\rm tail}_k\right| = O(\log n) \cdot \text{{\sc OPT}}$. Taken together, these two observations imply that $\left| \bB_k^\star \right| = O(\log n) \cdot \text{{\sc OPT}}$ throughout the tail phase. 
 \end{proof}

\begin{corollary}
\label{cor:new:tail:1}
At every time-step, the foreground solution $\bF$ is $3\epsilon$-tidy and $\epsilon$-stable.
\end{corollary}

\begin{proof}(Sketch)
When a thread $T_k$ normally terminates, the way we switch the hierarchical solution $\bB^\star_k$ to the foreground and abnormally terminate all the threads $T_{j}$ with $j < k$ is analogous to the corresponding procedure in \Cref{sec:review} (where we switch $\bB_k$ to the foreground, as opposed to $\bB^\star_k$). Accordingly, following the same line of reasoning that led us to \Cref{cor:new:approx} starting from \Cref{cor:new:thread}, we can derive \Cref{cor:new:tail:1} starting from \Cref{lm:new:tail:1}.
\end{proof}

 \paragraph{Remark.} In the lemma below, we  need to refer to the hierarchical solution $\bB_k^\star$ in the {\em copy phase} of the thread $T_k$. Towards this end, we follow the convention that throughout the copy phase, the hierarchical solution $\bB_k^{\rm tail}$ is empty, and the hierarchical solution $\bB_k^\star$ is constructed in the same manner as in the tail phase. In other words, during the copy phase, the hierarchical solution $\bB_k^\star$ is obtained by removing from $\bB_k$ all the sets that are at level $\leq \ell_k$, and adding back a static greedy solution for $L^{\bB_k}_{\ell_k}$. 
 We claim that with this convention, \Cref{item:lm:new:tail:1:2} of \Cref{lm:new:tail:1} continues to hold even during the copy phase.

  \begin{lemma}
  \label{lm:new:500}
At any time-step $t$ during the  copy or tail phase for thread $T_k$, we have $\tausus_k = O\left(\left| \bB^{\star}_k(t) \right|\right)$.
 \end{lemma}
\begin{proof}
    
    Suppose that $T_k$ enters suspension phase at time-step $\tsus$.

    We partition $\bB_k(\tsus)$ into sets at levels higher than $\ell_k$, and sets at levels at most $\ell_k$.
    Note that the sets in the former part are included in $\bB^\star_k$, and cannot be modified throughout the copy and tail phase. So, its size is at most $\bB^{\star}_k(t)$. Also, the size of the latter part can be bounded by its coverage size, since every set has nonempty coverage.
    So, we have
    \[\tausus_k = |\bB_k(\tsus)| 
    \le |\bB^{\star}_k(t)| + \left|C^{\bB_k}_{\ell_k}(\tsus)\right|
    \le |\bB^{\star}_k(t)| + 0.5\tausus_k.
    \]
    The proof follows by rearranging the terms in the above inequality.
    \end{proof}

  \begin{lemma}
\label{lm:new:tail:2}
Throughout the duration of copy and tail phases for thread $T_k$, we have $\left| \bR_k \right| = O(\tausus_k)$.
 \end{lemma}

 \begin{proof}(Sketch)
 Suppose that the thread $T_k$ enters suspension phase at time-step $\tsus_k$. By definition, at the start of time-step $\tsus_k$, we have $\left| \bB_k \right| = \tausus_k$. Using the same argument as in the proof of \Cref{lm:new:copy}, we conclude that the thread $T_k$ remains in suspension, copy or tail phase for at most $o(1) \cdot \tausus_k$ time-steps. During each such time-step, the size of $\bB_k$ increases by at most one (due to a potential lazy update w.r.t.~$L_k^{\bF}$). Thus, we infer that $\left| \bB_k(t) \right| \leq \left| \bB_k(\tsus_k) \right| + o(1) \cdot \tausus_k = O(\tausus_k)$ at every time-step $t$ during the copy and tail phases of $T_k$.

 Next, observe that during each time-step the thread $T_k$ remains in suspension, copy or tail phase, the size of $C_{\ell_k}^{\bB_k}$ increases by at most one (due to a potential lazy update of $\bB_k$ w.r.t.~$L_k^{\bF}$). Recall that $\left| C_{\ell_k}^{\bB_k}(\tsus_k) \right| \leq  \tausus_k$ by definition. Furthermore, the thread $T_k$ remains in suspension, copy or tail phase for at most $o(1) \cdot \tausus_k$ time-steps. Thus, at every time-step $t$ in the copy or tail phase of thread $T_k$, we have 
 $$\left|\tmB^{\rm tail}_k(t) \right| \leq |L_{\ell_k}^{\bB_k}(\tsus_k)| + (t-\tsus_k)  \leq  \left| C_{\ell_k}^{\bB_k}(\tsus_k) \right|  + o(1) \cdot \tausus_k = O(\tausus_k),$$
 where the first inequality holds because the number of sets in $\bB^{\rm tail}_k(t)$ cannot larger than the number of elements they are meant to cover.
 
 Since  $\bR_k \subseteq \bB_k \cup \tmB^{\rm tail}_k$, we  derive that $\left| \bR_k \right| \leq \left| \bB_k \right| + \left| \tmB^{\rm tail}_k \right| = O(\tausus_k)$ throughout copy and tail phases of $T_k$.  \end{proof}

 \begin{corollary}
\label{cor:new:tail:3}
At every time-step $t$, we have $\sum_{k \in [0, \ell_{\max}]} \left| \bR_k \right| = O(\log n) \cdot \text{{\sc OPT}}(L)$.
\end{corollary}

\begin{proof}(Sketch)
Note that every background thread $T_k$ has $\bR_k = \emptyset$ until the start of its copy phase. 

At the start of a time-step $t$, let $Z(t) \subseteq \left[0, \ell_{\max}\right]$ denote the collection of indices of those background threads that are in copy or tail phases. Now, the scheduler in \Cref{alg:new:transition} ensures that at every time-step $t$, the values in the collection $\{ \tausus_k : k \in Z(t) \}$ are geometrically decreasing. Let $k(t) := \arg\max_{k \in Z(t)} \{ \tausus_k \}$. Applying   \Cref{lm:new:tail:2}, \Cref{lm:new:500} and \Cref{item:lm:new:tail:1:2} of \Cref{lm:new:tail:1} for both copy and tail phases (see the remark just before \Cref{lm:new:500}), we now infer that   
\begin{eqnarray*}
\sum_{k \in [0, \ell_{\max}]} \left| \bR_k(t) \right| = \sum_{k \in Z(t)} \left| \bR_k(t) \right| = \sum_{k \in [0, \ell_{\max}]} O(\tausus_k) = O\left( \tausus_{k(t)}\right) = O\left(\left| \bB_{k(t)}^{\star}(t) \right| \right) = O(\log n) \cdot \textsc{OPT}(L(t)).
\end{eqnarray*}
This concludes the proof of the corollary.
\end{proof}

\paragraph{Proof of \Cref{th:overview:main}.} Our dynamic algorithm ensures the invariant that $\bF$ is a feasible set cover for the collection of all live elements. Hence, the maintained solution $\bF \bigcup_{k \in [0, \ell_{\max}]} \bR_k$ is also a feasible set cover for the collection of all live elements. The size of the maintained solution is $|\bF| + \sum_{k \in [0, \ell_{\max}]} \left| \bR_k \right|$, and so the approximation guarantee of our algorithm follows from \Cref{cor:new:tail:1}, \Cref{lem:approx-factor-full} and \Cref{cor:new:tail:3}.

Using standard data structures, the foreground solution, the background solutions and the buffer solutions can all be maintained in $\tO(f)$ worst-case update time. Furthermore, we have already emphasized that using appropriate pointer switches, we can construct the hierarchical solution $\bB_k^\star$ from $\bB_k$ and $\bB_k^{\rm tail}$ in $\tO(1)$ time, and we can also switch $\bB_k^\star$ into the foreground in $\tO(1)$ time, whenever necessary. Thus, the overall worst-case update time of our algorithm is also $\tO(f)$.

Finally, we observe that during each time-step within the copy and tail phases of the thread $T_k$, at most $\tO(1)$ sets get added to the buffer solution $\bR_k$. In contrast, the buffer solution $\bR_k$ remains empty throughout the computation and suspension phases of the thread $T_k$. Furthermore, whenever a  thread $T_k$  terminates, the buffer solution $\bR_k$ is reset to being empty and (provided $T_k$ normally terminates) we switch $\bB^\star_k$ to the foreground; but this event does not cause any insertion recourse (as opposed to causing potentially massive amount of deletion recourse). Finally, the foreground solution $\bF$ keeps getting lazily updated w.r.t.~$L$ at every time-step, and this also leads to a worst-case insertion recourse of at most $1$. Thus, the total worst-case insertion recourse of the maintained solution $\bF \bigcup_{k \in [0, \ell_{\max}]} \bR_k$ is still $\tO(1)$. Now, applying the reduction outlined in \Cref{sec:reduction}, we obtain a dynamic set cover algorithm with $O(\log n)$-approximation, $\tO(f)$ worst-case update time and $\tO(1)$ worst-case recourse. This concludes the proof of \Cref{th:overview:main}.

\section{Closing Remarks}\label{sec:closing}In this paper, we presented the first fully dynamic set cover algorithms to achieve non-trivial worst-case guarantees on both recourse and update time. In particular, we {\em simultaneously} obtained $\poly\log(n)$ worst-case recourse and $O(f \cdot \poly\log(n))$ worst-case update time, while guaranteeing  $O(\log n)$ or $O(f)$ approximations. 

A natural next goal is to refine these bounds further to match the best results known in the amortized setting, e.g., $(1+\eps)\ln n$ and $(1+\epsilon)f$ approximations with $\tO(1)$ worst-case recourse and $\tO(f\cdot \poly(1/\eps))$ worst-case update time. It seems plausible that our framework can be used to reduce the approximation factors to $(2+\eps) \ln n$ and $(2+\eps) f$ respectively. However, the gap of $2$ is an inherent barrier. Recall that to bound recourse, we used a combination of two solutions, the foreground solution and the buffer solutions, in the output. Each solution, at best, is a tight approximation that adds a factor of $\ln n$ or $f$ to the approximation ratio. Overcoming this barrier of $2$ will require more new ideas that go beyond the algorithmic framework that we introduced in this paper.

We also remark that our results only hold in the unweighted setting, i.e., when the sets have uniform cost. It is natural to ask whether our algorithmic framework extends to the weighted case. At a conceptual level, certain ideas that we introduced in this paper seem agnostic to sets having weights, while other ideas seem to fail. For instance, the idea of using a buffer solution to limit worst-case recourse is independent of whether the sets are weighted or unweighted. On the other hand, consider the idea of suspending and later restarting (or aborting) background processes to only allow a set of well-separated solutions to proceed to the buffer, that we used to bound our approximation ratio. This mechanism does not extend to the weighted case, because we can no longer trade off computation speed (which depends on the number of sets) with solution cost (which depends on the costs of sets) in the way that we did in our current algorithm. Nevertheless, the history of dynamic set cover algorithms makes us hopeful that our results will eventually be extended to the weighted setting. A case in point is the first $(1+\eps)f$-approximation result obtained by \cite{AbboudAGPS19}. The algorithm presented in this paper also applied only to the unweighted setting, but was later extended (with a mild dependence on the range of weights) to the weighted case by \cite{BhattacharyaHN19}. In contrast to this pair of results, the $f$-approximation algorithm obtained by \cite{AssadiS21} also applies to the unweighted setting only, but an extension to the weighted case has remained elusive.

Finally, we remark on the logarithmic bounds in update time and recourse that we obtain in this paper. In the low frequency regime (i.e., $O(f)$-approximation), there are a series of results in prior work that removed the $O(\log n)$ dependence in the {\em amortized} update time bounds~\cite{AssadiS21,BhattacharyaHNW21,BukovSZ25}. So, it is natural to ask whether we can obtain a dynamic set cover algorithm with $O(f)$ worst-case update time and $O(1)$ worst-case recourse. We note that neither of these results is known even in isolation (the only exception is \cite{GuptaKKP17} which obtains $O(1)$ worst-case recourse but at the cost of {\em exponential} update time), and designing algorithms that achieve either bound represents an interesting direction for future research.

\newpage

\part{$O(\log n)$-Competitive Algorithm}
\label{part:lnn}

\paragraph{Remark.} Throughout \Cref{part:lnn}, we use the notations and terminologies introduced in \Cref{sec:prelim}.

\section{Fully Dynamic $O(\log n)$-Competitive Set Cover Algorithm}\label{sec:description}The algorithm maintains a set of different set cover solutions, some of which correspond to different levels. 
The solutions are:
\begin{itemize}
    \item[-] a {\bf foreground} solution $\bF$, 
    \item[-] a set of $\ell_{\max} +1$ {\bf background} solutions $\bB_k$ for levels $k\in [0,\ell_{\max}]$, and 
    \item[-] a set of $\ell_{\max} +1$ {\bf buffer} solutions $\bR_k$ for levels $k\in [0,\ell_{\max}]$.
\end{itemize}    
We maintain the invariant that the foreground solution is feasible for the set of live elements at any time. The background solutions are not part of the output, but are used to update the foreground solution to maintain the desired approximation bounds. However, doing this directly incurs large recourse on the foreground solution. Instead, we use the buffer solution as an intermediary between the background and foreground solutions. We slowly update a buffer solution using the corresponding background solution, and once this process completes, switch the background solution to the foreground solution. The output comprises the foreground and the buffer solutions; so, switching the background solution (which is identical to the buffer solution) to the foreground does not cost us in recourse. 

The foreground solution and each of the background solutions are hierarchical solutions (with levels and passive levels), while the buffer solutions are merely collections of sets.
The algorithm creates $\ell_{\max}+1$ background threads $T_k$, one for each level $k\in [0,\ell_{\max}]$, where each thread runs a variant of the greedy set cover algorithm that we will describe later. 
These background threads run in a sequential order from $T_{\ell_{\max}}$ to $T_0$, and update their respective background solutions independently. The background threads run after the foreground thread. The algorithm periodically copies the background solutions into the buffer solutions. When copying is complete for a background solution, the algorithm switches the background solution to the foreground to update the foreground solution. 


\paragraph{Initialization.}
The algorithm initializes the foreground, background, and buffer solutions to be empty since there is no live element at this stage, i.e. $L(0)=\emptyset$. 

\subsection{Foreground Thread}
We describe the algorithm to maintain the foreground solution on the deletion or insertion of an element.

On deletion of element $e$, we set $\plev^{\bF}(e) \gets \lev^{\bF}(e)$. Note that $e\in \cov(\bF)$ as $\bF$ is a feasible solution for all live elements. 

On insertion of element $e$, let $s$ be the set containing $e$ with largest $\lev^{\bF}(s)$. 
If such a set doesn't exist, let $s$ be any set that contains $e$ -- in this case, we add $s$ to $\bF$ at level 0. In both cases, we assign $e$ to $s$, i.e. $e$ is added to $\cov(s)$ (but we don't change $\lev(s)$). We also set $\plev(e) = \lev(e) = \lev(s)$.

Besides the above changes triggered by the deletion or insertion of an element, the foreground solution can also be modified by switching to a background solution, which happens at the termination of a background thread. We describe background threads next and then discuss how the background, buffer, and foreground solutions interact.

\subsection{Background Threads}
At initialization, the algorithm starts a background thread $T_k$ for every level $k\in [0,\ell_{\max}]$.
Whenever a background thread terminates, the algorithm re-starts a thread at the same level. So, for every level $k$, there is always a background thread $T_k$ running in the algorithm. 

Conceptually, each background thread aims to run a greedy algorithm and copy the solution to the buffer.
Indeed, if we allow an $O(\log^2 n)$ approximation factor, we can simply let all threads copy their solutions, as described in \Cref{sec:take1}. In this case, all threads can run independently in parallel. However, to obtain an approximation factor of $O(\log n)$, we have to restrict the threads that are allowed to copy to the buffer at any time. This causes some additional complexity to coordinate the behavior of different background threads.

We partition the algorithm run by a background thread into five phases. For a background thread $T_k$, we first give a short overview of the purpose of each of these five phases:
\begin{itemize}
    \item {\bf Preparation:} The background thread initializes a greedy set cover algorithm with the set of elements $L^{\bF}_k$, i.e. the sub-universe of live elements at levels $\le k$ in the foreground solution $\bF$.
    \item {\bf Computation:} The background thread runs a greedy set cover algorithm on the sub-universe above to generate a hierarchical background solution $\bB_k$. 
    This background solution covers most elements in the sub-universe, but might leave a few elements uncovered. These latter elements are said to form the tail-universe of this background thread.
    \item {\bf Suspension:} At the end of the computation phase, the background thread might enter a suspension phase where it suspends operation.
    \item {\bf Copy:} Once out of the suspension phase, the background thread copies the background solution it computed to a corresponding buffer solution.
    \item {\bf Tail:} In this final phase, the background thread resumes the greedy algorithm on the (small) tail-universe and immediately copies each set in the solution to the corresponding buffer solution.
\end{itemize}

Next, we describe the algorithm for each phase in detail. In terms of data structures, the background thread $T_k$ maintains a priority queue $Q_k$ over all sets $\calS$, where the key of a set $s$ is the size of its intersection with the uncovered sub-universe $L^{\bF}_{k} \setminus \cov(\bB_k)$. We defer the details of the data structure to \Cref{sec:update-time}.

We also note that the sub-universe $L^{\bF}_k$ can change during the execution of a background thread $T_k$. 
This happens when an element $e$ in the sub-universe is deleted, or an element $e$ is inserted and the foreground thread sets $\lev^{\bF}(e) \le k$.
We describe how to handle these changes separately in each phase. 

Let $\cspd$ be a large constant that we determine later in \Cref{sec:lifetime}. It measures the number of elements processed in each time-step.

Specially, if $|L^\bF_k|\le \cspd$ at the beginning of $T_k$, all work of $T_k$ can be finished in one time-step. In this case, we execute $T_k$ to termination in the time-step. We call this the rule of base threads.

\paragraph{Preparation Phase.}
During the preparation phase, $T_k$ builds the priority queue $Q_k$. This requires scanning the sub-universe $L^{\bF}_{k}$, which is done at the speed of processing $\cspd$ elements in each time-step.

During the scanning, $T_k$ also sets $\plev^{\bB_k}(e) \gets \max\{k+1, \plev^{\bF}(e)\}$ for each element in the sub-universe. (The levels $\lev^{\bB_k}(e)$ of elements $e\in L^{\bF}_{k}$ will be decided later.) 

If the sub-universe $L^{\bF}_k$ changes during the preparation phase, we modify $Q_k$ to reflect this change in the sub-universe $L^{\bF}_k$ and reset the passive level of the element if required.

\paragraph{Computation Phase.}
The background thread $T_k$ executes the greedy algorithm given in \Cref{alg:greedy} in its computation phase, which is an adaptaion of the standard offline greedy algorithm for set cover. A crucial property of the offline algorithm is that the sets are added to the solution in order of monotonically decreasing coverage. But, this property might be violated in our setting because of element insertions. To assert this property artificially, the algorithm uses an index $p_k$ to indicate the level of the last set computed in the greedy algorithm. The algorithm ensures that $p_k$ monotonically decreases over time. We call $p_k$ the {\em copy pointer} of the background thread $T_k$.

In terms of the speed of processing, the greedy algorithm adds $\cspd$ elements to the coverage of the solution per time-step. Since the algorithm is processing elements much faster than insertions and deletions, the elements that change status (from live to dormant or vice-versa) during the execution of the greedy algorithm are a small fraction of all the elements in the greedy solution. This latter property is crucial in ensuring that the algorithm produces a competitive solution.

\begin{algorithm}
\caption{Greedy algorithm run by background thread $T_k$}
    $p_k\gets k+1$\;
    \While{the uncovered sub-universe $L^{\bF}_{k}\setminus \cov(\bB_{k})$ is nonempty}{
    Find the set $s$ that covers the maximum number of uncovered elements in the sub-universe 
    (using $Q_k.$FindMax), i.e. 
    $$s = \arg\max_{s\in \calS} |s\cap (L^{\bF}_{k}\setminus \cov(\bB_{k}))|.$$
    
    Add $s$ to $\bB_k$. Set $\cov(s)\gets s\cap (L^{\bF}_{k}\setminus \cov(\bB_k))$.\;\label{line:greedy-cov-assign}
    
    Set $\lev(s) \gets \min\{\lfloor \log_{1+\eps}|\cov(s)|\rfloor, p_k\}$.\;\label{line:set-level}

        For each $e\in \cov(s)$, set $\lev^{\bB_k}(e)\gets \lev(s)$, and update $Q_k$ to reflect that $e$ is removed from the uncovered sub-universe.\;
    
     Update $p_k \gets \min\{p_k, \lev^{\bB_k}(s)\}$.\;\label{line:pointer-update}
}\label{alg:greedy}
\end{algorithm}

We pause the greedy solution once the number of uncovered elements in the sub-universe becomes small. Intuitively, this is because at this stage, the changes in the sub-universe can significantly affect this small uncovered set.
These uncovered elements form the tail-universe which is handled later in the tail phase of the background thread. In particular,
\begin{itemize}
    \item After processing a set $s$ (i.e.\ when \Cref{alg:greedy} reaches the end of the while loop), if $|L^{\bF}_{k}\setminus \cov(\bB_{k})| \le |\bB_k|$, then pause the greedy algorithm and enter the suspension phase.
\end{itemize}
Specially, if $|L^{\bF}_{k}\setminus \cov(\bB_{k})| \le |\bB_k|\le \cspd$ when we pause the greedy algorithm, then all remaining work of $T_k$ can be finished in one time-step. In this case, we continue to execute the algorithm for copy and tail phases (to be defined later) until the termination of $T_k$. In this case, the thread is viewed to terminate in computation phase, and not viewed to enter later phases.
We call this the rule of shortcut threads. This implies that if $T_k$ enters later phases, then $|\bB_k(\tsus)| > \cspd$ where $\tsus$ is the time-step when $T_k$ enters suspension phase.

The sub-universe $L^{\bF}_k$ can change due to insertions and deletions of elements during the computation phase. We give the algorithm to handle these updates below in \Cref{alg:deletion,alg:insertion}.

\begin{algorithm}
\caption{Updates by background thread $T_k$ on deletion of element $e$ in $L^{\bF}_k$}
    \eIf{$e\in\cov(s)$ for some set $s\in \bB_k$}{
        $\plev^{\bB_k}(e)\gets \lev^{\bB_k}(e)$\;
    }{
        Update $Q_k$ to reflect that $e$ is removed from the uncovered sub-universe.
    }
\label{alg:deletion}    
\end{algorithm}
\begin{algorithm}
\caption{Updates by background thread $T_k$ on insertion of element $e$ in $L^{\bF}_k$ (i.e. $\lev^{\bF}(e) \le k$)}
    \eIf{$e$ is contained in some set in $\bB_k$}{
        Cover $e$ using the set $s$ with the maximum level in $\bB_k$ among all those that contain $e$. \; 
        $\plev^{\bB_k}(e) \gets \lev(s)$\;  
        $\lev^{\bB_k}(e) \gets \lev(s)$\;
    }{
        Update $Q_k$ to reflect that $e$ is inserted into the uncovered sub-universe.\;
        $\plev^{\bB_k}(e) \gets p_k$\;
    }
\label{alg:insertion}    
\end{algorithm}

\paragraph{Suspension Phase.}
When $|L^{\bF}_{k}\setminus \cov(\bB_{k})| \le  |\bB_k|$, the background thread $T_k$ enters the suspension phase.
In the suspension phase, the thread does not run the greedy algorithm, and does not copy any sets to the buffer solution.
When $T_k$ enters suspension phase, we take a snapshot of its solution size and denote it by $\tausus_k$.
In other words, $\tausus_k := |\bB_k(\tsus)|$ where $\tsus$ is the time-step when $T_k$ enters suspension phase.

If the sub-universe $L^{\bF}_k$ changes during the suspension phase, we update $\bB_k$ in an identical way to that in the computation phase (\Cref{alg:deletion,alg:insertion}). 

\bigskip

We now describe the criteria for suspended threads to enter the copy phase. The main desiderata is that the set of threads simultaneously in copy and tail phases form a geometric series in terms of their solution sizes. We give the algorithm for ensuring this below:


\begin{algorithm}
\caption{Transitioning background threads from suspension to copy phase}
    $\tau \gets \frac 12  \cdot \min_j \tausus_j$, where the minimum is taken over for all threads $T_j$ that are in copy and tail phases.\;\label{line:tau-init}
    \While{there exists some $T_{k'}$ currently in suspension phase such that $\tausus_{k'} \le \tau$}{
        $k \gets \max \{k': T_{k'} \text{ in suspension phase and } \tausus_{k'} \le \tau\}$\;
            
    Move $T_k$ from suspension to copy phase.\; 
     Update $\tau \gets \frac 12 \cdot \tausus_k$.\;
}
\label{alg:transition}
\end{algorithm}

Let $\tsus$ denote the time-step when a background thread $T_k$ enters the suspension phase. 
If $T_k$ remains too long in the suspension phase, then the previously computed solution $\bB_k(\tsus)$ can become uncompetitive due to many insertions and deletions. So, we add a thresholding condition: if $T_k$ is in the suspension until time-step $\tsus + 0.1 \cdot\tausus_k$, 
then we abort the thread and restart the preparation phase.

\paragraph{Copy Phase.}

During the copy phase, $T_k$ copies sets in $\bB_k$ to the corresponding buffer solution $\bR_k$ at the speed of $\cspd$ sets per time-step.

If the sub-universe $L^{\bF}_k$ changes during the copy phase, we update $\bB_k$ identically to the computation phase (\Cref{alg:deletion,alg:insertion}). In addition, any sets added to $\bB_k$ (because of element insertions) are immediately reflected in the buffer solution $\bR_k$. 

\paragraph{Tail Phase.}
When entering the suspension phase, the elements in the tail-universe are still to be covered. These elements remain uncovered during the suspension and copy phases. After the solution $\bB_k$ has been copied to $\bR_k$, the background thread $T_k$ enters the tail phase. In this phase, the algorithm resumes the greedy algorithm on the tail-universe.
The speed to run the greedy algorithm is $\cspd$ elements per each time-step, same as the computation phase. 
In the tail phase, $T_k$ copies each new set added to $\bB_k$ immediately to $\bR_k$. 

If the sub-universe $L^{\bF}_k$ changes during the tail phase, we handle this identically to the copy phase. I.e., we run  \Cref{alg:deletion,alg:insertion} and any sets added to $\bB_k$ (because of element insertions) are immediately reflected in the buffer solution $\bR_k$. 

\paragraph{Termination.}
The tail phase ends when $\bB_k$ (and therefore $\bR_k$) covers every element in the sub-universe $L^{\bF}_k$. Now, the algorithm switches the background solution $\bB_k$ to the foreground by removing all sets in $\bF$ at levels $\le k$, and replacing them with the sets in the corresponding levels in $\bB_k$. Additionally, $\bB_k$ might have sets at level $k+1$, which are now merged into level $k+1$ in the foreground solution $\bF$. 

We will show later that these transformations can be done efficiently within the update time afforded by a single time-step of the dynamic instance. Note that the foreground solution $\bF$ and the buffer solutions $\bR_k$ collectively constitute the output. Since the buffer solution $\bR_k$ is a copy of the background solution $\bB_k$ at this stage, the switch does not change the collection of sets in the output. Therefore, recourse is unaffected. But, this step is important in correctly maintaining our invariants on the data structures of the foreground solution. As part of the switch, the levels and passive levels of the relevant elements in $\bF$ are updated to those in $\bB_k$. 

After this switch for a background thread $T_k$, we abort all background threads $T_{k'}$, $k'<k$. When a thread gets aborted, it discards all sets copied to $\bR_{k'}$ and discards all data structures. Then, the threads $T_{k'}$, $k'\le k$ return to the preparation phase (on the new sub-universe $L^\bF_{k'}$). To distinguish between the two situations causing a background thread to end -- that it completed its tail phase and it got aborted by another background thread -- we call the former {\em normal termination} and the latter an {\em abnormal termination}.
Termination by the rule of base threads or shortcut threads is also viewed as normal termination.

\subsection{Sequential Behavior of Threads}
In the above algorithm description, we consider each thread separately. Since our overall algorithm is a sequential algorithm, the threads need to coordinate in a global scheduling. 
We use the following simple scheduling: The foreground thread is executed first. Then, the background threads are executed in the order of decreasing levels, from $T_{\ell_{\max}}$ to $T_0$. Each thread runs in a consecutive time block.

This order is related to the following features of the algorithm.
\begin{itemize}
    \item[-] The foreground thread handles the insertion or deletion in the foreground solution, which determines the change of uncovered sub-universes of background threads.
    \item[-] Although \Cref{alg:transition} is described as a global loop, the condition to enter copy phase from suspension phase is actually examined by each thread separately. Before executing any background threads, we initialize $\tau$ according to Line \ref{line:tau-init}. When executing background threads $T_k$ in suspension phase, we compare its snapshot size $\tausus_k$ with the thread $\tau$. If $\tausus_k \le \tau$, we let $T_k$ enter copy phase and update $\tau\gets \frac 12 \cdot \tausus_k$. The decreasing order of executing background threads ensures an equivalent behavior to \Cref{alg:transition}.
\end{itemize}
We explain some related details below.
\begin{itemize}
    \item[-] When a background thread is executed, it first handles the insertion or deletion, then continue the algorithm for its current phase.
    \item[-] (Rule of one phase in a time-step.) A background thread is only allowed to undergo one phase in a time-step, with the special case of a phase change in the end. So, once a background thread exists one phase and enters the next phase, it will stop execution for the current time-step, and start the next phase at the subsequent time-step.
    \item[-] Recall the normal termination of a background thread will abort all lower-level background threads. Since the background threads are executed in decreasing order, whenever we encounter a normal termination of a background level $T_k$, the lower-level background threads are not executed yet, and they can be directly aborted. Also, according to the rule of one phase in a time-step, these lower-level background threads will start preparation phase in the subsequent time-step.
\end{itemize}

\section{Analysis of the Fully Dynamic $O(\log n)$-Competitive Set Cover Algorithm}\label{sec:analysis}\subsection{Properties of the solution}

Recall that $\cov(\bF)$ and $C^{\bF}_k$ respectively denote the set of elements (live and dormant) in the coverage of all levels and levels $\le k$ of the foreground solution $\bF$. Thus, $\cov(\bF) \setminus C^{\bF}_k$ represents the set of elements in the coverage of levels $> k$. 

\paragraph{Feasibility.}
We start by showing that the sets of live elements in $\cov(\bB_k)$ and $L^\bF_k$ are identical at normal termination of a background thread $T_k$ (i.e. if it terminates at the end of its tail phase and not because it got aborted by another background thread):

\begin{claim}\label{fact:background-feasible}
We have the following properties:
\begin{enumerate}[(i)]
    \item For any time $t$ and any background thread $T_k$, we have
    $\cov(\bB_k(t))\cap L(t) \subseteq L^{\bF}_{k}(t)$.
    \item When a background thread $T_k$ terminates normally, we have $\cov(\bB_k)\cap L =  L^{\bF}_{k}$.
    \item When a background thread $T_k$ terminates normally, the update to the foreground solution does not change $L^{\bF}_j$ at any level $j > k$. 
\end{enumerate}
\end{claim}
\begin{proof}
    Fix any $k$. It is immediate that property (i) implies property (ii) since at normal termination of $T_k$, all live elements in $L^\bF_k$ are in $\cov(\bB_k)$.
%
    Moreover, when a background thread $T_k$ terminates normally, the foreground solution $\bF$ is updated only at levels $\le k+1$. This means that property (iii) is an immediate consequence of property (ii).

    
    We are left to prove property (i), which we do by induction on level $k$ and time $t$.
    As the base case, when any thread $T_k$ is created, $\cov(\bB_k)$ is empty, and the statement is trivial.
    We now consider the inductive steps:
    \begin{itemize}
    \item[-]
    The greedy algorithm run by $T_k$ adds a new set $s$ to $\bB_k$. Since Line \ref{line:greedy-cov-assign} assigns $\cov(s)$ to be a subset of $L^{\bF}_k$, the property is preserved.
    \item[-]
    Consider insertion of an element $e$. If $\lev^{\bF}(e) > k$, then it affects neither $\cov(\bB_k)$ nor $L^{\bF}_k$.
    Otherwise, $e$ is added to $L^{\bF}_k$, and may or may not be added to $\cov(\bB_k)$ depending on whether it is already covered by a set in $\bB_k$. So, the property is preserved.
    \item[-]
    Consider deletion of an element $e$. In this case, $e$ is removed from $L$, and $\cov(\bB_k)$ remains unchanged. So, the property is preserved.
    \item[-]
    Finally, consider the case that $\bF$ is updated at normal termination of another background thread $T_{k'}$. We have $k'< k$, otherwise $T_k$ would be aborted.
    By the inductive hypothesis, (iii) holds for $k'$, which means that $L^{\bF}_k$ is not changed by the update. So, the property is preserved. \qedhere
    \end{itemize}
\end{proof}

Using the claim above, we now show that the foreground solution is always feasible, i.e. all elements in $L$ are always covered in $\bF$:

\begin{lemma}\label{lem:feasible}
    $\bF$ is always a feasible set cover.
\end{lemma}
\begin{proof}
    We prove this lemma by induction over time. Initially, $L = \emptyset$, so the lemma trivially holds.
    We now consider the inductive steps:
    \begin{itemize}
        \item[-] 
        On insertion of an element $e$, the algorithm adds $e$ to $\cov(s)$ for some existing set $s\in \bF$ or adds a new set $s$ to $F$ and sets $\cov(s) = \{e\}$.
        \item[-]
        On deletion of an element $e$, the algorithm does not change $\cov(\bF)$.
        \item[-]
        Finally, when a thread $T_k$ terminates normally, the set of live elements covered in $\cov(\bF)$ does not change, by property (ii) of \Cref{fact:background-feasible}. \qedhere
    \end{itemize}
\end{proof}

\paragraph{Hierarchical solution.}
Next, we establish that the foreground and background solutions are hierarchical solutions, and do not have any duplicate elements. First, we show that the set of elements in the coverage of a background thread $T_k$ is disjoint from $\cov(\bF) \setminus C^{\bF}_k$. This ensures that when we switch the background solution to the foreground, we do not create duplicate elements.

\begin{claim}\label{fact:cov-subset-foreground-coverage}
    For any background thread $T_k$, we have $\cov(\bB_k) \cap (\cov(\bF) \setminus C^{\bF}_k) = \emptyset$.
\end{claim}
\begin{proof}
    We prove by induction on time. The base case is trivial since a background thread $T_k$ is initialized in the preparation phase with an empty solution, i.e., $\cov(\bB_k)$ is empty.
    For the inductive step, we have the following cases:
    \begin{enumerate}
        \item The greedy algorithm run by $T_k$ adds a new set $s$ to $\bB_k$. Since Line \ref{line:greedy-cov-assign} assigns $\cov(s)$ to be a subset of $L^{\bF}_k\subseteq C^{\bF}_k$, the property is preserved.
        \item A new element $e$ is inserted in $L^{\bF}_k$, and therefore also in $C^{\bF}_k$. \Cref{alg:insertion} adds $e$ to $\cov(\bB_k)$ if an existing set in $\bB_k$ contains $e$, else $e$ is added to the uncovered set of elements. In the latter case, $\cov(\bB_k)$ remains unchanged. Therefore, the property is preserved.
        \item An element $e$ is deleted from $L^{\bF}_k$. The foreground thread only modifies $\plev^{\bF}(e)$, so $e$ is retained in $C^{\bF}_k$ as a dormant element. Therefore, the property is trivially preserved. 
        \item A background thread at some level $k' < k$ normally terminates and updates the foreground solution $\bF$ at levels $\le k'+1$. 
        Since levels $> k$ are not changed in $\bF$, the set $\cov(\bF) \setminus C^{\bF}_k$ cannot grow because of this update to $\bF$. Therefore, the property is preserved.
        \qedhere
    \end{enumerate}
\end{proof}

Next, we show a property of the copy pointer that is crucial for ensuring that the foreground and background solutions are hierarchical solution:

\begin{claim}\label{fact:pointer-monotone}
    During the lifetime of a background thread $T_k$, the following hold for its  copy pointer $p_k$:
    \begin{itemize}
        \item[-] $p_k$ is monotone decreasing, starting from $k+1$.
        \item[-] for every element $e\in L^{\bF}_k$, we have $\plev^{\bB_k}(e) \ge p_k$.
    \end{itemize}        
\end{claim}
\begin{proof}
    Since $p_k$ is initialized to $k+1$, and only modified by Line \ref{line:pointer-update} in \Cref{alg:greedy}, the first part of the claim follows.
    
    For the second part, note that in the preparation phase of $T_k$, passive levels of all elements $e\in L^{\bF}_k$ are set to $k+1$. For an element inserted into $L^{\bF}_k$, its passive level is set to $p_k$, which is monotone decreasing.
    The passive levels in a background solution can only be modified due to deletion, which moves $e$ out of $L^{\bF}_k$.
\end{proof}

We are now ready to show that the foreground and background solutions are hierarchical solutions:

\begin{lemma}\label{fact:alg-sol-hierarchical}
    The foreground solution $\bF$ and the background solutions $\bB_k$ for every $k$ are hierarchical solutions as defined in \Cref{def:hierarchical-sol}.
    In particular, the coverages are non-empty and disjoint, and the levels and passive levels satisfy the level and passive level invariants (as given in \Cref{sec:prelim}).
\end{lemma}
\begin{proof}
    We verify the properties of a hierarchical solution:
    \begin{itemize}
        \item[-] Coverages are non-empty: When a set is added to a background solution (in the computation or tail phase) by the greedy algorithm, its coverage set is non-empty. Furthermore, coverage sets do not shrink over time because elements that become dormant due to deletion are not deleted from their respective coverage sets.

        In the foreground solution, sets are added in one of two ways: either they are obtained from a background solution after a normal termination, or they are added due to insertion of an element. In either case, the coverage sets are non-empty.
        \item[-] Coverages are disjoint: First, we consider background solutions. When a set is added by the greedy algorithm, its coverage is disjoint from all sets that are already in the background solution. Moreover, if a set gains coverage due to an element insertion, the newly inserted element is added to the coverage of only one set in the solution. Hence, the coverages of all sets in the background solution remain disjoint throughout.
        
        The coverages of sets in the foreground solution change over time in two ways. First, this may be due to an element insertion. In this case, the new element is added to the coverage of exactly one set in the foreground solution. Second, if a background solution is switched to the foreground after normal termination, then \Cref{fact:cov-subset-foreground-coverage} ensures that the set of elements in the coverage of the background solution is disjoint from the coverages of the sets that remain in the foreground solution after the switch.
        \item[-] Level Invariant:
        When a set is added to a background solution by the greedy algorithm, its level is set to satisfy the level invariant. When a set is added to the foreground solution to handle element insertion, its level is set to 0, which satisfies the level invariant.
        After a set is added, the level invariant is preserved because the coverage of a set cannot shrink, and the level of a set cannot change.
        
        \item[-] Passive Level Invariant: For a set added to a background solution by the greedy algorithm, Line~\ref{line:set-level} sets the level of the set to at most $p_k$, while all elements in its coverage have passive level at least $p_k$ by \Cref{fact:pointer-monotone}.
        After an element is inserted or deleted, the algorithm sets its passive level to be equal to its level. The switch of a background solution to the foreground after normal termination of a background thread preserves levels and passive levels. So, the passive level invariant is maintained throughout the algorithm. \qedhere
    \end{itemize}
\end{proof}

\paragraph{No Duplicate Sets.} In the previous lemma, we ruled out the possibility of the same element appearing in the coverages of two sets in a hierarchical solution. But, this does not rule out the same set appearing multiple times in the same solution with disjoint coverages. We rule out this latter possibility in the next lemma. 

\begin{lemma}\label{lem:no-duplicate-set}
    In any hierarchical solution $S$ (foreground or background solution), a set $s$ appears at most once.
\end{lemma}
\begin{proof}
    First, we consider a background solution $\bB_k$. We show the following property (we call this {\em maximal coverage}):
    
    {\bf (Maximal Coverage.)} For any element $e$ in the uncovered sub-universe $L^{\bF}_k\setminus \cov(\bB_k)$ and for any set $s\in \bB_k$, we have $e\notin s$. 
    
    Before proving this property, we observe that it implies the statement of the lemma. By maximal coverage, for any set $s$ in $\bB_k$, none of its elements can be in the uncovered sub-universe $L^{\bF}_k \setminus \cov(\bB_k)$. This implies that a duplicate copy of $s$ cannot be added to $\bB_k$ by the greedy algorithm.

    we now show maximal coverage holds inductively over time. For the base case, in the preparation phase of $T_k$, we have that $\bB_k = \emptyset$. So, the property is trivial.
    For the inductive step, there are three cases:
    \begin{itemize}
        \item[-] The greedy algorithm run by $T_k$ adds a new set $s$ to $\bB_k$:
        Since Line \ref{line:greedy-cov-assign} assigns $\cov(s)$ to be $s\cap (L^{\bF}_k\setminus \cov(\bB_k))$, it follows that no element in the remaining uncovered sub-universe $L^{\bF}_k\setminus \cov(\bB_k\cup\{s\})$ can belong to $s$.
        
        \item[-] A new element $e$ is added to $L^{\bF}_k$ due to element insertion:
            In \Cref{alg:insertion}, $e$ is added to the uncovered sub-universe if and only if $e$ does not belong to any set in $\bB_k$. So, maximal coverage holds for $e$.
            
        \item[-] An element $e\in L^{\bF}_k$ is deleted:
        Since \Cref{alg:deletion} does not change the coverages of sets, maximal coverage is unaffected.
    \end{itemize}
    This completes the proof of the lemma for background solutions.
    
    We now consider the foreground solution $\bF$. Note that by \Cref{lem:feasible}, the foreground solution is feasible and hence does not have any uncovered element. So, maximal coverage does not apply to the foreground solution. To establish the lemma for the foreground solution, we use the following property:
    
    {\bf (Highest Coverage.)}
    Suppose element $e$ belongs to $\cov^S(s)$ for some set $s$ in a hierarchical solution $S$. ($S$ can be either the foreground solution $\bF$ or any background solution $\bB_k$.)  It holds that $e$ does not belong to any set $s'$ at a higher level than $s$ in $S$. I.e. if $\lev^S(s') > \lev^S(s)$, then $e\notin s'$.

    Before proving this property, we first show that it implies the statement of the lemma for the foreground solution $\bF$.
    The foreground solution can add a set $s$ in two ways. First, if $s$ is added to handle insertion of element $e$, then $e$ cannot belong to any existing set in $\bF$. So, $s$ was not in $\bF$.
    
    The remaining case is when the foreground solution is updated by switching a background solution $\bB_k$ after its normal termination. Assume for contradiction that some set $s$ is in $\bF$ at some level $> k$ and also in $\bB_k$. 
    Let $t, t'$ denote the time-steps when $\bB_k$ was initiated and when $s$ was added to $\bB_k$ respectively.   
    Note that foreground sets at levels $> k$ can only be modified by the normal termination of a background thread at some level $\ge k$. But, such a normal termination of $\bB_{k'}$ for some $k' > k$ would abort $T_k$. Since $T_k$ had a normal termination, it must be that the sets in levels $> k$ in $\bF$ are identical at time-steps $t, t'$ (and also at termination of $T_k$).
    This means that $s$ was already in $\bF$ at a level $> k$ at time-step $t$. Now, consider any element $e\in s$. We will show that $e$ cannot be in the uncovered sub-universe of $T_k$ at any stage. Since this holds for all elements of $s$, it follows that $s$ cannot be added by the greedy algorithm to $\bB_k$. First, $e \notin L^{\bF}_k(t)$ because this would violate the highest coverage property of $\bF$ at time-step $t$. Thus, $e$ must have been inserted between time-steps $t$ and $t'$. But, in this case, the foreground thread will add $e$ to the set containing $e$ at the highest level in $\bF$. Since $s$ is such a candidate set at level $> k$, the element $e$ cannot be added to $L^{\bF}_k$, and therefore, cannot appear in the uncovered sub-universe of $T_k$. This establishes that $s$ cannot be in $\bB_k$, and therefore, we do not create duplicate sets in the foreground solution $\bF$.
    

    We are left to prove the highest coverage property. Although we only need it for the foreground solution to establish the lemma, we will first show it for background solutions and then use this to show it for the foreground solution. Our proof is by induction over time. For the base case, $\bB_k$ is empty in the preparation phase, and the property is trivial. For the inductive step, there are three cases:
    \begin{itemize}
        \item[-] The greedy algorithm run by $T_k$ adds a new set $s$ to $\bB_k$: 
        For newly covered elements, the highest coverage property follows from the maximal coverage property of those elements before the new set is added. Previously covered elements continue to satisfy the highest coverage property, since the new set has lowest level in $\bB_k$ by \Cref{fact:pointer-monotone}.         

         \item[-] A new element $e$ is added to $L^{\bF}_k$ due to element insertion:
         If $e$ is in some set in $\bB_k$, then \Cref{alg:insertion} will add it to the coverage of the set $s$ with the highest level (among sets in $\bB_k$ containing $e$). So, the highest coverage property holds for $e$.

        \item An element $e\in L^{\bF}_k$ is deleted:
        Since \Cref{alg:deletion} does not change the levels and coverages of sets, highest coverage is unaffected.
    \end{itemize}

    We now proceed to establishing the highest coverage property for the foreground solution. As before, our proof is by induction over time. For the base case, initially $\bF$ is empty, and the property is trivial.
    For the inductive step, there are three cases:
    \begin{itemize}
        \item The foreground solution is updated by switching a background solution $\bB_k$ after its normal termination: 
        This adds sets at levels $\le k+1$ to $\bF$; thus, the highest coverage property in $\bF$ for elements $e$ with $\lev^{\bF}(e)\ge k+1$ are not affected.
        The remaining elements are now covered according to $\cov(\bB_k)$. Let $e$ be such an element that is now covered at level $\lev^{\bB_k}(e) \le k+1$. We note that $e$ cannot belong to a set that was already in $\bF$ before the switch at level $\ge k+1$ by the inductive hypothesis for $\bF$. Furthermore, $e$ cannot belong to a set in $\bB_k$ at a level higher than $\lev^{\bB_k}(e)$ by the highest coverage property of $\bB_k$. Since these are the only new sets added to $\bF$ and they inherit their levels from those in $\bB_k$, the property continues to hold for these sets in $\bF$ after the switch. Combined, this establishes the highest coverage property for $e$ in $\bF$ after the switch.
    
        \item A new element $e$ is inserted:
        If $e$ belongs to any set in $\bF$, then the foreground thread adds $e$ to the coverage of the set at the highest level among all those containing $e$. Otherwise, it adds a new level-0 set if $e$ is not in any set in $\bF$. In both cases, the highest coverage property holds for $e$. 
        
        \item An element $e$ is deleted: 
        The foreground thread does not change the levels and coverages, and hence, the highest coverage property is unaffected.
    \end{itemize}
    This concludes the proof of the highest coverage property for the foreground solution.
\end{proof}

\eat{
\begin{fact}
    We have the following properties on a buffer solution $\bR_k$:
    \begin{itemize}
        \item[-] During preparation, computation, and suspension phases of $T_k$, we have $\bR_k=\emptyset$.
        \item[-] During the copy phase of $T_k$, we have $\bR_k\subseteq \bB_k$. At the end of the copy phase, $\bR_k = \bB_k$.
        \item[-] During the tail phase of $T_k$, we have $\bR_k = \bB_k$.
    \end{itemize}
\end{fact}
}

\subsection{Lifetime Bounds}\label{sec:lifetime}

We set $\cspd = 400$. 
Intuitively, we view $\cspd$ as a sufficiently large constant. So, we will keep some dependence on $\cspd$ in the analysis.

From the rule of one phase in a time-set, we will say a time-step $t$ is both the last time-step of the previous phase and the first time-step of the next phase. In this case, the next phase does not perform work at time-step $t$. So, in the following analysis, we will count the work starting from time-step $t+1$.

\begin{fact}\label{fact:prep-phase-time}
    Suppose that a background thread $T_k$ is in the preparation phase at time-step $t$. Then,
    $T_k$ will finish preparation phase and enter computation phase (or get aborted) by time-step $t+\left\lceil \frac{1.1}{\cspd}\cdot |L^{\bF}_{k}(t)|\right\rceil$.
\end{fact}
\begin{proof}
    Suppose that the preparation phase starts at time-step $\tprep$.
    During the preparation phase, $T_k$ scans elements in the sub-universe $L^{\bF}_{k}$, at a speed of $\cspd$ elements per time-step, except for the last time-step when $T_k$ finishes the preparation phase. On the other hand, the sub-universe $L^{\bF}_{k}$ itself may be modified due to insertion/deletion of elements, but at the speed of at most one element per time-step. So, the scan  finishes  in $\Delta$ time-steps after $t$, where
    $$\Delta \leq \left\lceil\frac{1}{\cspd-1}\cdot |L^{\bF}_{k}(\tprep)|\right\rceil
    \leq \left\lceil 0.01 \cdot |L^{\bF}_{k}(\tprep)|\right\rceil.$$

    Since $t$ is a time-step within the preparation phase, we have $t\in[\tprep, \tprep+\Delta-1]$, and hence $|L^{\bF}_{k}(t)| \ge |L^{\bF}_{k}(\tprep)| - (t-\tprep) \geq |L^{\bF}_{k}(\tprep)| -  (\Delta-1) \ge 0.99 \cdot |L^{\bF}_{k}(\tprep)|$. Thus, we get $$\Delta\le \left\lceil \frac{1}{\cspd - 1}\cdot |L^{\bF}_{k}(\tprep)|\right\rceil
    \le \left\lceil\frac{1}{0.99 \cdot (\cspd - 1)}\cdot |L^{\bF}_{k}(t)|\right\rceil
    \leq  \left\lceil\frac{1.1}{\cspd}\cdot |L^{\bF}_{k}(t)|\right\rceil.$$ This concludes the proof.
\end{proof}

\begin{fact}\label{fact:comp-phase-time}
    Suppose that a background thread $T_k$ is in the computation phase at time-step $t$. Then,
    $T_k$ will finish computation phase  by time-step $t+\left\lceil\frac{1.1}{\cspd}\cdot |L^{\bF}_{k}(t) \setminus  \cov(\bB_k(t))|\right\rceil$.
\end{fact}
\begin{proof}
    During the computation phase, $T_k$ keeps covering elements at a speed of $\cspd$ elements per time-step. These elements are from the uncovered universe $L^{\bF}_k \setminus \cov(\bB_k)$.
    The uncovered universe may be modified due to insertion/deletion of elements, but at the speed of at most one element per time-step. 
    In contrast, the uncovered universe cannot be modified by termination of lower levels, as per  \Cref{fact:background-feasible} (iii).
    This implies that the computation phase finishes in $\left\lceil\frac{1}{\cspd-1}\cdot |L^{\bF}_{k}(t) \setminus  \cov(\bB_k(t))|\right\rceil \leq \left\lceil\frac{1.1}{\cspd}\cdot |L^{\bF}_{k}(t) \setminus  \cov(\bB_k(t))|\right\rceil$ time-steps after $t$.
\end{proof}

\begin{lemma}\label{lem:copy-phase-time}
    Suppose that a background thread $T_k$ is in the copy phase at time-step $t$. Then,
    $T_k$ will finish copy phase and enter tail phase (or get aborted)  by time-step $t+\left\lceil\frac{1}{\cspd}\cdot |\bB_k(t)|\right\rceil$.
\end{lemma}
\begin{proof}
    During each time-step within the copy phase, the thread $T_k$ copies $\cspd$ sets from the background solution $\bB_k$ into the buffer solution $\bR_k$. Further, the collection of sets that define the background solution $\bB_k$ does {\em not} change during the copy phase. Accordingly, the copy phase finishes in $\left\lceil\frac{1}{\cspd}\cdot |\bB_k(t)|\right\rceil$ time-steps after $t$.
\end{proof}

\begin{lemma}\label{lem:tail-phase-time}
Suppose that a background thread $T_k$ is in the tail phase at time-step $t$. Then, $T_k$ will terminate by time-step $t+\min\left\{ \left\lceil\frac{1.1}{\cspd} \cdot |L^{\bF}_{k}(t) \setminus  \cov(\bB_k(t))|\right\rceil, \left\lceil\frac{1.4}{\cspd} \cdot |\bB_k(t)|\right\rceil \right\}$.
\end{lemma}
\begin{proof}
   Using the same argument as in the proof of \Cref{fact:comp-phase-time}, we infer that $T_k$ will terminate by time-step $t+ \left\lceil\frac{1.1}{\cspd} \cdot |L^{\bF}_{k}(t) \setminus  \cov(\bB_k(t))|\right\rceil$. It now remains to show that $T_k$ will terminate by time-step $t+ \left\lceil\frac{1.2}{\cspd} \cdot |\bB_k(t)|\right\rceil$.  
   
   Suppose that $T_k$ enters suspension phase, copy phase and tail phase respectively at time-steps $\tsus, \tcopy$ and $\ttail$. This means that 
   \begin{align}
  \tsus \leq \tcopy \leq \ttail \leq t, \text{ and } \\
 t-\ttail \leq \left\lceil\frac{1.1}{\cspd} \cdot |L^{\bF}_{k}(\ttail) \setminus  \cov(\bB_k(\ttail))| \right\rceil   \label{eq:new:0}.
   \end{align}

    As the uncovered universe  $L^{\bF}_{k}\setminus \cov(\bB_{k})$ grows by at most one element per time-step during $[\tsus, t]$, we get
    \begin{align}
    \label{eq:new:2}
    |L^{\bF}_{k}(t_2) \setminus \cov(\bB_{k}(t_2))| & \leq |L^{\bF}_{k}(t_1) \setminus \cov(\bB_{k}(t_1))| + (t_2 -t_1), \text{ for all } \tsus \leq t_1 \leq t_2 \leq t.
    \end{align}
    Furthermore, by the criterion to enter suspension phase, we have 
    \begin{equation}
    \label{eq:new:1}
    |L^{\bF}_{k}(\tsus) \setminus \cov(\bB_{k}(\tsus))| \le  |\bB_k(\tsus)|.
    \end{equation}
    Since  $T_k$ is in the suspension phase for at most $0.1\cdot |\bB_k(\tsus)|$ time-steps, it follows that  $$\tcopy - \tsus 
    \le 0.1\cdot |\bB_k(\tsus)|.$$ 
    
    During the interval $[\tsus, \ttail]$, the collection of sets that define $\bB_k$ remain unchanged. This implies that  
    \begin{equation}
    |\bB_k(\tsus)| = |\bB_k(\tcopy)| = |\bB_k(\ttail)|  \label{eq:new:4}
    \end{equation}
    
    Next, by \Cref{lem:copy-phase-time}, we have  $\ttail - \tcopy \leq \left\lceil\frac{1}{\cspd}\cdot |\bB_k(\tcopy)|\right\rceil \le \left\lceil 0.1 \cdot |\bB_k(\tcopy)|\right\rceil$.
    This gives us
     \begin{equation}
\label{eq:new:3}\ttail - \tsus = (\tcopy - \tsus) + (\ttail - \tcopy)  \leq 0.11 \left(|\bB_k(\tsus)| + |\bB_k(\tcopy)| \right) = 0.22 |\bB_k(\tsus)|.
    \end{equation}

    From the preceding discussion, we now derive that
    \begin{align}
    \label{eq:align:1}
    |L^{\bF}_{k}(\ttail)|    
    \setminus \cov(\bB_{k}(\ttail))| &  \leq 
    |L^{\bF}_{k}(\tsus) \setminus \cov(\bB_{k}(\tsus))| + (\ttail-\tsus) & (\text{\Cref{eq:new:2}}) \nonumber \\
    & \leq |\bB_k(\tsus)| + 0.22 |\bB_k(\tsus)| & (\text{\Cref{eq:new:1,eq:new:3}}) \nonumber \\  
    & \leq 1.22 \cdot |\bB_k(\ttail)| & (\text{\Cref{eq:new:4}}).
     \end{align}
 
Next, we observe that no set gets deleted from  $\bB_k$ during the tail phase. This gives us
\begin{equation}
\label{eq:last}
|\bB_k(\ttail)| \leq |\bB_k(t)|.
\end{equation}
Finally, we conclude that
\begin{align}
\label{eq:align:2}
|L^{\bF}_{k}(t) \setminus \cov(\bB_{k}(t))| & \leq |L^{\bF}_{k}(\ttail) \setminus \cov(\bB_{k}(\ttail))| + (t-\ttail) & (\text{\Cref{eq:new:2}}) \nonumber \\
& \leq 1.01 \cdot |L^{\bF}_{k}(\ttail) \setminus \cov(\bB_{k}(\ttail))| + 1 & (\text{\Cref{eq:new:0}}) \nonumber \\
& \leq 1.24 \cdot |\bB_k(\ttail)| + 0.01 \cdot |\bB_k(\tsus)| & (\text{\Cref{eq:align:1}, rule of shortcut threads}) \nonumber \\
& \leq 1.25 \cdot |\bB_k(t)| & (\text{\Cref{eq:new:4,eq:last}})
\end{align}
From \Cref{eq:align:2}, we infer that $T_k$ will terminate by time-step $t+ \left\lceil \frac{1.1}{\cspd} \cdot |L^{\bF}_{k}(t) \setminus  \cov(\bB_k(t))|\right\rceil \leq t+\left\lceil\frac{1.4}{\cspd} \cdot |\bB_k(t)|\right\rceil$. This concludes the proof of the lemma.
\end{proof}

\begin{corollary}\label{cor:copy-tail-phase-time}
    Suppose that a background thread $T_k$ is  in  copy or tail phase at time-step $t$. Then, the thread 
    $T_k$ will terminate by time-step   $t+\frac{5}{\cspd}\cdot \min\left\{\tausus_k, |\bB_k(t)| \right\} = t + \frac{5}{\cspd}\cdot \tausus_k$.
\end{corollary}

\begin{proof}
Suppose that $T_k$ enters suspension, copy and tail phases respectively at time-steps $\tsus, \tcopy$ and $\ttail$, and terminates at time-step $\tend$. This means that $\tsus \leq \tcopy \leq \ttail \leq \tend$, and $t \in [\tcopy, 
\tend]$.  Now, from the description of our algorithm, we recall that the collection of sets that define  $\bB_k$ does {\em not} change during the interval $[\tsus, \ttail]$. Further, no set gets deleted from $\bB_k$ during the interval $[\ttail, \tend]$.
Thus, we get
\begin{equation}
\label{eq:no:background:change}
\tausus_k = | \bB_k(\tsus) | = | \bB_k(\tcopy) | = |\bB_k(\ttail)|  \leq |\bB_k(t)|.
\end{equation}
This implies that 
\begin{equation}
\label{eq:no:background:0}
\min\left\{\tausus_k, |\bB_k(t)| \right\} = \tausus_k.
\end{equation}
We now consider two possible cases.

\medskip
\noindent {\em Case 1: $t \in [\ttail, \tend]$.} In this case, the corollary immediately follow from \Cref{eq:no:background:0} and \Cref{lem:tail-phase-time}.

\medskip
\noindent {\em Case 2: $t \in [\tcopy, \ttail]$.} In this case, we derive that
\begin{align*}
\tend - t & = (\tend - \ttail) + (\ttail - t) & \\
& \leq \frac{1.4}{\cspd} \cdot |\bB_k(\ttail)| + \frac{1}{\cspd} \cdot |\bB_k(t)| + 2 & (\text{\Cref{lem:copy-phase-time,lem:tail-phase-time}}) \\
& \leq \frac{1.4}{\cspd} \cdot |\bB_k(\ttail)| + \frac{1}{\cspd} \cdot |\bB_k(t)| + \frac{2}{\cspd} \cdot |\bB_k(\tsus)| & (\text{rule of shortcut threads}) \\
& \leq \frac{5}{\cspd} \cdot |\bB_k(t)| & (\text{\Cref{eq:no:background:change}}) 
\end{align*}
This concludes the proof of the corollary.
\end{proof}

\begin{lemma}
\label{lm:wait:time}
    Every background thread $T_k$  remains in the suspension phase for less than $0.1 \cdot \tausus_k$ time-steps.
\end{lemma}

The above lemma shows that removing the time limit of suspension phase results in an equivalent algorithm. We add the limit in algorithm description to simplify the proofs and eliminate circular argument.

We combine the results in this section as follows.
\begin{lemma}\label{lem:overall-lifetime}
    Suppose that a background thread $T_k$ is running at time-step $t$. Then, the lifetime of $T_k$ is at most $0.2\cdot |L^{\bF}_{k}(t)|$ time-steps.
\end{lemma}
\begin{proof}
Suppose $T_k$ starts at $\tprep$ and terminates at $\tend$.
If $|L^\bF_k(\tprep)|\le \cspd-1$, then $T_k$ terminate in one time-step by the rule of base threads. In this case, the statement holds. In the rest of the proof, we may assume $|L^\bF_k(\tprep)|\ge \cspd$, i.e.\ $1 \le \frac{|L^\bF_k(\tprep)|}{\cspd}$.

$T_k$ may terminate in one of the three cases: (1) normal termination at tail phase, (2) normal termination by rule of shortcut threads at computation phase, and (3) aborted by a higher thread. We may assume case (1) w.l.o.g., because its upper bound subsumes the other cases.

Suppose that $T_k$ enters computation phases, suspension phase, and copy phase respectively at time-steps $ \tcomp, \tsus$ and $\tcopy$.
We have the following time bounds:
\begin{align*}
    \tcomp-\tprep &\le \left\lceil \frac{1.1}{\cspd}\cdot |L^\bF_k(\tprep)|\right\rceil  & \text{(\Cref{fact:prep-phase-time})}\\
    \tsus-\tcomp &\le \left\lceil \frac{1.1}{\cspd}\cdot |L^\bF_k(\tcomp)|\right\rceil  & \text{(\Cref{fact:comp-phase-time})}\\
    \tcopy-\tsus &\le \left\lceil 0.1\cdot \tausus_k \right\rceil  & \text{(\Cref{lm:wait:time})}\\
    \tend-\tcopy & \le \frac{5}{\cspd}\cdot \tausus_k & \text{(\Cref{cor:copy-tail-phase-time})}
\end{align*}

We will show that $\tend-\tprep \le \Delta := 0.2\cdot |L^\bF_k(\tprep)|$.
During $t\in [\tprep, \tprep+\Delta]$, $|L^\bF_k(t)|$ can change by at most 1 per time-step, so we always have $0.8 |L^\bF_k(\tprep)|\le |L^\bF_k(t)|\le 1.2 |L^\bF_k(\tprep)|$.
To bound $\tausus_k = |\bB_k(\tsus)|$, we consider a bijection from each set in $\bB_k(\tsus)$ to an arbitrary element in its coverage. These elements are distinct and belong to $L^\bF_k$ when the set is added by the greedy algorithm. So, $\tausus_k \le 1.2 |L^\bF_k(\tprep)|$. In conclusion,
\[\tend - \tprep
\le(1.1+1.1\times 1.2 + 0.1\cspd + 5 + 3)\frac{1}{\cspd}|L^\bF_k(\tprep)|
\le 0.15|L^\bF_k(\tprep)| \le 0.2 |L^\bF_k(t)|\]
\end{proof}

\subsubsection{Proof of \Cref{lm:wait:time}}

   Suppose that the thread $T_k$ enters the suspension phase at time-step $\tsus$ and ends at time-step $\tend >  \tsus$. From the algorithm description, this can potentially happen because of one of the following three reasons.
    \begin{itemize}
    \item (i) $T_k$ normally terminates  at time-step $\tend$. 
    \item (ii) $T_k$ gets aborted because another  thread $T_j$, with $j > k$, normally terminates at time-step $\tend$.
    \item (iii) $T_k$ spends $0.1 \cdot \tausus_k$ time-steps in the suspension phase, and $\tend = \tsus + \left\lfloor 0.1 \cdot \tausus_k\right\rfloor$.
    \end{itemize}
    We will show that $T_k$  ends because of reason (i) or reason (ii) (and {\em never} because of reason (iii)). Specifically, suppose that $T_k$ is in suspension phase during the interval $[\tsus, \tsus + \deltasus]$, for some integer $\deltasus \geq 0$. We will show: 
    \begin{equation}
    \label{eq:toshow}
    \deltasus \le 0.1\cdot  \tausus_k.
    \end{equation}

    Consider any $t \in [\tsus, \tsus + \deltasus-1]$.  We say that $T_k$ is {\em blocked-from-above at $t$} iff there is some other thread $T_j$, with $j > k$ and $\frac{1}{2} \cdot \tausus_j < \tausus_k$, that is in either  copy  or tail phase at the end of time-step $t$. In contrast, we say that $T_k$ is {\em blocked-from-below at  $t$} iff (a) it is {\em not} blocked-from-above at  $t$, and (b) there is some other thread $T_j$, with $j < k$ and $\frac{1}{2} \cdot \tausus_j < \tausus_k$, that is in either copy or tail phase at the end of time-step $t$. The observation below follows from the description of our algorithm (see \Cref{alg:transition}).

    \begin{observation}
   \label{obs:block} $T_k$ is either blocked-from-above or blocked-from-below at every  $t \in [\tsus, \tsus+\deltasus-1]$.
    \end{observation}

    \begin{claim}
    \label{cl:above} 
    Suppose that the thread $T_k$ is blocked-from-above at some time-step $t \in [\tsus, \tsus+\deltasus-1]$. Then, the thread $T_k$ must terminate before time-step $t+ \frac{10}{\cspd} \cdot \tausus_k$.    \end{claim}

    \begin{proof}
    Since $T_k$ is blocked-from-above at some time-step $t \in [\tsus, \tsus+\Delta-1]$, there exists another thread $T_j$, with $j > k$ and $\frac{1}{2} \cdot \tausus_j < \tausus_k$, that is in either copy or tail phase at time-step $t$.  Then, \Cref{cor:copy-tail-phase-time} guarantees that  $T_j$ will terminate by time-step $t+ \frac{5}{\cspd} \cdot \tausus_j < t+ \frac{10}{\cspd} \cdot \tausus_k$. Finally, note that since $j > k$, the thread $T_k$ gets aborted when $T_j$ terminates.
    \end{proof}

    Let $\tdown(t)$ denote the collection of threads $T_j$, with $j < k$ and $\frac{1}{2} \cdot \tausus_{j} < \tausus_k$, that are in either copy or tail phase at the end of time-step $t$. By \Cref{cl:above}, if $T_k$ is blocked-from-above at time-step $\tsus$, then $\deltasus < \frac{10}{\cspd} \cdot \tausus_k < 0.1 \cdot \tausus_k$.
    This immediately implies \Cref{lm:wait:time}. Accordingly, throughout the rest of the proof, we make the following assumption. 

    \begin{assumption}
    \label{assume:new}
    The thread $T_k$ is {\em not} blocked-from-above at time-step $\tsus$, and $\deltasus \geq 1$. 
    \end{assumption}

    \begin{corollary}
    \label{cor:new:start}
    We have $\tdown(\tsus) \neq \emptyset$.
    \end{corollary}

    \begin{proof}
    Follows from \Cref{obs:block} and \Cref{assume:new}.
    \end{proof}

    Let $t^\star(T_j)$ denote the last time-step during which a given thread $T_j$ is in either copy or tail phase. For some integer $r \geq 0$, we now define a sequence of $r$ consecutive intervals $[t_1, t_2-1], [t_2, t_3-1], \ldots, [t_{r}, t_{r+1}-1]$, and  a sequence of indices $j_1, j_2, \ldots, j_r$, according to the procedure described in \Cref{alg:intervals}. 
   \begin{algorithm}
\caption{Constructing the intervals $[t_1, t_2-1],  \ldots, [t_r, t_{r+1}-1]$ and the indices $j_1,  \ldots, j_r$.}
    $t_1 \leftarrow \tsus$ \;  
    $T_{j_1} \leftarrow \arg \max_{T_j \in \tdown(t_1)} \{ t^\star(T_j)\}$, breaking ties arbitrarily (see \Cref{cor:new:start}). \; 
    $t_2 \leftarrow t^\star(T_{j_1}) + 1$ \; 
    $i \leftarrow 1$. \; 
    \While{$\tdown(t_{i+1}-1) \setminus \{T_{j_{i}}\} \neq \emptyset$}{
        $T_{j_{i+1}} \leftarrow \arg \max_{T_j \in \tdown(t_{i+1}-1)} \{ t^\star(T_j)\}$, breaking ties arbitrarily. \;  
        $t_{i+2} \leftarrow t^\star(T_{j_{i+1}}) + 1$ \; 
        $i \leftarrow i+1$ \;
}
$r \leftarrow i$ \; 
\label{alg:intervals}
\end{algorithm}

 \begin{observation}
    \label{obs:switch}
    For all $i \in [1, r-1]$, we have $\tausus_{j_{i+1}} \leq \frac{1}{2} \cdot \tausus_{j_{i}}$. 
    \end{observation}

    \begin{proof}
    By definition, we have $T_{j_{i+1}} \in \tdown(t_{i+1}-1)$, and  hence $t^\star(T_{j_{i+1}}) \geq t_{i+1} > t^\star(T_{j_{i}})$. It follows that $T_{j_{i+1}} \notin \tdown(t_{i}-1)$, for otherwise we would have defined $j_{i} \leftarrow j_{i+1}$ at time-step $t_{i}-1$. Accordingly, the thread $T_{j_{i+1}}$ enters copy phase during some time-step $t \in [t_{i}, t_{i+1}-1]$, when thread $T_{j_{i}}$ is in either copy or tail phase. Thus, from the description of \Cref{alg:transition}, we have $\tausus_{j_{i+1}} \leq \frac{1}{2} \cdot \tausus_{j_{i}}$.
    \end{proof}

    In particular, \Cref{obs:switch} implies that $r = O(\log n)$,  because $\tausus_j \in [1, n]$ for all background threads $T_j$.

    \begin{observation}
    \label{obs:end:block}
    Suppose that $t_{r+1} \in [\tsus, \tsus + \deltasus - 1]$. Then, the thread $T_k$ is blocked-from-above at $t_{r+1}$.
    \end{observation}

    \begin{proof}
    From the description of \Cref{alg:intervals}, we have $\tdown(t_{r+1}-1) \setminus \{ T_{j_r} \} = \emptyset$. Furthermore, since $t_{r+1} - 1 = t^\star(T_{j_r})$, the thread $T_{j_r}$ itself is no longer in copy or tail phase after the end of time-step $t_{r+1}-1$. This means that at the start of time-step $t_{r+1}$, every thread $T_j$ with $j < k$ that is in copy or tail phase has $\frac{1}{2} \cdot \tausus_j \geq \tausus_k$. Nevertheless, since we have assumed that $t_{r+1} \in [\tsus, \tsus+\deltasus-1]$, the thread $T_k$ does {\em not} enter copy or tail phase during time-step $t_{r+1}$. This can happen only if there exists some other thread $T_j$, with $j > k$ and $\frac{1}{2} \cdot \tausus_j < \tausus_k$, that is in copy or tail phase at the end of step $t_{r+1}$. By definition, this means that the thread $T_k$ is blocked-from-above at $t_{r+1}$.
    \end{proof}

    \begin{observation}
    \label{obs:sum}
    We have $t_{r+1} - t_1 \leq \frac{20}{\cspd} \cdot \tausus_{k}$.
    \end{observation}

    \begin{proof}
    Consider any $i \in [1, r]$. From the description in \Cref{alg:intervals}, it follows that the thread $T_{j_i}$ is in copy or tail phase at every time-step $t \in [t_i, t_{i+1}-1]$.  Accordingly, \Cref{cor:copy-tail-phase-time} guarantees that $$t_{i+1} - t_i \leq \frac{5}{\cspd} \cdot \tausus_{j_i}.$$
    Summing the above inequality over all $i \in [1, r]$, we get:
    $$t_{r+1} - t_1 \leq \frac{5}{\cspd} \cdot \sum_{i=1}^r \tausus_{j_i} \leq \frac{10}{\cspd} \cdot \tausus_{j_1} < \frac{20}{\cspd} \cdot \tausus_{k}.$$
    The penultimate step in the above derivation follows from \Cref{obs:switch}, whereas the last step holds because $T_{j_1} \in \tdown(t_1)$.
    \end{proof}

We are now ready to complete the proof of \Cref{lm:wait:time}. There are two cases to consider.

\medskip
\noindent {\em Case 1: $t_{r+1} \in [\tsus, \tsus + \deltasus -1]$.} In this case, \Cref{cl:above}, \Cref{obs:end:block} and  \Cref{obs:sum} together imply that $T_k$ terminates before time-step $t_{r+1} + \frac{10}{\cspd} \cdot \tausus_k \leq t_1 + \frac{30}{\cspd} \cdot \tausus_k = \tsus + \frac{30}{\cspd} \cdot \tausus_k \le \tsus +0.1 \cdot \tausus_k$.
This concludes the proof of \Cref{lm:wait:time}.

\medskip
\noindent {\em Case 2: $t_{r+1} \notin [\tsus, \tsus + \deltasus -1]$.} In this case, we clearly have $\deltasus \leq t_{r+1} - \tsus = t_{r+1} - t_1 \leq \frac{20}{\cspd} \cdot \tausus_k < 0.1 \cdot \tausus_k$. Here, the penultimate inequality follows from \Cref{obs:sum}.
This concludes the proof of \Cref{lm:wait:time}.

\subsection{Analyzing Recourse}


In this section, we establish a bound of $O(\log n)$ on the worst-case recourse of our dynamic set cover algorithm. We separate the total recourse into {\em insertion recourse} and {\em deletion recourse}. Insertion recourse measures the maximum number of sets that are added to the output of the algorithm in any time-step. Similarly, deletion recourse refers to the maximum number of sets that are removed from the output in any time-step. First, we bound the insertion recourse of the algorithm in the next lemma. 

\begin{lemma}\label{lem:recourse}
    The insertion recourse is $O(\log n)$.
\end{lemma}
\begin{proof}

    Note that the foreground solution $\bF$ and buffer solutions $\bR_k$ are part of the output of the algorithm. A background thread $T_k$ copies sets from the background solution $\bB_k$ to the corresponding buffer solution $\bR_k$ at the rate of $\cspd=O(1)$ sets per time-step in the copy and tail phases. Since there are at most $\ell_{\max} = O(\log n)$ background threads in copy or tail phases in any time-step, the buffer solutions cumulatively add at most $O(\log n)$ sets in a time-step. Now, consider the foreground solution $\bF$. Sets are added to $\bF$ in two ways. First, a single set may be added to $\bF$ on the insertion of an element that does not belong to any set in $\bF$. Second, $\bF$ is updated after normal termination of a background thread $T_k$. But, this process simply adds to $\bF$ the sets in $\bB_k$ that are already in the buffer solution $\bR_k$. Therefore, this latter step does not add to insertion recourse of the algorithm. Overall, this means that the insertion recourse is $O(\log n)$.
\end{proof}

We are left to bound the deletion recourse of the algorithm. First, note that the worst-case bound on insertion recourse immediately implies the same bound on deletion recourse, but only in an amortized sense. So, we need to de-amortize the deletion recourse of the algorithm. Instead of arguing specifically about our algorithm, we prove a more general property about de-amortizing deletion recourse. Consider any online optimization problem where a solution comprises a set of objects, and the goal is to minimize the number of objects subject to a feasibility constraint that changes online. Let us call such a problem an online unweighted minimization problem. We proved a general reduction for de-amortizing deletion recourse in \Cref{lem:new:recourse-reduction}, which we restate below. 

\recourse*

\eat{
\begin{lemma}\label{lem:recourse-reduction}
    Consider an online unweighted minimization problem, where the size of the optimal solution changes by at most $\delta$ in each time-step.
    Suppose there exists an algorithm $\cal A$ that maintains an $\alpha$-approximate solution and has worst-case insertion recourse at most $\beta$, for some parameters $\alpha, \beta\ge 1$. (The worst-case deletion recourse of algorithm $\cal A$ can be arbitrarily large.)
    Then, there exists an alternative algorithm $\cal A'$ that also maintains an $\alpha$-approximate solution and has worst-case recourse (both insertion and deletion) at most $2\beta+\delta\alpha$. 
\end{lemma}
\begin{proof}
    We define algorithm $\cal A'$ to simulate $\cal A$, but with the following modification: 
    Whenever $\cal A$ deletes objects from its solution, move those objects to a garbage set instead of deleting them immediately in $\cal A'$. The garbage set is added to the output of $\cal A$ to form the output of $\cal A'$.
    Moreover, $\cal A'$ performs an additional garbage removal operation at the end of the processing in each time-step: delete any $\beta+\delta\alpha$ objects (or all objects if fewer) from the garbage set.
    The recourse bound for $\cal A'$ is straightforward from this description.

    We bound the approximation factor of $\cal A'$ by induction over time. 
    For the base case, note that whenever the garbage set is empty (which is the case at initiation), the solutions of $\cal A$ and $\cal A'$ are identical. Hence, at these time-steps, $\cal A'$ is an $\alpha$-approximation.
    The remaining case is when the garbage set is non-empty at the end of a time-step $t$.
    In this case, $\cal A'$ must have deleted $\beta+\delta\alpha$ objects from the garbage set in time-step $t$. 
    Since $\cal A'$ has insertion recourse at most $\beta$, the solution size at the end of time-step $t$ decreases by at least $\delta\alpha$ compared to that at the end of time-step $t-1$. Combined with the induction hypothesis that the solution was an $\alpha$-approximation at the end of $t-1$, and the assumption that optimal solution size can only change by at most $\delta$ in a time-step, we conclude that the solution is still an $\alpha$-approximation at the end of time-step $t$.
\end{proof}
}

\subsection{Approximation Factor}

We start by introducing the notion of an {\em extended background solution}, which will be used in our analysis throughout the rest of this section.

\medskip
\noindent {\bf Extended background solution.} Let $T_k$ be a background thread at some time-step $t$. Consider a thought experiment where we complete the execution of the remaining operations specified by \Cref{alg:greedy} at the present time-step $t$, and let $\tmB_k(t)$ denote the solution returned by this algorithm. Note that $\tmB_k(t) \supseteq \bB_k(t)$ and $\cov(\tmB_k(t)) \supseteq L_k^{\bF}(t)$. In other words, $\tmB_k(t)$ is a feasible hierarchical solution for the input defined by $L_k^{\bF}(t)$.  We define $\tmB_k(t)$ to be the {\em extended background solution} of the thread $T_k$ at time-step $t$. 
Recall that we say $\tmB_k(t)$ is $\epsilon$-dirty at some level $j$ if $\left|P_j^{\tmB_k(t)}\right| > \epsilon \cdot \left|A_j^{\tmB_k(t)}\right|$, and $\epsilon$-tidy at level $j$ otherwise.

\begin{lemma}\label{lem:background-tidy}
Consider any background thread $T_k$ at any time-step $t$. Then, $\tmB_k(t)$ is $0.2$-tidy at every level $j \leq k$.
\end{lemma}
\begin{proof}
   Fix any level $j \leq k$. Let $\tcomp$ denote the time-step at which the background thread $T_k$ enters the computation phase, and let $\tend > \tcomp$ denote the time-step at which $T_k$ terminates. Recall that at time-step $\tcomp$, the copy pointer $p_k$ of the background thread $T_k$ is initialized to $k+1 > j$. Furthermore, the value of $p_k$ can only decrease over time (\Cref{fact:pointer-monotone}). Based on this observation, we next define a time-step $t_j$, as follows.
   \begin{itemize}
   \item If $p_k > j$ at time-step $\tend$, then we define $t_j := \tend$.
   \item Otherwise, we define $t_j$ to be the first (i.e., smallest) time-step at which $p_k$ becomes equal to $j$.
   \end{itemize}
   At time-step $\tcomp$, we have $\plev^{\bB_k}(e) \geq k+1$ for all elements $e \in L_k^{\bF}$. Subsequently, the passive levels in $\bB_k$ can only be modified due to insertions/deletions of elements. To be more specific, while handling the insertion of an element $e$, the thread $T_k$ ensures that $\plev^{\bB_k}(e) \geq p_k$ (see \Cref{alg:insertion}). Similarly, while handling the deletion of an element $e$, if $e \in \cov(\bB_k)$ then the thread $T_k$ ensures that $\plev^{\bB_k}(e) \geq p_k$; otherwise  $T_k$ simply removes $e$ from its uncovered sub-universe (see \Cref{alg:deletion}). This implies that up until the start of time-step $t_j$, no element $e$ has $\plev^{\bB_k}(e) \leq j$
   (including $t_j$, since at time-step $t_j$, $T_k$ handles insertion/deletion before the greedy algorithm drops $p_k$ to $\le j$).
   In other words, we have
   \begin{equation}
    \label{eq:copypointer}
   P_j^{\tmB_k}(t') = \emptyset \text{ for all time-steps } t' \in [\tcomp, t_j]. 
   \end{equation}
   If $t \leq t_j$, then \Cref{eq:copypointer} immediately implies the lemma. Accordingly, for the rest of the proof, we assume:
   \begin{equation}
   \label{eq:assume}
   t_j < t \leq \tend.
   \end{equation}

 During any given time-step within the interval $[t_j+1, t]$, only the element $e$ being inserted/deleted can potentially have its passive level in $\tmB_k$ set to $\leq j$, whereas the passive level of every other element in $\tmB_k$ remains unchanged.  Thus, from \Cref{eq:assume}, we infer that
   \begin{equation}
   \label{eq:passive:1}
   \left| P_j^{\tmB_k}(t) \right| \leq t-t_j.
   \end{equation}
  
   We now fork into each of two possible cases, one after another.

   \medskip
   \noindent 
   {\em Case 1: During the interval $[t_j+1, t]$, the thread $T_k$ does not exit computation phase.} Here, we observe that from time-step $t_j+1$ onward, the thread $T_k$ only adds sets at levels $\leq j$ to $\bB_k$. 
   During the computation phase, $T_k$ covers $\cspd$ elements per time-step, except for the last time-step of computation phase. Since we assumed $T_k$ does not exit computation phase, the latter case cannot happen.
   So, during $[t_j+1, t]$, $T_k$ covers $\cspd$ elements at levels $\leq j$ per time-step. It follows that
   \begin{equation}
   \label{eq:passive:2}
   \left| C_j^{\tmB_k}(t) \right| \geq \left| C_j^{\bB_k}(t) \right| = \cspd \cdot (t-t_j).
   \end{equation}

   From \Cref{eq:passive:1} and \Cref{eq:passive:2}, we get $\left| C_j^{\tmB_k}(t) \right| \geq \cspd\cdot  \left| P_j^{\tmB_k}(t) \right|$. Since $A_j^{\tmB_k}(t) = C_j^{\tmB_k}(t) \setminus P_j^{\tmB_k}(t)$,  this implies that $\tmB_k$ is $\epsilon$-tidy at level $j$ at time-step $t$, and concludes the proof of the lemma. 

   {\em Case 2: During the interval $[t_j+1, t]$, the thread $T_k$ is in suspension or copy phase for $\Delta$ time-steps, for some $\Delta > 0$.} Here, let $\tsus \in [t_j, t]$ denote the time-step at which $T_k$ enters the suspension phase. Since the sets that define the background solution $\bB_k$ do {\em not} change during suspension or copy phase, throughout the duration of the suspension or copy phase, we continue to have $|\bB_k| = \tausus_k = |\bB_k(\tsus)|$. By the rule of shortcut threads, $|\bB_k(\tsus)| > \cspd$.  
   Accordingly, \Cref{lem:copy-phase-time} implies that 
   \begin{equation}
   \label{eq:passive:3}
   \Delta < 0.1 \cdot |\bB_k(\tsus)| + \frac{1}{\cspd} \cdot |\bB_k(\tsus)| + 2 \le 0.11 \cdot |\bB_k(\tsus)|.
   \end{equation}

Next, observe that during the interval $[t_j+1, t]$, the  thread $T_k$ covers $\cspd$ elements at levels $\leq j$ per time-step,  except when $T_k$ is in suspension or copy phase. Thus, we have
\begin{equation}
\label{eq:passive:4}
\left| C_j^{\tmB_k}(t) \right| \geq \left| C_j^{\bB_k}(t) \right| = \cspd \cdot (t-t_j - \Delta).
\end{equation}

Since the thread $T_k$ enters the suspension phase at time-step $\tsus$, we must necessarily have $|\bB_k(\tsus)| \leq |L_k^{\bF}(\tsus) \setminus \cov(\bB_k(\tsus))|$. Further, since $t_j \leq \tsus$, all the elements in $L_k^{\bF}(\tsus) \setminus \cov(\bB_k(\tsus))$ are covered in $\tmB_k(\tsus)$ at levels $\leq j$. This implies that  
\begin{equation}
\label{eq:passive:5}
|\bB_k(\tsus)| \leq \left|L_k^{\bF}(\tsus) \setminus \cov(\bB_k(\tsus))\right| \leq \left|C_j^{\tmB_k}(\tsus) \right|.
\end{equation}

Next, we compare the two sets $C_j^{\tmB_k}(t)$ and $C_j^{\tmB_k}(\tsus)$. Specifically, we note if an element appears in one of these sets but not in the other, then that element must have been inserted or deleted during the interval $[\tsus+1, t]$. This gives us
\begin{equation}
\label{eq:passive:6}
\left| C_j^{\tmB_k}(\tsus) \right| \leq \left| C_j^{\tmB_k}(t) \right| + (t- \tsus) \leq \left| C_j^{\tmB_k}(t) \right| + (t- t_j).
\end{equation}

We now put together all these observations, and derive that
\begin{align*}
\left| C_j^{\tmB_k}(t) \right| & \geq \cspd \cdot (t-t_j - \Delta) & (\text{\Cref{eq:passive:4}}) \\
& \geq \cspd \cdot \left((t-t_j) - 0.11  \cdot \left| \bB_k(\tsus) \right|\right)  & (\text{\Cref{eq:passive:3}}) \\
& \geq \cspd \cdot \left(  (t-t_j) - 0.11 \cdot \left( \left| C_j^{\tmB_k}(t) \right| + (t-t_j)  \right)\right)  & (\text{\Cref{eq:passive:5,eq:passive:6}}) \\
(1+0.11\cspd)\left| C_j^{\tmB_k}(t) \right| &\ge 0.89\cspd\cdot (t-t_j)
\ge 0.89\cspd \cdot \left| P_j^{\tmB_k}(t) \right| & (\text{\Cref{eq:passive:1}})\\
\left| P_j^{\tmB_k}(t) \right| &\le \frac{1+0.11\cspd}{0.89\cspd}\cdot  \left| C_j^{\tmB_k}(t) \right| \le 0.2 \cdot \left| C_j^{\tmB_k}(t) \right|
\end{align*}
Since $A_j^{\tmB_k}(t) = C_j^{\tmB_k}(t) \setminus P_j^{\tmB_k}(t)$, rearranging the terms in the last inequality, we infer that $\tmB_k$ is $0.2$-tidy at level $j$ at time-step $t$. This concludes the proof of the lemma. 
\end{proof}

\begin{lemma}\label{lem:foreground-tidy}
  The foreground solution  $\bF$ is $0.5$-tidy at the end of every time-step.
\end{lemma}
\begin{proof}
    We prove by induction on time-step $t$. As the base case, the initial greedy solution is 0-tidy.
    For the inductive case, it suffices to prove the following claim for every level $k\in[0,\ell_{\max}]$: Suppose $\bF$ is $0.2$-tidy at level $k$ just before $T_k$ starts. Then, $\bF$ is $0.2$-tidy at level $k$ when $T_k$ terminates, and $\bF$ is $0.5$-tidy at level $k$ throughout the lifetime of $T_k$.

    We now prove the claim. Let $t_1$ be the last time-step before $T_k$ starts, and $t_2$ be the time-step that $T_k$ terminates. 
    At $t_2$, $T_k$ may normally terminate or aborted by a higher thread. In both cases, the normal termination of $T_k$ or a higher thread will replace levels $\le k$ of $\bF$ by the corresponding background solution. By \Cref{lem:background-tidy}, the background solution (which is identical to the extended background solution at normal termination) is $0.1$-tidy at level $k$. This establishes the first part of the claim.

    For the second part of the claim, we have $t-t_1 \le 0.2\cdot |L_k(t)|$ for every $t\in[t_1+1, t_2]$ by \Cref{lem:overall-lifetime}. So for every $t\in[t_1+1, t_2]$,
    \begin{align*}
        \frac{|P_k(t)|}{|A_k(t)|} & \le \frac{|P_k(t_1)|+(t-t_1)}{|A_k(t)|}
        \le \frac{0.2|A_k(t_1)|+(t-t_1)}{|A_k(t)|}\\
        &\le \frac{0.2(|A_k(t)|+(t-t_1))+(t-t_1)}{|A_k(t)|} \le 0.2 + 1.2\cdot \frac{0.2|L_k(t)|}{|A_k(t)|} \le 0.44
    \end{align*}
    Here, the last inequality uses $|L^\bF_k|\subseteq |A^\bF_k|$, which holds because all dormant elements in $\cov(\bF)$ are universally passive.
    This establishes the second part of the claim.
\end{proof}

\begin{lemma}\label{lem:background-stable}
    Consider any background thread $T_k$ that 
    starts (i.e., enters preparation phase) at time-step $\tprep$ and terminates at time-step $\tend > \tprep$. Suppose that the foreground solution $\bF$ is $\epsilon$-stable at time-step $\tprep$. Then, the  background solution $\bB_k$ is also $\epsilon$-stable at every time-step $t \in [\tprep, \tend]$.
\end{lemma}
\begin{proof}
To prove the lemma, we need to show that
\begin{equation}
\label{eq:extend:2}
\left| A_j^{\bB_k}(t) \cap s \right| < (1+\epsilon)^{j+1} \text{ for all sets } s \in \mathcal{S}, \text{ time-steps } t \in [\tprep, \tend] \text{ and levels } j. 
\end{equation}

As the thread $T_k$ never assigns an element to a level $> (k+1)$, we always have $A_j^{\bB_k} = A_{k+1}^{\bB_k}$ for all $j > (k+1)$. Accordingly, we only need to prove \Cref{eq:extend:2} for levels $j \leq k+1$. 

We first focus on the case where $j = k+1$. 
We start by observing that $A_{k+1}^{\tmB_k}(\tprep) \subseteq A_{k}^{\bF}(\tprep)$.  Since $\bF$ is $\epsilon$-stable (see \Cref{def:stable}) at time-step $\tprep$, we get
\begin{equation}
\label{eq:extend:200}
\left| A_{k+1}^{\tmB_k}(\tprep
) \cap s \right|\leq \left| A_{k}^{\bF}(\tprep
) \cap s \right|  < (1+\epsilon)^{k+1} \text{ for all sets } s \in \mathcal{S}.
\end{equation}
Subsequently, the thread $T_k$ changes (resp.~create) the passive level of an element only when it gets deleted (resp.~inserted). Furthermore, \Cref{alg:deletion} and \Cref{alg:insertion} ensure that $T_k$ sets the passive levels of the elements being inserted/deleted to $\leq  (k+1)$. Thus, we have $A_{k+1}^{\tmB_k}(\tprep) \supseteq A_{k+1}^{\tmB_k}(t)$ at each time-step $t \in [\tprep, \tend]$, and so \Cref{eq:extend:200} implies that 
\begin{equation}
\label{eq:extend:201}
\left| A_{k+1}^{\tmB_k}(t) \cap s \right| < (1+\epsilon)^{k+1} \text{ for all sets } s \in \mathcal{S} \text{ and all } t \in [\tprep, \tend].
\end{equation}
Since we always have $A_j^{\bB_k} \subseteq A_j^{\tmB_k}$, \Cref{eq:extend:201} implies that $\bB_k$ satisfies \Cref{eq:extend:2} for level $j = k+1$.

For the rest of the proof, we focus on a level $j \leq k$. Let $t_j \in [\tprep, \tend]$ be the first (i.e., smallest) time-step at which the copy pointer $p_k$ of the thread $T_k$ becomes equal to $j$ (see \Cref{fact:pointer-monotone}). Otherwise, if $p_k(t) > j$ for all  $t \in [\tprep, \tend]$, then we define $t_j := \tend$. Now, we consider two possible cases, depending on the value of $t$. 

\medskip
\noindent {\em Case 1: $t < t_j$.} In this case, we have $A_j^{\bB_k}(t) = \emptyset$, and so \Cref{eq:extend:2} trivially holds.

\medskip
\noindent {\em Case 2: $t \geq t_j$.} Here, we consider two sub-cases, depending on the status of the set $s$. First, suppose that $s \in \bB_k$ and $\lev^{\bB_k}(s) \ge j+1$. In this sub-case, every element $e \in s$ is assigned a level $\lev^{\bB_k}(e) \geq j+1$, and this level cannot change in future. Furthermore, whenever an element $e$ that belongs to $s$ gets inserted at some future time-step $t' \geq t_j$, the thread $T_k$ would assign $e$ to a level $\geq j+1$ (see \Cref{alg:insertion}). Thus, we have $A_j^{\bB_k}(t) \cap s = \emptyset$, and so \Cref{eq:extend:2} trivially holds. 

Finally, consider the remaining sub-case that $s\notin \bB_k$ or $\lev^{\bB_k}(s) \le j$.
Let $s_j$ be the set that causes $p_k$ to drop to $\le j$ in the greedy algorithm at time-step $t_j$.
In this case, $s$ is not added to $\bB_k$ before $s_j$ is selected by the greedy algorithm.
According to \Cref{alg:greedy}, when the greedy algorithm selects $s_j$ with $\lev(s_j) \le j$, it cannot find any set with marginal coverage $\geq (1+\epsilon)^{j+1}$. 
This implies that
\begin{equation}
\label{eq:extend:500}
\left| s \cap \left(L^\bF_k(t_j) \setminus \cov(\bB_k(t_j) \right) \right| < (1+\epsilon)^{j+1}. 
\end{equation}
Now, consider any element $e \in A_j^{\bB_k}(t)$. Clearly, the element $e$ either  belongs to $L^\bF_k(t_j) \setminus \cov(\bB_k(t_j))$, or it gets inserted at some time-step $t' \in [t_j, t]$. In the latter scenario, however, the thread $T_k$ would either make the
element $e$ universally passive upon insertion, or set $\plev^{\bB_k}(e)$ to $p_k(t') \leq j$. This implies that 
\begin{equation}
\label{eq:extend:501}
A_j^{\bB_k}(t) \subseteq L^\bF_k(t_j) \setminus \cov(\bB_k(t_j)).
\end{equation}
From the preceding discussion, we derive that
\begin{align*}
\left| A_j^{\bB_k}(t) \cap s \right| & \leq \left| s \cap \left(L^\bF_k(t_j) \setminus \cov(\bB_k(t_j) \right) \right|  & (\text{\Cref{eq:extend:501}}) \\
& < (1+\epsilon)^{j+1} & (\text{\Cref{eq:extend:500}})
\end{align*}
In other words, \Cref{eq:extend:2} holds. This concludes the proof of the lemma.
\end{proof}

\begin{lemma}\label{lem:foreground-stable}
   The foreground solution $\bF$ is always $\eps$-stable.
\end{lemma}
\begin{proof}
    The foreground solution $\bF$ can change only because of one of the following three events.
    \begin{itemize}
    \item (i) Insertion of an element.
    \item (ii) Deletion of an element.
    \item (iii) A background thread $T_k$ normally terminates, and accordingly we remove all the sets in $\bF$ at levels $\leq k$, and replace them with the sets in $\bB_k$. We refer to this event as a {\em background-switch by thread $T_k$}.
    \end{itemize}
    During event (i), the element being inserted is made universally passive in $\bF$. This does not affect the active coverages. On the other hand, during event (ii), the element being deleted becomes universally passive. This may remove the element from some active coverages, which can only make the stable property weaker. So, the foreground solution $\bF$ can never cease to be $\epsilon$-stable due to the insertion or deletion of an element. 
    
    For the rest of the proof, we focus on the background-switch by a thread $T_k$ at (say) time-step $t$. Let $\tprep < t$ be the time-step at which the thread $T_k$ gets created (i.e., enters the preparation phase).
    (By the rule of one phase per time-step, $T_k$ starts to work at the subsequent time-step after it gets created at $\tprep$. So, $\tprep < t$ even if $T_k$ terminates in one time-step by the rule of base threads.)
    We perform an induction over time-steps. By our induction hypothesis, the foreground solution $\bF$ is $\epsilon$-stable at all time-steps $< t$. Since $\tprep < t$, by \Cref{lem:background-stable} the background solution $\bB_k$ is  $\epsilon$-stable at time-step $t$. Furthermore,  \Cref{fact:background-feasible} guarantees that $\cov(\bB_k) \cap L = L_k^{\bF}$ just before the background-switch by thread $T_k$ at time-step $t$. So, immediately after the background-switch by thread $T_k$ at time-step $t$, we continue to have 
    \begin{equation}
    \label{eq:wrapup:1}
    \left| A_j^{\bF} \cap s \right| = \left| A_j^{\bB_k} \cap s \right| < (1+\epsilon)^{j+1} \text{ for all sets } s \in \mathcal{S} \text{ and levels } j \leq k. 
    \end{equation}
   The background thread $T_k$ sets the levels of elements in $L_k^{\bF}$ to be $\leq k+1$.
   For the passive levels, we have the following claim before the switch: $T_k$ sets $\plev^{\bB_k}(e) > k+1$ if and only if $\plev^\bF(e) > k+1$ and $e\in L^\bF_k$. This is because every  element $e$ in the initial sub-universe is assigned $\plev^{\bB_k}(e)  = \max\{\plev^\bF(e), k+1\}$, and every new elements $e$ inserted to $L^\bF_k$ is assigned  $\plev^{\bB_k}(e) \le k+1$.
   From the claim, we have that for all levels $j > k+1$, the background-switch by thread $T_k$ does {\em not}  change  $A_j^{\bF}$. Thus, applying our induction hypothesis, immediately after the background-switch by thread $T_k$ at time-step $t$ we continue to have 
    \begin{equation}
    \label{eq:wrapup:2}
    \left| A_j^{\bF} \cap s \right| < (1+\epsilon)^{j+1} \text{ for all sets } s \in \mathcal{S} \text{ and levels } j > k+1. 
    \end{equation}
    Finally, we focus on level $j = k+1$. The background-switch by thread $T_k$ inserts some sets into level $k+1$, thereby bringing up the levels of some elements from $\leq k$ to $k+1$.
    For every element $e$ that belongs to one of these sets, however, 
    whether $\plev^\bF(e)>k+1$ is not changed by the switch according to the claim.
    This ensures that $A_{k+1}^{\bF}$ also does {\em not} change due to the background-switch by thread $T_k$ at time-step $t$. Thus, by our induction hypothesis, immediately after the switch we continue to have
    \begin{equation}
    \label{eq:wrapup:3}
    \left| A_{k+1}^{\bF} \cap s \right| < (1+\epsilon)^{k+2} \text{ for all sets } s \in \mathcal{S}. 
    \end{equation}
    The lemma now follows from \Cref{eq:wrapup:1,eq:wrapup:2,eq:wrapup:3}.
    \end{proof}

\begin{lemma}\label{lem:foreground-approx}
    The foreground solution is  an $O(\log n)$-approximate minimum set cover.
\end{lemma}
\begin{proof}
    Combine \Cref{lem:foreground-tidy,lem:foreground-stable} and apply \Cref{lem:approx-factor-full}.
\end{proof}
\begin{lemma}\label{lem:background-approx}
    For any running background thread $T_k$, the size of extended background solution solution $|\tmB_k|\le O(\log n)\cdot \OPT(L)$.
\end{lemma}
\begin{proof}
    From \Cref{lem:background-tidy}, $\tmB_k$ is $\eps$-tidy at levels $\le k$, but might be $\eps$-dirty at levels $> k$. 
    Consider modifying the greedy algorithm of $T_k$ by initializing $p_k$ to $\ell_{\max}+1$ instead of $k+1$. This may change the definition of levels and passive levels in $\bB_k$. Nevertheless, notice that the modification does not affect the choice of sets in the greedy algorithm. So, we can obtain the same sets as $\tmB_k$. Now, we can apply the same argument as in \Cref{lem:background-tidy,lem:background-stable} to show that $\tmB_k$ is $\eps$-tidy and $\eps$-stable under modified levels and passive levels. By \Cref{lem:approx-factor-full}, this implies $|\tmB_k|\le O(\log n)\cdot \OPT(L)$.
\end{proof}

\begin{lemma}\label{lem:approx-factor}
    The output solution is an $O(\log n)$-approximate minimum set cover.
\end{lemma}
\begin{proof}
    The output solution is union of foreground solution and background solutions of threads in copy or tail phase.
    By combining \Cref{lem:foreground-approx,lem:background-approx}, we have that $|\bF|\le O(\log n)\cdot \OPT(L)$, and $ |\bB_k|\le |\tmB_k| \le O(\log n)\cdot \OPT(L)$ for each thread in copy or tail phase. However, we are not done yet because there might be $O(\log n)$ such background threads.

    Suppose there are $r$ background threads $T_{k_1}, T_{k_2}, \ldots T_{k_r}$ in the copy or tail phase, ordered by when they entered the copy phase.     By the criteria to enter the copy phase, we have $\tausus_{k_{i+1}} \le \frac 12 \tausus_{{k_i}}$ for every $i< r$.    That is, $\{\tausus_{{k_i}}\}_i$ forms a geometric series with ratio $\le \frac 12$. 
    We claim that for any thread $T_k$ in the copy or tail phase, $\tausus_k\le |\bB_k| \le 2.2\cdot \tausus_k$. It follows that $\sum_{i=1}^r |\bB_{k_i}| \le O(1) \cdot |\bB_{k_1}|\le O(\log n)\cdot \OPT(L)$.

    Suppose $T_k$ enters suspension phase at time-step $\tsus$, and let $t$ be the time-step indicated by the claim. The lower bound of the claim is simply because $\tausus_k = |\bB_k(\tsus)|$, and $T_k$ cannot remove sets from $\bB_k$ after $\tsus$. Next, we prove the upper bound of the claim. By the criteria to enter suspension phase, $|L^\bF_k(\tsus) \setminus \cov(\bB_k(\tsus))|\le \tausus_k$.
    By \Cref{cor:copy-tail-phase-time,lm:wait:time}, $t-\tsus \le 0.2\cdot \tausus_k$. After $\tsus$, $T_k$ can only add sets to $\bB_k$ through the greedy algorithm during the tail phase. Each set added by the greedy algorithm must cover an element in the uncovered sub-universe. Such an element is either in $L^\bF_k(\tsus) \setminus \cov(\bB_k(\tsus))$, or inserted during $[\tsus+1, t]$. So, the number of sets added to $\bB_k$ from $\tsus$ to $t$ is at most $1.2\cdot \tausus$. The claim follows.
\end{proof}

\subsection{Implementation Details and Update Time}\label{sec:update-time}

First, we describe the data structures used to maintain the foreground, background, and buffer solutions.

\paragraph{Data Structures.}

The algorithm maintains three types of data structures:
\begin{enumerate}
    \item For each solution $S$ maintained by the algorithm (i.e.\ foreground, background and buffer solutions), its collection of sets are organized by levels, and the sets in each level are stored in a balanced binary search tree. We call these set BSTs and use $\Tset(S, \ell)$ to denote the set BST for solution $S$ at level $\ell$.
    The algorithm maintains a global pointer to the set BST for each level in every solution.
    \item For each hierarchical solution $S$ maintained by the algorithm (i.e.\ foreground and background solutions),  its coverage $\cov(S)$ is organized by levels, and the elements in each level of $\cov(S)$ are stored in two balanced binary search trees that store the alive and dormant elements respectively. We call these element BSTs and denote it $\Telem(S, \ell)$ for solution $S$ at level $\ell$. A node in an element BST stores the identity of an element $e$, its level $\lev^S(e)$, passive level $\plev^S(e)$, and the set whose coverage it belongs to in $S$, i.e. $s\in S$ such that $e \in \cov^S(s)$. Similar to set BSTs, the algorithm maintains a global pointer to the element BST for each level in every hierarchical solution.
    \item As mentioned in \Cref{sec:description}, we also have a priority queue $Q_k$ for each background thread $T_k$. The purpose of this priority queue is to facilitate efficient execution of a step of the greedy algorithm in the computation and tail phases of $T_k$. In particular, every node in the priority queue is a set $s$ and its key $Q_k(s)$ is the number of uncovered elements in $s$, i.e. 
    \[
        Q_k(s) = |s\cap (L^{\bF}_k\setminus \cov(\bB_k))|.
    \]        
    We remark that $Q_k$ only stores sets $s$ with key value $Q_k(s) > 0$. This ensures that the size of the priority queue $Q_k$ is at most the maximum frequency $f$ times the size of uncovered sub-universe $|L^{\bF}_k\setminus \cov(\bB_k)|$.

    We insert or delete elements from the uncovered sub-universe in $T_k$ implicitly by updating the keys in $Q_k$. To insert an element $e$, we add $1$ to $Q_k(s)$ for every set $s$ containing $e$. (If such a set $s$ was not in $Q_k$ prior to this step, we add it to $Q_k$ with key value $1$.) Similarly, to delete an element $e$, we subtract $1$ from $Q_k(s)$ for every set $s$ containing $e$. (If this drops $Q_k(s)$ to $0$, we remove $s$ from $Q_k$.) 
    

    \item Finally, for each background thread $T_k$, we have an array indexed by all sets where for each set $s$, it stores a BST containing all elements in $s$ that are in the uncovered universe of $T_k$, i.e.\ $s\cap (L^{\bF}_k\setminus \cov(\bB_k))$. We call this the uncovered BST and denote it $\Tunc(k, s)$.\footnote{For space efficiency, if we want to avoid storing an explicit array of size $|\cal S|$, then we can also store the global pointers in a BST or a hash table that only indexes the sets that actually appear in any priority queue. We ignore this in the rest of the description.}
    
\end{enumerate}

\paragraph{Update Time Bound.}

We now give details of how the data structures are updated by the algorithm, and derive resulting bounds on the update time of the algorithm. 

\begin{lemma}\label{lem:update-time}
    In any time-step, the foreground thread or a background thread $T_k$ calls $O(f\log n)$ priority queue and BST operations, and runs for $O(f \log n)$ time outside these calls. Therefore, the worst-case update time of the algorithm is $O(f\log^3 n)$.
\end{lemma}
\begin{proof}
In bounding update time, we distinguish between two types of operations done at each time-step. The first type of operation is to handle element insertion or deletion in the instance at the current time-step. This is done for both foreground and background threads. The second type of operation applies only to the background threads, where they execute a set of operations based on the phase they are in. We bound the number of data structural steps for these two types of operations separately.

First, we consider an element insertion or deletion:
\begin{itemize}
    \item[-] Insertion of element $e$: Note that the element belongs to at most $f$ sets.

    We first consider the foreground thread.
    For each set $s$ that contains $e$, we query $s$ in each set BST $\Tset(\bF, \ell)$. These queries reveal the levels of all sets containing $e$ that are already in $\bF$. If there is at least one such set, let $s^*$ be the set at the highest level $\ell^*$ in $\bF$ among all such sets. We add a new node to the element BST $\Telem(\bF, \ell^*)$ containing $e$, $\lev^\bF(e) = \ell^*$, $\plev^\bF(e) = \ell^*$, and $s^*$. This means we are implicitly adding $e$ to $\cov(s^*)$. Next, suppose $e$ is not in any set $s\in \bF$. In this case, we first add an arbitrary set containing $e$ to the set BST $\Tset(\bF, 0)$ at level $0$. Then, we repeat the above steps for adding $e$ to the element BST $\Telem(\bF, 0)$.

    Now, we consider a background thread $T_k$ at level $k\le \lev^{\bF}(e)$. The first step is identical to the foreground thread, i.e. for each set $s$ that contains $e$, we query $s$ in each set BST $\Tset(\bB_k, \ell)$. If there is at least one set containing $e$ in $\bB_k$, then we use the same set of updates as in the foreground solution. If none of the sets in $\bB_k$ contains $e$, then we insert $e$ in $Q_k$ as described in the previous subsection by updating $Q_k(s)$ for $s\ni e$. Correspondingly, we also add $e$ to the uncovered BST $\Tunc(k, s)$ for every $s\ni e$.

    The running time for the operations in any single thread for element insertion is dominated by
    $O(f\log n)$ BST operations for set BSTs, 
    $O(1)$ BST operations for element BSTs,
    $f$ BST operations for uncovered BSTs, and 
    $f$ priority queue updates for element addition to $Q_k$.

    \item[-] Deletion of element $e$:
    For the foreground thread, we query $e$ in the element BSTs $\Telem(\bF, \ell)$ to retrieve $\lev^{\bF}(e)$. Then, we update $\plev^\bF(e)$ in $\Telem(\bF, \lev^\bF(e))$.
    We move the element from the alive element BST to the dormant element BST.

    For each background thread $T_k$, we similarly query the element BSTs $\Telem(\bB_k, \ell)$ to decide whether $e\in \cov(\bB_k)$. If yes, the update is identical to the foreground thread. If no, we delete $e$ from $Q_k$ as described in the previous subsection by updating $Q_k(s)$ for $s\ni e$. Correspondingly, we also remove $e$ from the uncovered BST $\Tunc(k, s)$ for every $s\ni e$.

    A special case is when the deleted element $e$ is in $\cov(s)$ for the set $s$ that is currently being added to $\bB_k$ by $T_k$ in the computation or tail phases. Since the process of updating data structures for adding $s$ to $\bB_k$ does not happen in the single time-step, it is ambiguous as to whether $e\in \cov(\bB_k)$ while $s$ is being added. To avoid ambiguity, we do the following. If $e$ is not in the element BSTs $\Telem(\bB_k, \ell)$, then we query $e$ in the uncovered BST $\Tunc(k, s)$. If yes, then we first add $e$ to $\cov(\bB_k)$ using the steps described below in the computation and tail phases, and then proceed with the deletion of $e$ (which will now change the passive level of $e$ instead of removing it altogether from the solution $\bB_k$).

    The running time for the operations in any single thread for element deletion is dominated by 
    $O(\log n)$ BST operations for element BSTs, 
    $f$ BST operations for uncovered BSTs, and 
    $f$ priority queue updates for element deletion from $Q_k$.
\end{itemize}

We now analyze the operations performed by a background thread $T_k$ based on the phase it is in:
\begin{itemize}

    \item[-] Preparation Phase: 
    Elements in $L^\bF_k$ form the initial uncovered sub-universe for background thread $T_k$ in the preparation phase. In order to retrieve these elements, the algorithm uses the alive element BSTs $\Telem(\bF, \ell)$ for $\ell\le k$. For each element $e$, it is added to the priority queue $Q_k$ as described earlier. Furthermore, the element is added to the uncovered BSTs $\Tunc(k, s)$ for every $s\ni e$. This takes $f$ priority queue and BST operations. Since $O(\log n)$ elements are processed in each time-step, this takes $O(f\log n)$ priority queue and BST operations per time-step.

    Note that the set of elements in $L^\bF_k$ can change due to insertions, deletions, and because of background solutions at levels $< k$ being switched to the foreground. When an element $e$ is inserted or deleted, the update algorithm for insertion or deletion given above is applied both to the foreground thread and the background thread $T_k$. But, if a background thread at level $\ell < k$ is switched to the foreground, it can change the levels of the elements in $L^{\bF}_\ell$ (although it does not change the set of elements in $L^\bF_k$ overall by \Cref{fact:background-feasible}). This can create inconsistencies, e.g., if the preparation phase were handling elements by level in the foreground solution. 

    To avoid such inconsistencies, our goal is to handle elements in order of a fixed index. To do so, the background thread $T_k$ in the preparation phase uses the following strategy to select the next element to process.
    Recall that the process runs in a consecutive time block, so that the foreground data structures are stable.
    It runs queries on all the alive element BSTs $\Telem(\bF, \ell)$ at levels $\ell \le k$ to identify the successor of the last element processed in each BST. Then, it selects the element $e$ with the smallest index among these successors and processes it next. This ensures that elements are processed in order of their index irrespective of their levels in the foreground solution.
    In one time-step, the process finds $\cspd$ elements in $L^\bF_k$, and each element takes $O(\log n)$ BST operations to query all element BSTs. 
    The process takes $O(\log n)$ BST operations.


    \item[-] Computation and Tail Phases:
    In the greedy algorithm, we first retrieve the set $s$ with the largest key in $Q_k$. This is single priority queue operation. First, we insert set $s$ in the set BST $\Tset(\bB_k, \lev(s))$ where $\lev(s)$ is as defined in Line~\ref{line:set-level} in \Cref{alg:greedy}. Next, we retrieve $\cov(s)$ using $\Tunc(k, s)$. For each element $e\in \cov(s)$, we insert it in the element BST $\Telem(\bB_k, \lev(s))$. At the same time, we remove $e$ from $Q_k$ (by updating the key as described earlier) and also remove it from the uncovered BSTs $\Tunc(k, s)$ for every $s\ni e$.
    
    To add set $s$ to $\bB_k$, we process each element $e\in \cov(s)$ as follows. We remove $e$ from the uncovered sub-universe by deleting it from $Q_k$, and removing it from $\Tunc(k, s')$ for every $s\ni e$.
    Then, we insert $e$ to $\Telem(\bB_k, \lev(s))$ with relevant information. ($\lev^{\bB_k}(e)$ is set to $\lev(s)$, while $\plev^{\bB_k}(e)$ is determined before at preparation phase or when $e$ is inserted.) Since the thread processes $\cspd$ elements per time-step, this is $O(f)$ priority queue operations and $O(1)$ BST operations.

    \item[-] Termination:
    At the normal termination of a background thread $T_k$, the switch between $\bB_k$ and $\bF$ is implemented as follows. 
    We use pointer switch to replace the foreground data structures $\Tset(\bF, \ell)$ and  $\Telem(\bF, \ell)$ at levels $\ell \le k$ by the corresponding background data structures $\Tset(\bB_k, \ell)$ and $\Telem(\bB_k, \ell)$. Besides this, we merge the data structures for level $k+1$, i.e., we merge $\Tset(\bF, k+1)$ with $\Tset(\bB_k, k+1)$, and merge $\Telem(\bF, k+1)$ with $\Telem(\bB_k, k+1)$.
    This is $O(\log n)$ BST operations.\qedhere
\end{itemize}
\end{proof}

The above implementation details and running time analysis work for the goal of bounding worst case insertion recourse.
To bound the deletion recourse as well, we need to deamortize the deletion recourse according to \Cref{lem:new:recourse-reduction}. The garbage collection step can be efficiently implemented as follows. When we switch out some set BST from the foreground solution, we move that to a garbage set. The algorithm performs an additional garbage collection step that removes $O(\log n)$ sets in the garbage set from the output.

\begin{lemma}\label{lm:full:reduction}
    The deamortized algorithm has worst-case update time $O(f\log^3 n)$.
\end{lemma}
\begin{proof}
    Given \Cref{lem:update-time}, it remains to prove that the garbage collection step described in \Cref{lem:new:recourse-reduction} can be efficiently implemented.
    We construct a garbage BST to manage the garbage sets in the output solution.
    The algorithm can only remove sets from the output solution by switching at the normal termination of a background thread. This step is implemented by pointer switch of BSTs, so we can merge the BSTs of sets that are switched out to the garbage BST in $O(\log n)$ BST operations.
    At the end of each time-step, we perform the garbage collection step by removing $O(\log n)$ sets from the garbage BST (and the output solution) as required by \Cref{lem:new:recourse-reduction}, which takes $O(\log n)$ BST operations.
\end{proof}

One issue is that the size of coverage of some solution might exceed $n$, since we defined $n$ to be a universal upper bound for the number of live elements.
Nevertheless, for any hierarchical solution $S$ maintained by the algorithm, $|\cov(S)|$ can be upper bounded by $O(n)$ by the tidy property of $\bF$ and $\tmB_k$. So, all data structures maintained by the algorithm have size at most $O(n)$. It follows that each BST or priority queue operation can be implemented in $O(n)$ time. This also validates our claim that $\ell_{\max} = O(\log n)$ when defining $\ell_{\max}$.

\newpage

\part{$O(f)$-Competitive Algorithm}
\label{part:f}

\paragraph{Remark.} The presentation in \Cref{part:f} is self-contained. In particular, we use slightly different notations and terminologies in \Cref{part:f} than the ones introduced in \Cref{sec:prelim} and \Cref{part:lnn}.

\section{Preliminaries}\label{sec:prelim-f}

In the Dynamic Set Cover problem, we are given a (fixed) universe $U$ of elements and a collection of sets $\calS$ that cover $U$. At any time-step $t\ge 0$, there is a set of {\em live} elements $L^{(t)}$ that have to be covered by the algorithm. Elements that are not live are called {\em dormant} and denoted $D^{(t)}:=U\setminus L^{(t)}$. At any time-step, at most one element might change the state of live and dormant. We view such a change as an insertion or deletion to the set $L^{(t)}$.
We denote $n$ to be an upper bound of $|L^{(t)}|$, the number of elements to be covered. We denote $f$ to be an upper bound of the frequency of elements, that is  the number of sets containing it.

A solution is a collection of sets in a given space $\calS$. At time-step $t$, a solution $S^{(t)}\subseteq \calS$ is {\em feasible} if it covers $L^{(t)}$.
We consider the unweighted setting, that is the cost of every set is 1.

This part is devoted to the following theorem.
\begin{theorem}\label{thm:f-main}
    There is a deterministic algorithm for the dynamic set cover problem that maintains an $O(f)$-approximate solution with worst-case recourse of $O(\log n)$ and worst-case update time of $O(f\log^3 n)$.
\end{theorem}

\begin{define}[Primal-dual solutions]
    Recall the LP relaxation of unweighted set cover on universe $L$:
    \begin{align*}
        (P) ~ & \min \sum_{s\in \calS} x_s
        & (D) ~ & \max \sum_{e\in L} y_e\\
        \text{s.t. } & \forall e\in L, \sum_{s\in \calS: e\in s} x_s \ge 1
        &\text{s.t. } & \forall s\in \calS, \sum_{e\in s} y_e \le 1\\
        & x\ge 0 & & y\ge 0
    \end{align*}

    A dual solution defines a dual value $w_e\ge 0$ for all live elements $e\in L^{(t)}$.
    Although we consider the dual LP on universe $L^{(t)}$, it will be convenient to define dual values for some dormant elements in $D^{(t)}$ as well.
    Such a dual solution is called an {\em extended dual solution}. (Let $w_e=0$ by default for other dormant elements.)

    We use $w_s:=\sum_{e\in s} w_e$ to denote the total dual value in a set $s$. 
    As part of the definition of extended dual solution, we require the following invariant:
    \begin{itemize}
        \item (Dual feasibility invariant) $\forall s\in \calS, w_s\le 1$.
    \end{itemize}
    We remark that this is stronger than the constraints in dual LP because the sum includes extended dual value of dormant elements.
    
    A set $s\in \calS$ is said to be tight if $w_s \ge (1+\eps)^{-1}$.
    Given an extended dual solution, we define its corresponding primal solution as the collection of all tight sets.

    In our algorithm, a primal-dual solution $\bF$ consists of two parts:
    \begin{enumerate}
        \item An extended dual solution $\{w_e(\bF)\}_{e\in U(\bF)}$. Here, $U(\bF)\subseteq U$ is the sub-universe of relevant element of $\bF$. The dual values are explicitly maintained for $U(\bF)$, and are regarded as 0 for other elements.
        \item A subset of its corresponding primal solution denoted by $\calS(\bF)$.  We remark that the primal solution is not necessarily feasible for the primal LP, but we require the following invariant suggested by complementary slackness.
    \end{enumerate}
    \begin{itemize}
        \item (Tight set invariant) $\forall s\in \calS(\bF), w_s(\bF) \ge (1+\eps)^{-1}$.
    \end{itemize}
\end{define}

\begin{lemma}\label{lem:f-approx-general}
    If a feasible extended dual solution $\bF$ satisfies $w_U(\bF) \le (1+\eps)w_L(\bF)$, then the corresponding primal solution is $((1+\eps)^2 f)$-competitive.
\end{lemma}
\begin{proof}[Proof sketch]
\[|\calS(\bF)| \le \sum_{s\in \calS(\bF)} (1+\eps)w_s \le (1+\eps)\sum_{e\in U}\sum_{s\in \calS(\bF):e\in s} w_e \le (1+\eps)f w_U \le  (1+\eps)^2 f\cdot w_L \le (1+\eps)^2 f\cdot OPT\]
\end{proof}

\begin{define}[Hierarchical solutions]
     Based on a primal-dual solution $\bF$, we assign a level $\lev_s(\bF) \in [0, \lmax]$ to every set $s\in \calS(\bF)$, and $\lev_e(\bF)\in [0, \lmax]$ for some elements in $U(\bF)$.\footnote{The notation $[0,\lmax]$ refers to the set of integers in the interval.} Denote $C(\bF)\subseteq U(\bF)$ to be the set of elements whose level is defined in $\bF$ (can be live or dormant).
     Denote $E(\bF):=U(\bF)\setminus C(\bF)$ to be the elements in the sub-universe whose levels are not fixed. We call $C(\bF)$ covered elements and $E(\bF)$ exposed elements.

    As part of the definition of hierarchical solutions, we require the following invariants.
    First, the levels of elements are related to the dual values.
    \begin{itemize}
        \item (Level invariant) $\forall e\in C(\bF), w_e(\bF) \le (1+\eps)^{-\lev_e(\bF)}$.
    \end{itemize}
    Second, a covered element $e\in C(\bF)$ must be covered by a tight set $s\in \calS(\bF)$. (The converse is not true. It is possible that $e\in s, s\in \calS(\bF)$, but $e$ is exposed.) Moreover, it must be covered by the highest one. 
    \begin{itemize}
        \item (Highest level invariant) $\forall e\in C(\bF), \lev_e(\bF) := \max\{\lev_s(\bF):s\in \calS(\bF), e\in s\}$.
    \end{itemize}

    An element $e\in C(\bF)$ is called {\em active} if $w_e(\bF) = (1+\eps)^{-\lev_e(\bF)}$, and {\em passive} if $w_e(\bF) < (1+\eps)^{-\lev_e(\bF)}$.
\end{define}

    We use a subscript $k$ on many notations to denote the corresponding sets defined on levels $[0, k]$. Formally, we denote:
    \begin{itemize}
        \item $\calS_k(\bF):=\{s\in \calS(\bF): \lev_s(\bF) \le k\}$, sets in the primal solution at levels $[0, k]$.
        \item $C_k(\bF):=\{e\in C(\bF): \lev_e(\bF)\le k\}$, the set of covered elements at levels $[0, k]$;
        \item  $L_k(\bF):=\{e\in L \cap C_k(\bF) : \lev_e(\bF) \le k\}$,  the set of live elements at levels $[0, k]$;
        \item  $D_k(\bF):= \{e\in D \cap C_k(\bF) :\lev_e(\bF) \le k\}$,  the set of dormant elements at levels $[0, k]$;
        \item $A_k(\bF) := \{e\in L_k(\bF): w_e(\bF) = (1+\eps)^{-\lev_e(\bF)}\}$, the set of active live elements at levels $[0, k]$;
        \item $P_k(\bF) := \{e\in L_k(\bF): w_e(\bF) < (1+\eps)^{-\lev_e(\bF)}\}$, the set of passive live elements at levels $[0, k]$.
    \end{itemize}
    Note that $C_k(\bF) =  L_k(\bF) \uplus D_k(\bF), L_k(\bF) =  A_k(\bF) \uplus P_k(\bF)$.

\begin{define}
    A hierarchical solution $\bF$ is said to be $\eps$-tidy at level $k$ if $|D_k(\bF)| + |P_k(\bF)| \le \eps \cdot (|A_k(\bF)| + |E(\bF)|)$, and $\eps$-dirty at level $k$ otherwise.
    If $\bF$ is $\eps$-tidy at every level $k\in[0,\lmax]$, we say it is $\eps$-tidy.
\end{define}

\section{Fully Dynamic $O(f)$-Competitive Set Cover Algorithm}\label{sec:description-f}The algorithm maintains a set of primal-dual solutions, including a {\em foreground solution} $\bF^{(t)}$, and a set of $\lmax+1$ {\em background solutions} $\bB_k^{(t)}$ for levels $k\in[0,\lmax]$.
Intuitively, we expect $\bF^{(t)}$ to be a feasible primal-dual solution where majority of dual values are assigned to live elements, so that \Cref{lem:f-approx-general} guarantees good approximation. We can handle insertion and deletion by local fix (see \Cref{sec:foreground-thread}). But, doing so may accumulate dual values in dormant elements, as the local fix algorithm need to keep dormant elements to satisfy tight set invariant.
To clean up the dormant elements, we repeatedly run a rebuild procedure in parallel for each level $k$. The procedure rebuilds the levels $[0, k]$ of $\bF^{(t)}$ to form the background solution $\bB_k$. 
Since $\bB_k$ is identical to $\bF$ at levels $> k$, the algorithm only needs to maintain a partial solution of $\bB_k$ at levels $[0, k]$.  We denote $\Sstar(\bB_k)\subseteq \calS(\bB_k)$ to be the sets in the partial solution.
Ideally, the procedure can remove all dormant elements in $D_k(\bF)$ and raise all remaining passive elements in $P_k(\bF)$ to level $k+1$. After the rebuild finishes, the algorithm performs a switch step which replaces $\bF$ with the new solution $\bB_k$. Then, it aborts the rebuild procedures on lower levels, since they are based on an outdated version of $\bF$.

One issue here is that the switch step can replace a large subset of $\calS(\bF)$, yet must be executed in one time-step. To avoid large recourse, we have to gradually add $\Sstar(\bB_k)$ to the output solution before the switch step. It is tempting to include all background solutions in the output, but doing so will harm the approximation factor since there can be $O(\log n)$ background solutions. To resolve this, we only include a subset of background solutions in the output solution, and only allow these included background solutions to switch.

Algorithmically, we use a set of {\em buffer solutions} $\calR_k$ for $k\in[0,\lmax]$  to store copies of $\Sstar(\bB_k)$  before the potential switch steps.
The output solution consists of the foreground solution and all buffer solutions.
We remark that the buffer solutions are not primal-dual solutions; they only store  primal solutions in order to control recourse. The switch steps still happen between the foreground and background solutions.

\paragraph{Initialization.} The algorithm initializes the foreground, background, and buffer solutions to be all empty. (We assume $L^{(0)}=\emptyset$.)

\paragraph{Multi-threading.}
It will be convenient to view the algorithm as independent threads. The threads include:
\begin{enumerate}
    \item a foreground thread that handles updates in the foreground solution, and
    \item $\lmax+1$ background threads $T_k$, one for each level $k$, that gradually execute a primal-dual rebuild algorithm on the sub-universe $L_k(\bF^{(t)})$.
    The background threads also need to handle updates in the set $L_k(\bF^{(t)})$  because of changes in $\bF^{(t)}$. The actual sub-universe involved in the background thread is $L_k(\bF^{(t_0)})$ at the time $t_0$ when $T_k$ is created, and all elements inserted to levels $\le k$ of $\bF$ during the lifetime of $T_k$. 
\end{enumerate}

The foreground solution can only be modified by the foreground thread, or by the switch step at the termination of a background thread.

\paragraph{Sequential Behavior.}
Although we describe the algorithm as independent threads, the actual behavior of the algorithm within a time-step is sequential. A time-step is organized as follows:
\begin{enumerate}
    \item The algorithm receives the update (insertion or deletion).
    \item The foreground thread executes.
    \item The background threads execute one by one, from higher level to lower level. If a background thread is not running, the algorithm creates that thread.
    If a background thread terminates, the algorithm processes the switch triggered by the termination, and aborts all lower threads. The first thread is said to normally terminate.
\end{enumerate}

\subsection{Foreground Thread}\label{sec:foreground-thread}

On deletion of element $e$, it becomes dormant, but we keep its dual value in the foreground solution. The solution does not change.

On insertion of element $e$, let $s$ be the highest set in $\calS(\bF)$ covering $e$, i.e. $s=\arg\max_{s\in \calS(\bF), e\in s} \lev_s(\bF)$.
If such a set $s$ exists, we assign $e$ to $s$ by setting $\lev_e(\bF)=\lev_s(\bF)$ and $w_e(\bF)=0$.
Otherwise, $e$ is not covered by $\bF$. In this case, we pick the set $s$ with smallest dual slackness $1-w_s$ among the sets that cover $e$. We set $w_e(\bF)=1-w_s(\bF)$, so that $s$ becomes tight and dual feasibility is preserved. We add $s$ to $\calS(\bF)$ at level 0.

\subsection{Background Threads}

The background thread intends to run a primal-dual rebuild described in \Cref{alg:f-rebuild} at the pace of processing $\cspd$ elements per time-step to rebuild the sub-universe.
Before the rebuild, the thread performs a initialization procedure described in \Cref{alg:f-prepare}. This is called the preparation phase of the thread.
During the execution of \Cref{alg:f-prepare,alg:f-rebuild}, the sub-universe might be updated due to the foreground thread. 
Such an update must be insertion or deletion of an element in the sub-universe.
The updates are handled by \Cref{alg:f-insert,alg:f-delete} after the preparation phase. (The behavior in preparation phase will be described later.)
Within a time-step, a background thread first handles the update, then continues the rebuild process.

The algorithm runs primal-dual algorithm by lazily raising the dual value of the exposed elements in sub-universe. Denote $p_k$ to be a level pointer that decreases from $k+1$ to 0. Conceptually, all exposed elements in $E(\bB_k)$ have dual values raised to $(1+\eps)^{-p_k}$. Whenever a set becomes tight during this process, we place the set and all exposed element in the set at level $p_k$. By doing so, the elements become covered and their dual values are fixed.


\paragraph{Data Structures.}
The background solution $\bB_k$ should be viewed as a complete primal-dual solution that is identical to $\bF$ for elements with $\lev_e(\bF) \ge k+1$. However, the algorithms only maintain a partial solution that differs from $\bF$. Denote $\Ustar(\bB_k)\subseteq U(\bB_k)$ to be a sub-universe of elements that we explicitly maintain. \footnote{Precisely speaking, the sub-universe consists of elements in the initial sub-universe $L_k(\bF^{(t_0)})$ at the time-step $t_0$ when $T_k$ is created, and elements inserted to $L_k(\bF)$ after $t_0$. In addition, deleted elements may be removed from the sub-universe.}
Denote $\Cstar(\bB_k)$ to be elements in $\Ustar(\bB_k)$ with level defined, i.e., $\Cstar(\bB_k) := \Ustar(\bB_k)\cap C(\bB_k)$.
For each element $e\in \Ustar(\bB_k)$, we maintain its dual value $w_e(\bB_k)$; For each element $e\in \Cstar(\bB_k)$, we maintain its level $\lev_e(\bB_k)$.
The information of elements outside the sub-universe will only be accessed aggregately when the algorithm queries the total dual value of a set $s\in \calS$. Such queries can be efficiently retrieved from data structure of $\bF$.

We denote $\Sstar(\bB_k)$ to be the sets in the partial solution. Sets in $\Sstar(\bB_k)$ are tight sets of $\bB_k$ at levels $\le k+1$, and they can be combined with $\calS(\bF)\setminus \calS_k(\bF)$ to form $\calS(\bB_k)$.  These sets are explicitly maintained by the background thread $T_k$ 

The exposed elements $E(\bB_k)$ are viewed to participate an ongoing primal-dual algorithm, so that their dual values are regarded as $(1+\eps)^{-p_k}$ but their levels are not decided yet.
The algorithm explicitly maintains $E(\bB_k)$ and guarantee $E(\bB_k)= \Ustar(\bB_k)\setminus \Cstar(\bB_k)$.
Denote $\Sin(E)$ to be the collection of sets incident to $E$. The algorithm explicitly maintains $\Sin(E(\bB_k))$ in the sense that for each insertion and deletion of an element $e$ in $E(\bB_k)$, the algorithm updates $\Sin(E(\bB_k))$ accordingly by checking all sets incident to $e$.

The total dual value $w_s(\bB_k)$ of a set $s\in \calS$ consists of three parts:
\begin{enumerate}
    \item Elements at levels $\ge k+1$ of $\bF$ contribute their dual values in $\bF$.
    \item Elements in $\Cstar(\bB_k)$ contribute their dual values in $\bB_k$.
    \item Exposed elements contribute $|s\cap E(\bB_k)|\cdot (1+\eps)^{-p_k}$. This part is maintained lazily in the sense that the algorithm maintains $|s\cap E(\bB_k)|$ in appropriate data structures for each set $s\in \Sin(E(\bB_k))$.
\end{enumerate}
We maintain $w_s(\bB_k)$ for sets $s\in \Sin(\Ustar(\bB_k))$ explicitly, in the sense that an update on $w_e$ or whether $e\in E(\bB_k)$ triggers update on $w_s$ for all incident sets $s$.
For other sets $s\notin \Sin(\Ustar(\bB_k))$, we have $w_s(\bB_k)=w_s(\bF)$, and we can retrieve its total dual value from the foreground data structure.

\begin{algorithm}
    \caption{Initialization of primal-dual algorithm run by background thread $T_k$} \label{alg:f-prepare}
    $p_k\gets k+1$\;
    \ForEach{$e\in A_k(\bF)$}{
        $w_e(\bB_k) \gets (1+\epsilon)^{-(k+1)}$ \tcp{$e$ is added to $E(\bB_k)$.}
    }
    \ForEach{$e\in P_k(\bF)$}{
        Run \Cref{alg:f-insert} on element $e$.
    }
\end{algorithm}

\begin{algorithm}
    \caption{Primal-dual rebuild algorithm run by background thread $T_k$}\label{alg:f-rebuild}
    \For{$p_k \gets k+1$ downto $0$}{
        Raise $w_e(\bB_k)\gets (1+\eps)^{-p_k}$ for each $e\in E(\bB_k)$\;
        \While{$\exists s \in \Sin(E(\bB_k))$ s.t. $w_s(\bB_k) \ge (1+\eps)^{-1}$}{
        Add $s$ to $\Sstar(\bB_k)$ at level $p_k$\;
        \ForEach{$e\in s\cap E(\bB_k)$}{
            $\lev_e(\bB_k)\gets p_k, w_e(\bB_k)\gets (1+\eps)^{-p_k}$ \tcp{$e$ is moved from $E(\bB_k)$ to $\Cstar(\bB_k)$.}
        }
        \tcp{$s$ is removed from $\Sin(E(\bB_k))$.}
    }
    }
    Switch $\bF$ with $\bB_k$.
\end{algorithm}

\begin{algorithm}
    \caption{Handling deletion of an element $e$ in $L_k(\bF)$}\label{alg:f-delete}
    \eIf{$e\in s$ for some set $s\in \Sstar(\bB_k)$}{
        \If{$e\in E(\bB_k)$}{
            $\lev_e(\bB_k)\gets p_k, w_e(\bB_k)\gets (1+\eps)^{-p_k}$ \tcp{$e$ is moved from $E(\bB_k)$ to $\Cstar(\bB_k)$.}
        }
        $e$ becomes dormant in $\bB_k$ \tcp{$w_e(\bB_k)$ and $\lev_e(\bB_k)$ do not change.}
    }{
        Remove $e$ from $E(\bB_k)$ and $U^*(\bB_k)$ \tcp{$w_e(\bB_k)$ becomes 0.}
    }
\end{algorithm}
\begin{algorithm}
    \caption{Handling insertion of an element $e$ in $L_k(\bF)$}\label{alg:f-insert}  
    \If{$e\in s$ for some set $s\in \Sstar(\bB_k)$}{
        Let $s$ be the highest set in $\bB_k$ that contains $e$\;
        $\lev_e(\bB_k) \gets \lev_s(\bB_k), w_e(\bB_k)\gets 0$
        \tcp{$e$ is added to $\Cstar(\bB_k)$.}
    }\ElseIf{$1-\max_{s\in \calS: e\in s}\{w_s(\bB_k)\} < (1+\eps)^{-p_k} $}{
        $s\gets \arg\max_{s\in \calS: e\in s}\{w_s(\bB_k)\}$\;
        Add $s$ to $\Sstar(\bB_k)$ at level $p_k$\; 
        $\lev_e(\bB_k)\gets p_k, w_e(\bB_k)\gets 1-w_s(\bB_k)$
        \tcp{$e$ is added to $\Cstar(\bB_k)$.}
    }\Else{
        $w_e(\bB_k) \gets (1+\epsilon)^{-p_k}$ \tcp{$e$ is added to $E(\bB_k)$.}
    }
\end{algorithm}

Next, we describe the phases of a background thread $T_k$.
Let $\cspd$ be a large constant that will be decided later.
Specially, if the sub-universe is small, we do not separate the phases.
\begin{itemize}
    \item (Rule of base threads.) If $|L_k(\bF)|\le \cspd$ when $T_k$ is created, all work of $T_k$ can be finished in one time-step. In this case, we execute \Cref{alg:f-prepare,alg:f-rebuild} in one time-step.
    $T_k$ is not viewed to enter any phase.

    (According to the sequential behavior of threads and since $|L_k(\bF)|$ is monotone, the algorithm only processes the highest basic thread in a time-step.)
\end{itemize}

\paragraph{Preparation Phase.}
In the preparation phase, $T_k$ executes \Cref{alg:f-prepare} at the speed of processing $\cspd$ elements per time-step.

During the preparation phase, $L_k(\bF)$ might be modified by insertion and deletion (after the foreground thread decided that the update affects $L_k(\bF)$).
A deletion in $L_k(\bF)$ is handled by \Cref{alg:f-delete}.
An insertion in $L_k(\bF)$ during the second for loop for passive elements is handled by \Cref{alg:f-insert}.
Specially, on insertion of element $e$ during the first loop for active elements, we view $e$ as a passive element in $P_k(\bF)$ and simply append the new element to the list of second for loop. This special behavior is because we have not initialized the dual values of active elements during the first loop, which is necessary for \Cref{alg:f-insert}.

\paragraph{Computation Phase.}
In the computation phase, $T_k$ executes \Cref{alg:f-rebuild}. But, once $|E(\bB_k)|\le |\Sstar(\bB_k)|$ after finishing a while loop (adding a set), it pauses \Cref{alg:f-rebuild} and enters the suspension phase.

The bottleneck of each while loop is maintaining the data structures for elements in $s\cap E(\bB_k)$. The algorithm processes $\cspd$ elements per time-step.

In computation phase and later phases, updates in $L_k(\bF)$ are handled by \Cref{alg:f-delete,alg:f-insert}.

Specially, if the solution size is small, we do not pause the algorithm.
\begin{itemize}
    \item (Rule of shortcut threads.) If $|E(\bB_k)|\le\cspd$ and $|\Sstar(\bB_k)|\le \cspd$ during computation phase of $T_k$, all remaining work can be done in one time-step. In this case, we continue executing \Cref{alg:f-rebuild} to the end, and copy all sets in $\Sstar(\bB_k)$ to $\calR_k$ in one time-step. $T_k$ is viewed to terminate in computation phase, and not viewed to enter later phases.
\end{itemize}

\paragraph{Suspension Phase.}
When a thread $T_k$ enters the suspension phase, we take a snapshot of its solution size and denote it $\tausus_k$. (We have $\tausus_k>\cspd$
 by the rule of shortcut threads.) 

During the suspension phase, $T_k$ checks whether $\tausus_k\le \frac 12 \cdot \tau$, where $\tau$ is the minimum of $\tausus_j$ among the threads $T_j$ in copy and tail phases. 
If this condition holds, $T_k$ enters copy phase; otherwise, $T_k$ remains in suspension phase.

According to the sequential behavior of background threads, after a higher thread $T_k$ enters the copy phase, lower threads in suspension phase will compare to the updated threshold $\tau=\tausus_k$.

\paragraph{Copy Phase.}
During the copy phase, $T_k$ copies sets in $\Sstar(\bB_k)$ to $\calR_k$ at the speed of $\cspd$ sets per time-step. ($\Sstar(\bB_k)$ may grow by one set per time-step due to \Cref{alg:f-insert}, but cannot shrink. So, the order to copy sets is arbitrary since we can catch up the updates.)

\paragraph{Tail Phase.}
During the tail phase, $T_k$ resumes the primal-dual algorithm. 
In addition, we maintain $\calR_k$ to be a faithful copy of $\Sstar(\bB_k)$. 

If $|E(\bB_k)|\le \cspd$ during tail phase, all remaining work can be done in one time-step. In this case, we continue executing \Cref{alg:f-rebuild} to the end in the current time-step.

\paragraph{Termination.}
When a background thread normally terminates, at the end of \Cref{alg:f-rebuild}, the algorithm switches the $\bF$ with $\bB_k$.
In terms of primal solution, the algorithm removes sets $\calS_k(\bF)$ and moves $\calR_k$ to the foreground, which costs no insertion recourse. In terms of data structure for dual values, the algorithms discards the foreground data structures at levels $\le k$, moves the data structures at levels $\le k$  of $\bB_k$ to foreground, and merges the data structure at level $k+1$ of $\bB_k$ with the foreground data structure at level $k+1$.
All these operations can be done in one time-step.

\section{Analysis of the Fully Dynamic $O(f)$-Competitive Set Cover Algorithm}\label{sec:analysis-f}\subsection{Invariants}


\begin{fact}\label{fact:termination-Ek-empty}
    When a background thread $T_k$ normally terminates, $E(\bB_k)=\emptyset$.
\end{fact}
\begin{proof}
Assume for contradiction that $e\in E(\bB_k)$ at the end of \Cref{alg:f-rebuild}. Then, there exists an set $s$ containing $e$, and $s\in \Sin(E(\bB_k))$ by definition of $\Sin(E(\bB_k))$. Since $w_s(\bB_k)\ge w_e(\bB_k) = (1+\eps)^{-p_k}=1$ when $p_k$ decreases to 0, $s$ should be selected by the while loop to cover $e$, a contradiction.
\end{proof}

\begin{lemma}\label{lem:f-switch-not-change-subuniverse}
    When a background thread $T_k$ normally terminates, the switch step cannot change $L_j(\bF)$ at any level $j \ge k+1$.
\end{lemma}
\begin{proof}
    Assume for contradiction that the earliest violation happens at time-step $t$ for levels $j>k$. Let $t^-, t^+$ respectively denote the time before and after the switch step at termination of $T_j$. Suppose $T_j$ is created at time-step $\tprep$.

    The switch step removes all elements in $L_k(\bF^{(t^-)})$ from $\bF$. We first show that any element $e\in L_k(\bF^{(t^-)})$ is added back to $L_j(\bF^{(t^+)})$. By the assumption that the lemma holds before $t$, $L_j(\bF)$ cannot be modified by the termination of lower-level threads during $[\tprep, t]$. So, $e$ is either in the initial sub-universe $L_j(\bF^{(\tprep)})$ or inserted to the sub-universe after $\tprep$. According to \Cref{alg:f-prepare,alg:f-insert}, at initialization or insertion, $e$ is either assigned $\lev_e(\bB_k) \le k+1 \le j$ or added to $E(\bB_k)$. In the latter case, $e$ must be removed from $E(\bB_k)$ before the termination by \Cref{fact:termination-Ek-empty}. Since $e\in L^{(t)}$, $e$ can only be removed from $E(\bB_k)$ by \Cref{alg:f-rebuild}, and $e$ must be assigned $\lev_e(\bB_k) \le k+1 \le j$ when removed from $E(\bB_k)$. The algorithm cannot modify $\lev_e(\bB_k)$ after it is assigned. So, we have $\lev_e(\bB_k) \le j$ before the switch. Since $\lev_e(\bF) = \lev_e(\bB_k)$ after the switch, we have $e\in L_j(\bF^{(t^+)})$.

    The switch step inserts the partial solution of $\bB_k$ to foreground. Since $E(\bB_k)=\emptyset$ by \Cref{fact:termination-Ek-empty} and $\Ustar(\bB_k)=\Cstar(\bB_k)\cup E(\bB_k)$, the switch step can only add elements in $\Cstar(\bB_k)$ to foreground. Among these elements, the dormant elements in $D^{(t)}$ are not involved in the lemma. We next show that the live elements in $\Cstar(\bB_k)$ must belong to $L_k(\bF^{(t^-)})$ before the switch.
    An element $e$ can be added to $\Cstar(\bB_k)$ in three ways: (1) in \Cref{alg:f-prepare} as a passive element, (2) in the inner loop of \Cref{alg:f-rebuild}, (3) in \Cref{alg:f-insert} as a passive element. Let $t_e\in[\tprep, t]$ be the time-step when $e$ is added to $\Cstar(\bB_k)$. In all three cases, $e$ is either in $L_k(\bF)$ or inserted into $L_k(\bF)$ at $t_e$. By the assumption that the lemma holds before $t$, $L_k(\bF)$ can only be modified by the foreground thread during $[\tprep, t-1]$. This is also true at time-step $t$ before the switch, since there cannot be another normal termination at time-step $t$ by the sequential behavior of threads.  So, $\lev_e(\bF)$ cannot change from $t_e$ to $t^-$, and $e\in L_k(\bF^{(t^-)})$.
\end{proof}

The above lemma implies that during the lifetime of a background thread $T_k$, $L_k(\bF)$ is stable in the sense that $L_k(\bF)$ can only be modified by the foreground thread, and cannot be modified by termination of lower background threads. (It also cannot be modified by higher background threads since termination of higher threads will abort $T_k$.)
Since all modification of $L_k(\bF)$ due to foreground thread are handled by \Cref{alg:f-delete,alg:f-insert}, it is valid to explicitly maintain $E(\bB_k)$.
\begin{lemma}\label{lem:f-Ek-correct}
    For any background thread $T_k$, we have $E(\bB_k)\subseteq L_k(\bF)$, and any element $e\in E(\bB_k)$ cannot belong to any set in $\Sstar(\bB_k)$ at levels $> p_k$.
\end{lemma}
\begin{proof}
    For the first statement, we prove that (1) any element added to $E(\bB_k)$ is from $L_k(\bF)$. This is easy to check in \Cref{alg:f-prepare,alg:f-insert}. (2) Any deletion from $E(\bB_k)\cap L_k(\bF)$ in $L_k(\bF)$ is also removed from $E(\bB_k)$. By \Cref{lem:f-switch-not-change-subuniverse}, $L_k(\bF)$ can only be modified due to insertion and deletion during the lifetime of $T_k$. Whenever an element $e\in E(\bB_k)$ is removed from $L_k(\bF)$ due to a deletion, \Cref{alg:f-delete} will remove $e$ from $E(\bB_k)$. 

    We prove the second statement by induction on time for each element $e$. For the base case, we prove that $e$ does not belong to any set in $\Sstar(\bB_k)$ when added to $E(\bB_k)$. If $e$ is added to $E(\bB_k)$ in the first loop of \Cref{alg:f-prepare}, $\Sstar(\bB_k)$ is empty; If $e$ is added to $E(\bB_k)$ by \Cref{alg:f-insert} or in the second loop of \Cref{alg:f-prepare}, the algorithm guarantees that $e$ does not belong to any set in $\Sstar(\bB_k)$.
    
    For the inductive step, since the levels in $\bB_k$ cannot change once determined, and the algorithm is only allowed to add sets at level $p_k$, the only nontrivial case is when $p_k$ decreases. Assume for contradiction that $e$ is contained in a set $s\in \Sstar(\bB_k)$ at level $p_k$ when the outer loop of $p_k$ finishes in \Cref{alg:f-rebuild}. $s$ can be added to $\Sstar(\bB_k)$ by \Cref{alg:f-rebuild}, or by \Cref{alg:f-prepare,alg:f-insert}. In the former case, every element in $s\cap E(\bB_k)$ should be removed from $E(\bB_k)$, a contradiction. In the latter case, $s$ must be a tight set when added to $\Sstar(\bB_k)$ by \Cref{alg:f-prepare,alg:f-insert}. Since the algorithms cannot decrease $w_s(\bB_k)$, $s$ is still tight at the end of outer loop of $p_k$. Also, $s\in \Sin(E(\bB_k))$ since $e\in E(\bB_k)$. This contradicts the termination condition of while loop in \Cref{alg:f-rebuild}.
\end{proof}

Next, we establish the invariants required in the definition of primal-dual solution and hierarchical solution.

\begin{lemma}[Level invariant]\label{lem:f-level-invariant}
    In the foreground solution and all background solutions, any element $e\in C(\bB_k)$ satisfies $w_e(\bB_k) \le (1+\eps)^{-\lev_e(\bB_k)}$.
\end{lemma}
\begin{proof}
    We prove the invariant by induction on time. Initially the statement is true.

    The foreground solution can be modified in the following cases.
    \begin{enumerate}
        \item[-] Insertion in foreground thread: The algorithm either sets $w_e(\bF)=0$, or sets $\lev_e(\bF)=0$. The invariant holds in both cases.
        \item[-] Deletion in foreground thread: $\lev_e(\bF)$ and $w_e(\bF)$ do not change.
        \item[-] Switch step: $\lev_e(\bF)$ and $w_e(\bF)$ can only be modified by copying from the background solution. The invariant holds by inductive hypothesis for background solutions before the switch.
    \end{enumerate}
    
    For a background thread $T_k$, initially the invariant for elements outside $L_k(\bF)$ follows inductive hypothesis for the foreground solution.
    Whenever $T_k$ determine $\lev_e(\bB_k)$ for an element $e$ (in the inner loop of \Cref{alg:f-rebuild}, or in \Cref{alg:f-prepare,alg:f-insert} as a passive element), we have $\lev_e(\bB_k) = p_k$, $w_e(\bB_k)\le (1+\eps)^{-p_k}$ in all cases.
    The levels and dual values cannot change for elements in $\Cstar(\bB_k)$.
\end{proof}

\begin{lemma}[Highest level invariant]\label{lem:f-highest-level-invariant}
    For any element $e\in C(\bF)$ (resp.\ $\Cstar(\bB_k)$), $\lev_e(\bF)$ (resp.\ $\lev_e(\bB_k)$) is equal to the maximum level of sets in $\calS(\bF)$ (resp.\ $\Sstar(\bB_k)$) that contains it.
\end{lemma}
\begin{proof}
    We prove the invariant by induction on time. Initially the statement is true.

    The foreground solution can be modified in the following cases.
    \begin{enumerate}
        \item[-] Insertion in foreground thread: The algorithm directly guarantees the invariant.
        \item[-] Deletion in foreground thread: The solution does not change.
        \item[-] Switch step at termination of $T_k$: Let $t^-, t^+$ be the time before and after switch. The elements $e$ with $\lev_e(\bF^{(t^-)})\ge k+1$ and the highest set containing $e$ are not modified. Next, consider elements $e\in C_k(\bF^{(t^-)})\cap C(\bF^{(t^+)})$. By inductive hypothesis, $e$ cannot belong to any set at levels $\ge k+1$ of $\bF^{(t^-)}$.
        As argued in \Cref{lem:f-switch-not-change-subuniverse}, $e$ must be covered by a set in $\Sstar(\bB_k^{(t^-)})$ at levels $\le k+1$. 
        By the inductive hypothesis for $\bB_k^{(t^-)}$, $\lev_e(\bB_k^{(t^-)})$ is equal to the highest level of sets covering $e$ in $\Sstar(\bB_k^{(t^-)})$, which is also the highest level of sets covering $e$ in $\calS(\bF^{(t^+)})$.
    \end{enumerate}

     For a background thread $T_k$, initially $\Sstar(\bB_k)$ is empty. A background solution $\bB_k$ can be modified in the following cases.
    \begin{enumerate}
        \item[-] Adding a passive element in \Cref{alg:f-prepare,alg:f-insert}: The algorithm directly guarantees the invariant.
        \item[-] Deletion in \Cref{alg:f-delete}: The solution does not change.
        \item[-] Adding a tight set in \Cref{alg:f-rebuild}: For each $e\in s\cap E(\bB_k)$, we set $\lev_e=\lev_s=p_k$. By \Cref{lem:f-Ek-correct}, $e$ cannot belong to an existing higher set in $\Sstar(\bB_k)$. Since the algorithms of $T_k$ can only add sets at level $p_k$ and the level pointer $p_k$ is monotone decreasing, there cannot be a higher set added to $\Sstar(\bB_k)$ in the future.
    \end{enumerate}
\end{proof}

\begin{lemma}[Dual feasibility invariant]\label{lem:f-dual-feasible}
    For the dual value $w$ defined by the foreground solution or any background solution, $w_s\le 1$ for all sets $s\in \calS$.
\end{lemma}
\begin{proof}
    We prove the invariant by induction on time. Initially the statement is true.

    The foreground solution can be modified in the following cases.
    \begin{enumerate}
        \item[-] Insertion in foreground thread: The foreground algorithm guarantees the invariant.
        \item[-] Deletion in foreground thread: The solution does not change.
        \item[-] Switch step at termination of $T_k$: The invariant holds by inductive hypothesis for the background solution before the switch.
    \end{enumerate}

    For a background solution $\bB_k$, we first show that the initialization step in \Cref{alg:f-prepare} can only decrease the dual values of elements compared to $\bF$, so that the invariant follows inductive hypothesis for the foreground solution before the initialization step. Elements in $A_k(\bF)$ have $w_e(\bF)=(1+\eps)^{-\lev_e(\bF)}$ by definition of active elements, which is higher then their initial value $(1+\eps)^{-(k+1)}$ in $\bB_k$. After \Cref{alg:f-prepare} handles all active elements, the passive elements are initialized to satisfy the invariant.

    During the background thread, $w_s(\bB_k)$ may increase due to insertion in $L_k(\bF)$, or the change of level $p_k$. In the former case, \Cref{alg:f-insert} guarantees the invariant. In the latter case, when \Cref{alg:f-rebuild} decrease $p_k$ by one to increase $w_e(\bB_k)$ by a factor of $1+\eps$ for all elements in $e\in E(\bB_k)$, only the sets in $\Sin(E(\bB_k))$ have their $w_s$ increased, and they increase by at most by a factor of $1+\eps$. Since the termination condition of the while loop rules out any $s\in \Sin(E(\bB_k))$ with $w_s\ge (1+\eps)^{-1}$, this increase will not violate the invariant.
\end{proof}

\begin{lemma}[Tight set invariant]\label{lem:f-set-in-sol-tight}
    For the foreground solution $\bF$ and all background solutions $\bB_k$, all sets in $\calS(\bF)$ or $\Sstar(\bB_k)$ are tight (i.e., $w_s\ge (1+\eps)^{-1}$).
\end{lemma}
\begin{proof}

    We first prove the invariant for a background solution $\bB_k$.
     Whenever a set $s$ is added to $\Sstar(\bB_k)$ by \Cref{alg:f-prepare,alg:f-rebuild,alg:f-insert}, $s$ must be tight.  $s$ remains tight since the background algorithms cannot decrease $w_s(\bB_k)$.
     
    We next prove the invariant for the foreground solution by induction on time. Initially the statements is true.
    The foreground solution can be modified in the following cases.
    \begin{enumerate}
        \item[-] Insertion in foreground thread: The algorithm assigns the new element $e$ to an existing set $s\in \calS(\bF)$ or adds a new set $s$ to $\calS(\bF)$.
        The choice of $w_e$ guarantees that $s$ remains tight in the former case and $s$ is tight in the latter case.
        \item[-] Deletion in foreground thread: The solution does not change.
        \item[-] Switch step at termination of $T_k$: Let $t^-, t^+$ be the time before and after switch. All elements and sets at levels $\ge k+1$ of $\bF$ before the switch are not affected. By \Cref{lem:f-highest-level-invariant}, removing elements at levels $\le k$ does not affect the total values of sets at levels $\ge k+1$. So, the invariant continues to hold for sets at levels $\ge k+1$ before the switch. The other sets are replaced by $\Sstar(\bB_k)$. For a new set $s\in \Sstar(\bB_k)$, we have $w_s(\bF^{(t^+)}) = w_s(\bB_k^{(t^-)})$. So, the invariant for these sets follows the inductive hypothesis for the background solution before the switch.
    \end{enumerate}
\end{proof}


\subsection{Feasibility}
\begin{lemma}\label{lem:f-feasible}
    The foreground solution (and hence the output solution) is a feasible set cover at the end of every time-step.
\end{lemma}
\begin{proof}
    We prove the lemma by induction on time. Initially the lemma is trivial.
    The foreground solution can be modified in the following cases.
    \begin{enumerate}
        \item[-] Insertion in foreground thread: The foreground algorithm either assigns $e$ to an existing set $s\in \calS(\bF)$ or adds a new set $s$ that cover $e$ to $\calS(\bF)$.
        \item[-] Deletion in foreground thread: The solution does not change.
        \item[-] Switch step at termination of $T_k$: Denote $t$ to be the time-step of switch step and $t^-, t^+$ to be the times before and after the switch step.
        By the inductive hypothesis that $\bF$ is a feasible set cover before the switch, $L_{\lmax}(\bF^{(t^-)}) = L^{(t)}$. By \Cref{lem:f-switch-not-change-subuniverse}, $L_{\lmax}(\bF^{(t^-)}) = L_{\lmax}(\bF^{(t^+)})$. So, $L_{\lmax}(\bF^{(t^+)}) = L^{(t)}$. By \Cref{lem:f-highest-level-invariant}, all elements in $L_{\lmax}(\bF)$ are contained in sets in $\calS(\bF)$. So, $\bF$ is feasible after the switch.
    \end{enumerate}
\end{proof}

\subsection{Recourse}
\begin{lemma}\label{lem:f-recourse}
    The insertion recourse is $O(\log n)$.
\end{lemma}
\begin{proof}
    The output solution consists of the foreground solution $\bF$ and all buffer solutions $\calR_k$.
    The algorithm can add sets to $\bF$ in two ways:
    \begin{enumerate}
        \item Handling an insertion by the foreground thread. This causes at most one insertion recourse.
        \item Switch steps at normal terminations of background threads. This causes no insertion recourse.
    \end{enumerate}
    For each level $k$, the algorithm can add sets to $\calR_k$ in two ways:
    \begin{enumerate}
        \item In the copy and tail phases of $T_k$, or at the termination of $T_k$ when $T_k$ is a base or shortcut thread, $T_k$ copies sets from $\Sstar(\bB_k)$ to $\calR_k$ at the speed of $\cspd$ sets per time-step. This causes at most $2\cspd=O(1)$ recourse, because $T_k$ can go through each phase at most once in a time-step.
        \item  Handling an insertion in the sub-universe of $\bB_k$ when $T_k$ is in the copy and tail phases. This may add a new set to $\Sstar(\bB_k)$, which is repeated in $\calR_k$ in the copy and tail phases. This causes at most one insertion recourse.
    \end{enumerate}  
    In conclusion, the overall insertion recourse is $O(\log n)$ since there are $\lmax+1 = O(\log n)$ background threads.
\end{proof}

To bound the deletion recourse by $O(\log n)$, we call \Cref{lem:new:recourse-reduction}, which designs a garbage collection process to deamortize the deletion recourse.

\recourse*

\subsection{Lifetime Bounds}
In this section, we bound the lasting time of a background thread in each phase.
Recall that a base thread $T_k$ with $|L_k(\bF^{(\tprep)})|\le \cspd$ when created at $\tprep$ will terminate in one time-step and will not enter any phase. So, for a background thread $T_k$ in any phase, we have $|L_k(\bF^{(\tprep)})|> \cspd$.

A background thread $T_k$ is allow to undergo many phases in a time-step. So, the last time-step $t$ of the previous phase is also viewed as the first time-step of the next phase. For each phase $T_k$ enters within a time-step, the algorithm performs the individually scheduled work amount (e.g., processing $\cspd$ elements).
When $T_k$ normally terminates or gets aborted, it will restart at the next time-step.

We set $\cspd=500$ to be a sufficiently large constant.

\begin{fact}\label{lem:f-prep-phase-time}
    Suppose that a background thread $T_k$ is in the preparation phase at time-step $t$. Then,
    $T_k$ will finish preparation phase (enter computation phase or get aborted) by time-step $t+\frac{1.1}{\cspd}\cdot |L_k(\bF^{(t)})|$.
\end{fact}
\begin{proof}
    Suppose $T_k$ is created at time-step $\tprep$, and finishes preparation phase at $\tcomp$.
    During the preparation phase, $T_k$ runs \Cref{alg:f-prepare} and scans elements in the sub-universe $L_k(\bF) = A_k(\bF) \uplus P_k(\bF)$, at a speed of $\cspd$ elements per time-step, except for the last time-step $\tcomp$. On the other hand, $L_k(\bF^{(t)})$ may be modified due to insertion/deletion of elements, but at the speed of at most one element per time-step.
    So, $\tcomp-\tprep \le \frac{1}{\cspd-1}\cdot |L_k(\bF^{(\tprep)})|.$
    Since $t\in[\tprep, \tcomp]$, we have $$|L_k(\bF^{(t)})|\ge |L_k(\bF^{(\tprep)})| - (t-\tprep)  \ge \frac{\cspd-2}{\cspd-1}\cdot |L_k(\bF^{(\tprep)})|.$$
    Combine the above inequalities with $|L_k(\bF^{(\tprep)})|> \cspd$ from the rule of base threads, we conclude
    $$\tcomp-t\le \tcomp-\tprep\le\frac{1}{\cspd-1}\cdot|L_k(\bF^{(\tprep)})|\le \frac{1}{\cspd-2}\cdot |L_k(\bF^{(t)})| \le \frac{1.1}{\cspd}\cdot |L_k(\bF^{(t)})|.$$
\end{proof}

\begin{fact}\label{lem:f-comp-phase-time}
    Suppose that a background thread $T_k$ is in the computation phase at time-step $t$. Then,
    $T_k$ will finish computation phase  by time-step $t+\left\lceil\frac{1.1}{\cspd}\cdot |E(\bB_k^{(t)})|\right\rceil$.
\end{fact}
\begin{proof}
    During the computation phase, $T_k$ runs \Cref{alg:f-rebuild} at a speed of removing $\cspd$ elements from $E(\bB_k)$ per time-step, except for the last time-step. The only way to add elements to $E(\bB_k)$ during the computation phase is due to \Cref{alg:f-insert} handling insertion, which can only add one element in a time-step. (By \Cref{lem:f-switch-not-change-subuniverse,lem:f-Ek-correct}, termination of lower threads cannot add elements to $E(\bB_k)$.) The algorithm terminates when $E(\bB_k)=\emptyset$. So, the computation phase finishes in $\left\lceil\frac{1}{\cspd-1}\cdot \left|E(\bB_k^{(t)})\right|\right\rceil \leq \left\lceil\frac{1.1}{\cspd}\cdot \left|E(\bB_k^{(t)})\right|\right\rceil$ time-steps after $t$.
\end{proof}

\begin{fact}\label{lem:f-tail-phase-time-basic}
    Suppose that a background thread $T_k$ is in the tail phase at time-step $t$. Then, $T_k$ will terminate by time-step $t+\frac{1.1}{\cspd} \cdot |E(\bB_k^{(t)})|$.
\end{fact}
\begin{proof}
    During the tail phase, $T_k$ runs \Cref{alg:f-rebuild} at a speed of removing $\cspd$ elements from $E(\bB_k)$ per time-step, including the last time-step because of the shortcut rule. The only way to add elements to $E(\bB_k)$ during the computation phase is due to \Cref{alg:f-insert} handling insertion, which can only add one element in a time-step. The algorithm terminates when $E(\bB_k)=\emptyset$. So, the tail phase finishes in $\frac{1}{\cspd-1}\cdot \left|E(\bB_k^{(t)})\right| \leq \frac{1.1}{\cspd}\cdot \left|E(\bB_k^{(t)})\right|$ time-steps after $t$.
\end{proof}

\begin{fact}\label{lem:f-copy-phase-time}
    Suppose that a background thread $T_k$ is in the copy phase at time-step $t$. Then,
    $T_k$ will finish copy phase by time-step $t+\frac{2.1}{\cspd}\cdot \tausus_k$.
\end{fact}
\begin{proof}
    During each time-step within the copy phase, the thread $T_k$ copies $\cspd$ sets from $\Sstar(\bB_k^{(t)})$ to the buffer solution $\calR_k$. On the other hand, $\Sstar(\bB_k)$ may grow by one set per time-step due to \Cref{alg:f-insert}.
    So, the copy phase finishes in $\left\lceil\frac{1}{\cspd-1}\cdot \tausus_k\right\rceil \le \frac{1.1}{\cspd}\cdot \tausus_k + 1 \le \frac{2.1}{\cspd}\cdot \tausus_k$ time-steps after $t$. Here, the last step uses $\tausus_k\ge \cspd$ by rule of shortcut threads.
\end{proof}

\begin{lemma}\label{lem:f-copy-tail-time-weak}
    Suppose a background thread $T_k$ is in the \color{blue} copy or tail phase \color{black} at time-step $t$, $T_k$ spent at most $\frac{30}{\cspd}\cdot \tausus_k$ time-steps in suspension phase.
    Then, $T_k$ will terminate by time-step $t+ \frac{5}{\cspd}\cdot \tausus_k$.
\end{lemma}
\begin{proof}
    Suppose $T_k$ enters suspension, copy and tail phases at $\tsus, \tcopy, \ttail$ respectively, and terminates at $\tend$.
    (If $T_k$ is aborted in copy phase, let $\ttail = \tend$.) Since $t\in [\tcopy, \tend]$, it suffices to prove $\tend-\tcopy\le \frac{5}{\cspd}\cdot \tausus_k$.
    During the suspension and copy phases, $E(\bB_k)$ can only be modified by insertion or deletion in one element per time-step. That is,
    \begin{equation}\label{eq:sus-copy-exposed-difference}
        ||E(\bB_k(t_1))|-|E(\bB_k(t_2))|| \le |t_1-t_2|, \forall t_1, t_2\in[\tsus, \ttail]
    \end{equation}

    The assumption gives $\tcopy - \tsus \le \frac{30}{\cspd}\cdot \tausus_k$. By \Cref{lem:f-copy-phase-time}, $\ttail-\tcopy\le \frac{2.1}{\cspd}\cdot \tausus_k$. If $\tend = \ttail$, we are done. Otherwise, it remains to bound $\tend-\ttail \le \frac{2.9}{\cspd}\cdot \tausus_k$.
    From (\ref{eq:sus-copy-exposed-difference}), we have $$|E(\bB_k(\ttail))| \le |E(\bB_k(\tsus))| + (\ttail-\tsus)\le  \left(1+\frac{32.1}{\cspd}\right)\tausus_k\le 1.1\tausus_k.$$
    By \Cref{lem:f-tail-phase-time-basic}, we have that $\tend-\ttail \le \frac{1.1}{\cspd}|E(\bB_k(\ttail))|\le \frac{1.3\tausus_k}{\cspd}$. 
    
\end{proof}

\begin{lemma}\label{lem:f-suspension-phase-time}
Suppose a background thread $T_k$ is in the copy or tail phase at time-step $t$. Then, $T_k$ either enters copy phase or gets aborted in suspension phase by time-step $t+\frac{30}{\cspd}\cdot \tausus_k$.
\end{lemma}
\begin{proof}
    We prove the lemma by induction first on time in increasing order, then on level $k$ in decreasing order. That is, in our inductive step, the assumption of \Cref{lem:f-copy-tail-time-weak} holds for all threads at levels higher than $k$, and also holds for all threads that entered copy phase before $t$.

    If $T_k$ is in suspension phase but does not enter copy phase at any time-step $t'$, one of the following reasons must hold:
    \begin{enumerate}
        \item (stopped from above) There exists a higher thread $T_j, j> k$ in copy or tail phase such that $\tausus_j < 2 \tausus_k$.
        \item (stopped from below) There exists a lower thread $T_\ell, \ell<k$ in copy or tail phase such that $\tausus_\ell < 2 \tausus_k$. By the sequential behavior of threads, $T_\ell$ must have entered copy phase before $t'$.
    \end{enumerate}
    If $T_k$ is stopped from above at time-step $t$, then $T_k$ cannot enter copy phase before $T_j$ terminates since $\tausus_j, \tausus_k$ cannot change until the threads terminate. But, termination of $T_j$ will abort $T_k$. So, $T_k$ will get aborted in suspension phase. By \Cref{lem:f-copy-tail-time-weak}, this will happen by time-step $t+\frac{5}{\cspd}\cdot \tausus_j \le t+\frac{10}{\cspd}\tausus_k$. 

    The remaining case is that $T_k$ is stopped from below at time-step $t$. In this case, we define a sequence of levels $(\ell_1, \ell_2, \ldots, \ell_r)$ and time-steps $(t_0, t_1, t_2, \ldots, t_r)$ as follows. We start with $t_0=t$. For each $i\ge 0$, we assume $T_k$ is stopped from below at time-step $t_i$, and let $\ell_{i+1}$ be the highest thread in copy or tail phase at $t_i$ such that $\ell<k, \tausus_\ell<2\tausus_k$. Then, let $t_{i+1}$ be the time-step when $T_{\ell_{i+1}}$ terminates. The sequence ends at $r$ if $T_k$ is not stopped from below or not in suspension phase after $T_{\ell_r}$ terminates at $t_r$.

    By \Cref{lem:f-copy-tail-time-weak}, $t_{i+1}-t_i \le \frac{5}{\cspd}\cdot \tausus_{\ell_i}$ for each $i$ in the sequence. To bound the total waiting time $t_r-t$, we bound the sum of $\tausus_{\ell_i}$ as follows. For each $i\in[2,r]$, $\ell_i$ is higher than $\ell_{i-1}$ since $T_{\ell_i}$ is not aborted at the termination of $T_{\ell_{i-1}}$. Hence, $T_{\ell_i}$ is not in copy or tail phase at $t_{i-2}$, otherwise we should not have chosen $\ell_{i-1}$. That is, $T_{\ell_i}$ enters copy phase during $[t_{i-2}+1, t_{i-1}]$ when $T_{i-1}$ is in copy or tail phase. By the rule to enter copy phase, $\tausus_{\ell_i}\le \frac 12 \tausus_{\ell_{i-1}}$.
    Since this holds for all $i\in[2, r]$, $\tausus_{\ell_i}$ form a geometric series. We conclude $$t_r-t \le \frac{5}{\cspd}\sum_{i=1}^r \tausus_{\ell_i} \le \frac{10}{\cspd}\cdot \tausus_{\ell_1} \le \frac{20}{\cspd}\cdot \tausus_k.$$

     Note that $T_k$ cannot enter copy phase before $t_r$ since there is always some $T_{\ell_i}$ stopping $T_k$ from below during $[t, t_r]$.
    Also, $T_k$ is not stopped from below after $t_r$. So, after $t_r$, there are the following cases.
    \begin{enumerate}
        \item $T_k$ gets aborted.
        \item $T_k$ enters copy phase.
        \item $T_k$ is stopped from above. We can use the same argument as the case of stopped from above at $t$, with an additional waiting time $t_r-t$. Then, $T_k$ gets aborted in suspension phase by $t+\frac{30}{\cspd}\cdot \tausus_k$.
    \end{enumerate}
    The lemma holds in all cases.
\end{proof}

We combine the results in this section as follows.
\begin{lemma}\label{lem:f-full-lifetime-bound}
    Suppose a background thread $T_k$ is running at time-step $t$. Then, the lifetime of $T_k$ is at most $0.1|L_k(\bF^{(t)})|$ time-steps.
\end{lemma}
\begin{proof}
    Suppose $T_k$ starts at $\tprep$ and terminates at $\tend$. Suppose $T_k$ enters computation phase, suspension phase and copy phase at $\tcomp, \tsus, \tcopy$ respectively. If $T_k$ terminates before these phases, we collapse the phases by redefining the notations to be $\tend$. Any time bound for collapsed phases is trivially true since the time gap is 0 under this definition, even though the assumptions for corresponding lemmas may not hold. By the rule of base threads, we may assume $|L_k(\bF^{(\tprep)})|\ge \cspd$.
    
    We have the following time bounds:
    \begin{align*}
        \tcomp-\tprep &\le \frac{1.1}{\cspd}\cdot |L_k(\bF^{(\tprep)})| & \quad (\text{\Cref{lem:f-prep-phase-time}})\\
        \tsus-\tcomp &\le \frac{1.1}{\cspd}\cdot |L_k(\bF^{(\tcomp)})| + 1 & \quad (\text{\Cref{lem:f-comp-phase-time,lem:f-Ek-correct}})\\
        \tend-\tsus &\le \frac{35}{\cspd}\cdot \tausus_k & \quad (\text{\Cref{lem:f-copy-tail-time-weak,lem:f-suspension-phase-time}})
    \end{align*}
    Let $\Delta = 0.1 |L_k(\bF^{(\tprep)})|$. During $t\in [\tprep, \tprep+\Delta]$, $|L_k(\bF^{(t)})|$ can change by at most 1 per time-step, so we have $|L_k(\bF^{(t)})|\in[0.9 |L_k(\bF^{(\tprep)})|, 1.1 |L_k(\bF^{(\tprep)})|]$. To bound $\tausus_k = |\Sstar(\bB_k^{\tsus})|$, we consider a bijection from each set $s\in \Sstar(\bB_k^{\tsus})$ to an element $e$ added to $\Cstar(\bB_k)$ accompanied by $s$ in \Cref{alg:f-prepare,alg:f-rebuild,alg:f-insert}. These elements are distinct and belong to $L_k(\bF)$ when added to $\Cstar(\bB_k)$. So, $\tausus_k\le 1.1 |L_k(\bF^{(\tprep)})|$. In conclusion,
    \[\tend-\tprep\le (1.1+ 1.1\times 1.1 + 1 + 35\times 1.1)\frac{1}{\cspd}\cdot |L_k(\bF^{(\tprep)})| \le  \frac{45} {\cspd}\cdot |L_k(\bF^{(\tprep)})|\]
    Since $\frac{45}{\cspd} < 0.1$, our above argument is valid for $[\tprep, \tend] \subseteq [\tprep, \tprep+\Delta]$. Because $|L_k(\bF^{(t)})| > 0.9 |L_k(\bF^{(\tprep)})|$, we conclude that $\tend-\tprep \le \frac{50}{\cspd}|L_k(\bF^{(t)})| = 0.1|L_k(\bF^{(t)})|$.
\end{proof}

\begin{lemma}\label{lem:f-comp-phase-time-full}
    Suppose a background thread $T_k$ is in the computation phase and $E(\bB_k) > |\Sstar(\bB_k)|$ at time-step $t$. Then, $T_k$ will terminate by $t+0.1|E(\bB_k^{(t)})|$.
\end{lemma}
\begin{proof}
    Let $\tsus$ be the time when $T_k$ finishes computation phase, and $\tend$ be the time when $T_k$ terminates. By \Cref{lem:f-comp-phase-time}, $\tsus-t\le \frac{1.1}{\cspd}\cdot |E(\bB_k^{(t)})| + 1 \le \frac{2.1}{\cspd}\cdot |E(\bB_k^{(t)})|$. Here, the last step uses the rule of shortcut threads. By \Cref{lem:f-copy-tail-time-weak,lem:f-suspension-phase-time}, $\tend-\tsus \le \frac{35}{\cspd}\tausus_k$ (trivially true if $T_k$ does not enter suspension phase). Since sets in $\Sstar(\bB_k)$ cannot be removed, we have $\tausus_k\le |\Sstar(\bB_k^{(t)})| + (\tsus-t) \le (1+\frac{2.1}{\cspd})\cdot |E(\bB_k^{(t)})|$.
    In conclusion, $\tend-t \le (\frac{2.1}{\cspd} + \frac{35}{\cspd}(1+\frac{2.1}{\cspd}))\cdot |E(\bB_k^{(t)})|\le 0.1|E(\bB_k^{(t)})|$.
\end{proof}

\subsection{Tidy Property}

Recall that we say a hierarchical solution $\bF$ is $\eps$-tidy if $|D_k(\bF)| + |P_k(\bF)| \le \eps \cdot (|A_k(\bF)| + |E(\bF)|)$ for all $k\in[0,\lmax]$.
We next establish the tidy property for the foreground threads and all background threads. It will be useful in bounding approximation factor.
Intuitively, we can expect from the lifetime bounds that there cannot accumulate too many dormant and passive elements.

\begin{lemma}\label{lem:f-background-tidy}
    Suppose a background thread $T_k$ is in computation or later phases at time-step $t$, and $\bF$ was $\eps$-tidy when $T_k$ was created at time-step $\tprep$. Then, $\bB_k^{(t)}$ is $0.1$-tidy at levels $j \le k$, and $(1.2\eps+0.2)$-tidy at levels $j > k$.
\end{lemma}
\begin{proof}
    First, we consider levels $j\le k$. Notice that any element added to $C(\bB_k)$ by $T_k$ must have $\lev_e(\bB_k)=p_k$, and $\lev_e(\bB_k)$ cannot change once decided. So, before $p_k$ is decreased to $j$, $D_j(\bB_k) = P_j(\bB_k)=\emptyset$.
    Suppose $p_k$ is decreased to $j$ at time $t_j$. After $t_j$, the only way to generate a dormant element in $D_j(\bB_k)$ is by \Cref{alg:f-delete}, and the only way to generate a passive element in $P_j(\bB_k)$ is by \Cref{alg:f-insert}. So, $|D_j(\bB_k)| + |P_j(\bB_k)|$ can increase by at most 1 per time-step after time-step $t_j$. That is, $$|D_j(\bB_k^{(t)})| + |P_j(\bB_k^{(t)})| \le t-t_j.$$
    We claim that $t-t_j \le 0.1 |E(\bB_k^{(t_j)})|$, which implies the tidy property for level $j$.

    We now prove the claim. Just before $p_k$ decrease to $j\le k$, \Cref{alg:f-rebuild} have finished the last set at level $j+1$. At that time, there are two cases: either $|E(\bB_k)|>|\Sstar(\bB_k)|$, or $|E(\bB_k)|\le |\Sstar(\bB_k)|$.
    In the former case, the claim follows \Cref{lem:f-comp-phase-time-full}. In the latter case, $T_k$ will first enter suspension phase, and must wait until the tail phase to decrease $p_k$ to $j$. Then, the claim follows \Cref{lem:f-tail-phase-time-basic}.

    Next, we consider levels $j > k$. By the tidy property of $\bF^{(\tprep)}$,
    $$|D_j(\bF^{(\tprep)})| + |P_j(\bF^{(\tprep)})|\le \eps \cdot |A_j(\bF^{(\tprep)})|.$$
    At the beginning of $\bB_k$, we view all elements in $A_k(\bF)$ are exposed in $E(\bB_k)$ and $A(\bB_k)$ is empty. With a slight abuse of notation, we view dormant and passive elements from $\bF$ to be in $D_j(\bB_k)$ and $P_j(\bB_k)$. The position of an element in $E(\bB_k), A(\bB_k), P(\bB_k)$ or $D(\bB_k)$ follows the algorithm after it gets processed by \Cref{alg:f-prepare} in the preparation phase. (The alternate view is valid because the lemma applies after the preparation phase.) Then, we have 
    $$|D_j(\bB_k^{(\tprep)})| + |P_j(\bB_k^{(\tprep)})|\le \eps (|A_j(\bB_k^{(\tprep)})| + |E(\bB_k^{(\tprep)})|).$$
    
    We next discuss how this relation changes from $\tprep$ to $t$.
    A new dormant element in $D_j(\bB_k)$ (not from $D_j(\bF^{(\tprep)})$) can only be generated by \Cref{alg:f-delete} or the foreground thread deletion algorithm. Similarly, a new passive element in $P_j(\bB_k)$ (not from $P_j(\bF^{(\tprep)})$) can only be generated by \Cref{alg:f-insert} or the foreground thread insertion algorithm. In addition, \Cref{alg:f-prepare} in preparation phase may move a passive element in $P_k(\bF)$ to level $k+1$ or to $E(\bB_k)$, which cannot increase $|D_j(\bB_k^{(\tprep)})| + |P_j(\bB_k^{(\tprep)})|$. In conclusion, $|D_j(\bB_k^{(\tprep)})| + |P_j(\bB_k^{(\tprep)})|$ can only increase by at most 1 per time-step due to insertion/deletion.
    \Cref{alg:f-prepare,alg:f-rebuild} may assign an exposed element to be active at levels $p_k \le k+1$, which does not change $|A_j(\bB_k)| + |E(\bB_k)|$.
    \Cref{alg:f-insert,alg:f-delete} may remove or add an element to $E(\bB_k)$ due to deletion or insertion.
    So, $|A_j(\bB_k)| + |E(\bB_k)|$ may decrease by at most 1 per time-step.
    By \Cref{lem:f-full-lifetime-bound}, $\Delta_t := t-\tprep\le 0.1|L_k(\bF^{(\tprep)})| =  0.1(|A_j(\bB_k^{(\tprep)})| + |E(\bB_k^{(\tprep)})|)$. So, 
    \begin{align*}
        |D_j(\bB_k^{(t)})| + |P_j(\bB_k^{(t)})|
        &\le |D_j(\bB_k^{(\tprep)})| + |P_j(\bB_k^{(\tprep)})| +\Delta_t\\
        &\le (\eps +0.1) (|A_j(\bB_k^{(\tprep)})| + |E(\bB_k^{(\tprep)})|)\\
        &\le (\eps +0.1) (|A_j(\bB_k^{(t)})| + |E(\bB_k^{(t)})|+\Delta_t)\\
        &\le (\eps + 0.1)(1+\frac{0.1}{1-0.1}) (|A_j(\bB_k^{(t)})| + |E(\bB_k^{(t)})|)\\
        &\le (1.2\eps + 0.2) (|A_j(\bB_k^{(t)})| + |E(\bB_k^{(t)})|)
    \end{align*}
\end{proof}

\begin{lemma}\label{lem:f-foreground-tidy}
    The foreground solution $\bF$ is always $0.5$-tidy.
\end{lemma}
\begin{proof}
    We prove by induction on time-step $t$. As the base case, the initial greedy solution is 0-tidy. For the inductive case, it suffices to prove the following claim for every level $k\in[0,\lmax]$: Suppose $\bF$ is 0.2-tidy at level $k$ when $T_k$ starts. Then, $\bF$ is 0.2-tidy at level $k$ when $T_k$ terminates, and $\bF$ is 0.5-tidy at level $k$ throughout the lifetime of $T_k$.

    When $T_k$ terminates, it either normally terminates or gets aborted due to a normal termination of a higher thread. In both cases, \Cref{lem:f-background-tidy} guarantees that the $\bF$ is 0.1-tidy at level $k$ after the switch step. This establishes the first part of the claim.

    For the second part of the claim, suppose $T_k$ is created at time-step $\tprep$ and terminates at time-step $\tend$. By \Cref{lem:f-full-lifetime-bound}, for any $t\in [\tprep, \tend]$,
    $$\Delta_t :=\tend-\tprep \le 0.1 |L_k(\bF^{(t)})| = 0.1(|A_k(\bF^{(t)})|+ |P_k(\bF^{(t)})|).$$
    Notice that $E(\bF)$ is always empty since $\bF$ is feasible. Since $\bF$ is 0.2-tidy when $T_k$ is created,
    $$|D_k(\bF^{(\tprep)})| + |P_k(\bF^{(\tprep)})|\le 0.2 |A_k(\bF^{(\tprep)})|.$$
    Combining the above inequalities gives
    \begin{equation}
        \Delta_t \le |P_j(\bB_k^{(\tprep)})|
    \end{equation}
    During the lifetime of $T_k$, a new dormant element in $D_k(\bF)$ can be generated by the foreground thread deletion algorithm, a new passive element in $P_k(\bF)$ can be generated by the foreground thread insertion algorithm, and an element can be removed from $A_k(\bF)$ by the foreground thread deletion algorithm. These events can happen for one element per time-step due to insertion/deletion. In addition, termination of a lower thread $T_j, j<k$ may remove some dormant elements, or move some passive elements to be either passive or active with levels $j+1 \le k$. These events cannot increase $|D_k(\bF)|$, $|P_k(\bF)|$ or decrease $|A_k(\bF)|$. For any $t\in[\tprep, \tend]$, we have the following.
        \begin{align*}
        |D_k(\bF^{(t)})| + |P_k(\bF^{(t)})|
        &\le |D_k(\bF^{(\tprep)})| + |P_k(\bF^{(\tprep)})| +\Delta_t\\
        &\le 0.2 (|A_k(\bF^{(t)})|+\Delta_t)+\Delta_t \\
        &\le 0.2|A_k(\bF^{(t)})| + 0.12(|A_k(\bF^{(t)})|+ |P_k(\bF^{(t)})|)\\
        |D_k(\bF^{(t)})| + |P_k(\bF^{(t)})|  & \le \frac{0.2+0.12}{1-0.12}|A_k(\bF^{(t)})| \le 0.5|A_k(\bF^{(t)})|
    \end{align*}
\end{proof}

\subsection{Approximation Factor}
The $O(f)$-competitiveness of the foreground solution and each (non-empty) background solution follows \Cref{lem:f-approx-from-tidy}.

The buffer solutions $\calR_k$ are empty for background threads before copy phase, and have size $O(\tausus_k)$ after entering copy phase. Since $\tausus_k$ of threads in copy and tail phases must form a geometric series, the total size of buffer solutions is at most constant times larger than the highest thread, which is $O(f)$-competitive.

\begin{define}
    The foreground solution $\bF$ is said to be $\eps$-tidy at level $k$ if $|D^\bF_k| + |P^\bF_k| \le \eps \cdot |A^\bF_k|$, and $\eps$-dirty at level $k$ otherwise.
    If $\bF$ is $\eps$-tidy at every level $k\in[0,\lmax]$, we say it is $\eps$-tidy.

    A background solution $\bB_k$ is said to be $\eps$-tidy at level $j\ge p_k$ if $|D^{\bB_k}_j| + |P^{\bB_k}_j| \le \eps \cdot |A^{\bB_k}_j| + |E(\bB_k)|$, and $\eps$-dirty at level $j$ otherwise.
    If $\bB_k$ is $\eps$-tidy at every level $j\in[p_k,\lmax]$, we say it is $\eps$-tidy.
\end{define}

\begin{lemma}\label{lem:f-approx-from-tidy}
    If $\bF$ or $\bB_k$ is an $\eps$-tidy hierarchical solution for some $\eps=\Omega(1)$, then it is $O(f)$-competitive.
\end{lemma}
\begin{proof}
    Given \Cref{lem:f-approx-general,lem:f-dual-feasible,lem:f-level-invariant,lem:f-highest-level-invariant,lem:f-set-in-sol-tight}, it remains to prove $w_D\le \eps w_L$ for $\bF$ or $\bB_k$.

    Extend the notation $D_j(\bF)$ to all integers $j$, so that $D_j(\bF)=D_{\lmax}(\bF)$ when $j>\lmax$, and $D_j(\bF) = \emptyset$ when $j<0$.
    \begin{align*}
        w_D &= \sum_{j=0}^{\infty} w_{D_j(\bF)} - w_{D_{j-1}(\bF)} 
        &\le \sum_{j=0}^{\infty} (1+\eps)^{-j} (|D_j(\bF)| - |D_{j-1}(\bF)|) & \quad\text{(level invariant)}\\
        &= \sum_{j=0}^{\infty} ((1+\eps)^{-j}-(1+\eps)^{-(j+1)})|D_j(\bF)|
        &\le \sum_{j=0}^{\infty} \eps(1+\eps)^{-(j+1)}\cdot \eps|A_j(\bF)| & \quad\text{(tidy property)}\\
        &=\eps \sum_{j=0}^{\infty} (1+\eps)^{-j} (|A_j(\bF)| - |A_{j-1}(\bF)|) 
        &= \eps\cdot w_A \le \eps \cdot w_L  &\quad\text{(definition of active elements)}\\
    \end{align*}
    The proof for background solutions is similar if we move $E(\bB_k)$ to $A_{p_k}(\bB_k)$ and take the sum from $p_k$ instead of 0. This is valid because no element or set in $\bB_k$ can have level $< p_k$.
\end{proof}

\begin{lemma}\label{lem:f-sstar-le-tausus}
    Suppose that a background thread $T_k$ is in the copy or tail phase at time-step $t$. Then, $|\Sstar(\bB_k^{(t)})|\le 3\tausus_k$.
\end{lemma}
\begin{proof}
    Suppose $T_k$ enters suspension phase at $\tsus$. 
    Define potential function $\phi:= |\Sstar(\bB_k)|+|E(\bB_k)|$.
    We have $\phi^{(\tsus)}\le 2\tausus_k$ by the criteria to enter suspension phase.
    During the suspension, copy and tail phases, an insertion in \Cref{alg:f-insert} either adds a new exposed element, or adds a set to cover a new passive element. In both cases, $\phi$ increase by 1. A deletion does not increase $\phi$.
    When \Cref{alg:f-rebuild} adds a set $s$ to $\Sstar(\bB_k)$, the algorithm must remove at least one exposed element, and $\phi$ cannot increase. 
    
    We have $t-\tsus\le 0.1\tausus_k$ by \Cref{lem:f-suspension-phase-time,lem:f-copy-tail-time-weak}. In conclusion,
    $$|\Sstar(\bB_k^{(t)})| \le \phi^{(t)}\le \phi^{(\tsus)} + (t-\tsus) \le 3\tausus_k.$$
\end{proof}

\begin{lemma}\label{lem:f-approx-final}
    The output is $O(f)$-competitive.
\end{lemma}
\begin{proof}
    The foreground solution is $O(f)$-competitive by \Cref{lem:f-approx-from-tidy,lem:f-foreground-tidy}. For each background thread after preparation phase, $\Sstar(\bB_k)$ is $O(f)$-competitive by \Cref{lem:f-approx-from-tidy,lem:f-background-tidy}. Since buffer solutions $\calR_k$ are empty in preparation phase and always a subset of $\Sstar(\bB_k)$, each $\calR_k$ is $O(f)$-competitive.

    By \Cref{lem:f-sstar-le-tausus}, $|\calR_k| \le 3\tausus_k$ for each buffer solution.
    By the criteria to enter copy phase, $\tausus_k$ are upper bounded by a geometric series with ratio 0.5. So, the total size of buffer solutions is upper bounded by the largest $\tausus_k$, which is $O(f)$-competitive by \Cref{lem:f-approx-from-tidy,lem:f-background-tidy} at the time $T_k$ enters suspension phase.
\end{proof}

\subsection{Implementation Details and Update Time}First, we describe the data structures used to maintain the foreground, background, and buffer solutions.

\paragraph{Data Structures.}

The algorithm maintains the following data structures:

\begin{enumerate}
    \item The primal solutions maintained by the algorithm, including $\calS(\bF)$ for the foreground solution, $\Sstar(\bB_k)$ for the background solutions, and $\calR_k$ for the buffer solutions, are organized by levels, and the sets in each level are stored in a balanced binary search tree. We call these set BSTs and use $\Tset(S, \ell)$ to denote the set BST for solution $S$ at level $\ell$.
    The algorithm maintains a global pointer to each set BST.

    \item The covered elements in hierarchical solutions maintained by the algorithm, including $C(\bF)$ for the foreground solution and $\Cstar(\bB_k)$ for the background solutions, are partitioned into active, passive and dormant elements, and organized by levels. 
    The elements in each level of each part are stored in a balanced binary search trees. We call these element BSTs and use $\Tact(S, \ell), \Tpas(S, \ell),\Tdor(S, \ell)$ to denote the three BSTs for solution $S$ at level $\ell$. A node in an element BST stores the identity of an element $e$, its level $\lev_e(S)$ and dual value $w_e(S)$. The algorithm maintains a global pointer to each element BST.

    \item For each background thread $T_k$, we have an array indexed by all sets where for each set $s$, it stores a BST containing all exposed elements in $s\cap E(\bB_k)$. We call these exposed BSTs and denote them $\Texp(\bB_k, s)$. In contrast to other element BSTs, the exposed BSTs do no maintain levels or dual values (since they are undefined).
    
    \item For the foreground solution and each level $\ell$, we have a BST maintaining the total dual value of elements in $s$ at level $\ell$ of $\bF$ for all sets $s$ incident to  level $\ell$ of $\bF$.  We call this foreground incident BSTs and denote it $\Tin(\bF, \ell)$.
    Similarly, for each background solution $\bB_k$ and each level $\ell\le k$, we have a BST maintaining the total dual value of elements in $s$ at levels $\ell$ of $\bB_k$ for all sets $s$ incident to level $k$ of $\bF$.  We call this background incident BSTs and denote it $\Tin(\bB_k, \ell)$. The algorithm maintains a global pointer to each incident BST. In contrast to the set BSTs, we allow duplicate sets in the incident BSTs.
    
    \item For each background thread $T_k$, we use a priority queue $Q_k$ to maintain the dual values $w_s(\bB_k)$ for all sets $s\in \Sin(E(\bB_k))$.
    Recall that $w_s(\bB_k)$ consists of three parts: the total dual value in $\Cstar(\bB_k)$, the total dual value at levels $\ge k+1$ in $\bF$, and $(1+\eps)^{-p_k}\cdot|s\cap E(\bB_k)|$. Denote $\wstar_s$ to be the sum of first two parts. Since \Cref{alg:f-rebuild} needs to decide whether there exists a set with $w_s\ge (1+\eps)^{-1}$ for changing $p_k$, the $Q_k$ actually maintains
    a target value $t_s:=\frac{(1+\eps)^{-1}-\wstar_s}{|s\cap E(\bB_k)|}$ for sets $s$, and converts the maximum query on $w_s$ to minimum query on the target values. It is straightforward to check that the condition $w_s\ge (1+\eps)^{-1}$ is equivalent to $t_s\le (1+\eps)^{-p_k}$. \footnote{Here, dividing by 0 outputs $+\infty$ when the nominator is positive and $-\infty$ when the nominator is nonpositive.}
\end{enumerate}

\paragraph{Update Time Bound.}

We now give details of how the data structures are updated by the algorithm, and derive resulting bounds on the update time of the algorithm.

\begin{lemma}\label{lem:f-query-time}
    Given the index of an element $e$ or set $s$ and a hierarchical solution $\bF$ or $\bB_k$, the data structures can answer in $O(\log^2 n)$ time:
    \begin{enumerate}
        \item Whether $e\in C(\bF), s\in \calS(\bF), e\in \Cstar(\bB_k), e\in E(\bB_k)$ or $s\in \Sstar(\bB_k)$.
        \item $\lev_e(\bF)$ and $w_e(\bF)$ if $e\in C(\bF)$, or $\lev_e(\bB_k)$ and $w_e(\bB_k)$ if $e\in \Cstar(\bB_k)$.
        \item $w_s(\bF)$ or $w_s(\bB_k)$.
    \end{enumerate}
\end{lemma}
\begin{proof}
    Since the input is only the index, we first need to find the node representing the element or set in the corresponding element BSTs, set BSTs or priority queues. Search a given index in a given BST or priority queue takes $O(\log n)$ time. Because we do not know the level of the input element or set, we need to search in $O(\log n)$ BSTs for all levels. This takes $O(\log^2 n)$ time in total.

    Queries 1 and 2 can be answered by information stored in the node. Next, we consider query 3 on total dual value $w_s$. For the foreground solution, the value is maintained in the foreground incident BSTs. For the background solution, the value consists of three parts, which can be found in the foreground incident BSTs $\Tin(\bF, \ell), \ell>k$, background incident BSTs $\Tin(\bB_k, \ell), \ell\le k$, and exposed BST $\Texp(k, s)$ respectively. Here, we make $O(\log n)$ queries on $O(\log n)$ BSTs, which costs $O(\log^2 n)$ time.
\end{proof}

\begin{lemma}\label{lem:f-operation-runtime}
    The data structures can be maintained in $O(f\log^2 n)$ time for the following operations:
    \begin{enumerate}
        \item Adding or removing a set in $\calS(\bF), \Sstar(\bB_k), \calR_k$ or $\Sin(E(\bB_k))$.
        \item Changing element status of live/dormant, active/passive, covered/exposed in $C(\bF)$, $\Cstar(\bB_k)$ or $E(\bB_k)$.
        \item Updating $w_e(\bF), e\in C(\bF)$ or $w_e(\bB_k), e\in \Cstar(\bB_k)$.
        \item Adding or removing an element in $C(\bF)$, $\Cstar(\bB_k)$ or $E(\bB_k)$.
    \end{enumerate}
    The input to the operations is the hierarchical solution, index of the element or set, and its new level and weight if they are involved.
\end{lemma}
\begin{proof}
    As argued in \Cref{lem:f-query-time}, we can find the node representing the element or set in the corresponding element BSTs, set BSTs or priority queues in $O(\log^2 n)$ time.
    We now bound the running time of each operation.
    \begin{enumerate}
        \item Adding or removing a set in $\calS(\bF), \Sstar(\bB_k)$ or $\calR_k$ can be handled by adding or removing a node in the corresponding set BST.
        Adding or removing a set in $\Sin(E(\bB_k))$ can be handled by adding or removing a node in the priority queue.
        These operations take $O(\log n)$ time. The total running time is $O(\log^2 n)$ for finding the node.
        \item For all types of change of status, we can maintain the data structures by removing the node representing the element from the corresponding element BST ($\Tact,\Tpas,\Tdor$ or $\Texp$), and add a new node representing the element to another element BST. This takes $O(\log n)$ time. 
        
        If the operations involves $E(\bB_k)$, we check whether each incident sets is added or removed from $\Sin(E(\bB_k))$. By item 1, this takes $O(f\log^2 n)$ time in total. We also need to update the target values of all incident sets in $Q_k$ to reflect the changes in $|s\cap E(\bB_k)|$ (which can be queried from $\Texp(k,s)$). This is at most $f$ priority queue operations. The total running time is $O(f\log^2 n)$.

        \item For change of dual value of an element $e$ in $C(\bF)$ (resp.\ $\Cstar(\bB_k)$), we first update its node in the corresponding element BST, and then update the total dual value of all incident sets in $\Tin(\bF, \lev_e(\bF))$ (resp.\ $\Tin(\bB_k, \lev_e(\bB_k))$). After that, for each incident set $s$, we update its target value in $Q_\ell$ for all levels $\ell$. (A change of dual value in foreground may affect the target value in all background threads.) These operations take $O(f\log^2 n)$ time.
        
        \item Adding or removing an element can be handled by adding or removing a node in the corresponding element BST ($\Tact,\Tpas,\Tdor$ or $\Texp$).
        Since this changes the dual value of the element, we repeat the argument in item 3.
        If the operation involves $E(\bB_k)$, we also repeat the argument in item 2. The total running time is $O(f\log^2 n)$.
        \qedhere
    \end{enumerate}
\end{proof}

\begin{lemma}\label{lem:f-update-time}
    In any time-step, the foreground thread or a background thread $T_k$ calls runs for $O(f \log^2 n)$ time. Therefore, the worst-case update time of the algorithm is $O(f\log^3 n)$.
\end{lemma}
\begin{proof}

First, we consider an element insertion or deletion. 
This is handled by the foreground thread, and \Cref{alg:f-insert,alg:f-delete} for each background thread. All these algorithms executes $O(f)$ queries described in \Cref{lem:f-query-time} and $O(1)$ operations described in \Cref{lem:f-operation-runtime}. So, the running time is $O(f\log^2 n)$. 

Next, we analyze the operations performed by a background thread $T_k$ based on the phase it is in. This includes \Cref{alg:f-prepare} in the preparation phase and \Cref{alg:f-rebuild} in the computation and tail phases. We limit the work in each phase of a time-step to be processing $\cspd$ elements and copying $\cspd$ sets. Processing an element or copying a set calls $O(f)$ queries described in \Cref{lem:f-query-time} and $O(1)$ operations described in \Cref{lem:f-operation-runtime}. So, these work take $O(f\log^2 n)$ time.
The remaining step that cannot be bounded by these operations is when $p_k$ decreases in \Cref{alg:f-rebuild}. In this situation, the data structure need no update, because we do not maintain the value of exposed elements, and the priority queue maintains target values instead of total dual values. The condition of the while loop in \Cref{alg:f-rebuild} can be checked using the minimum target value in $Q_k$, since $w_s\ge (1+\eps)^{-1}$ is equivalent to $t_s\le (1+\eps)^{-p_k}$.

Finally, we bound the running time of a termination step at the normal termination of a background thread $T_k$.
We use pointer switch to replace the foreground data structures ($\Tset(\bF, \ell),\Tact(\bF, \ell),\Tpas(\bF, \ell),\Tdor(\bF,\ell)$ and $\Tin(\bF,\ell)$) at levels $\ell \le k$ by the corresponding background data structures ($\Tset(\bB_k, \ell),\Tact(\bB_k, \ell),\Tpas(\bB_k, \ell),\Tdor(\bB_k,\ell)$ and $\Tin(\bB_k,\ell)$).
Besides this, we merge the data structures for level $k+1$.
This is $O(\log n)$ BST operations and takes $O(\log^2 n)$ time.

By highest level invariant, the merge cannot create duplicate nodes in element BSTs or set BSTs.
Merging incident BSTs at level $k+1$ may create duplicate sets in $\Tin(\bF, k+1)$. We allow duplicate in incident BSTs.
We claim that the size of each incident BST is at most $\tO(nf)$, so that each BST operation takes $O(\log n)$ time. Note that duplicate sets cannot be created within a background solution, in particular the highest one $\bB_{\lmax}$. So, whenever $T_{\lmax}$ terminates, the switch step replaces all foreground incident BSTs with no duplicate sets. Since each element is incident to at most $f$ sets, this is at most $nf$ sets. Between two termination of $T_{\lmax}$, each step can process $\tO(1)$ elements, which adds $\tO(f)$ sets to all incident BSTs. By \Cref{lem:f-full-lifetime-bound}, the lifetime of $T_{\lmax}$ is $O(n)$. Since $T_{\lmax}$ cannot be aborted, we must witness its termination every $O(n)$ time-steps, and the total size of each incident BSTs cannot exceed $\tO(nf)$.
\end{proof}

The above implementation details and running time analysis work for the goal of bounding worst case insertion recourse.
To bound the deletion recourse as well, we need to deamortize the deletion recourse according to \Cref{lem:new:recourse-reduction}. The garbage collection step can be efficiently implemented as follows. When we switch out some set BST from the foreground solution, we move that to a BST representing all garbage sets. The algorithm performs an additional garbage collection step that removes $O(\log n)$ sets in the garbage set from the output.

\begin{lemma}\label{lem:f-time-final}
    The deamortized algorithm has worst-case update time $O(f\log^3 n)$.
\end{lemma}

We have concluded the proof of \Cref{thm:f-main} by combining \Cref{lem:f-feasible,lem:f-recourse,lem:f-approx-final,lem:f-time-final}.

{\small
\bibliography{references}
}

\end{document}